\documentclass[english,11pt,a4paper,dvipsnames]{scrartcl}
\addtokomafont{disposition}{\rmfamily}
\usepackage{amsmath, amssymb, amsthm, mathrsfs}
\usepackage{dsfont}
\usepackage{eurosym}
\usepackage{setspace, natbib}
\usepackage{bibunits}
\usepackage{bm}
\usepackage{threeparttable}
\usepackage{graphicx}
\usepackage{enumerate}
\usepackage{bibentry}

\usepackage{booktabs}
\usepackage{tabularx}		
\usepackage{multirow}		
\usepackage{commath}			
\usepackage{babel}
\usepackage{tikz}         
\usepackage[titletoc,title]{appendix}
\usepackage{pifont}
\usepackage[left=1in, right=1in, top=1in, bottom=1.5in]{geometry}
\usepackage[
breaklinks,
colorlinks=true,
pagebackref=false,
pdfauthor={},%
pdftitle={},%
citecolor=blue,
linkcolor=blue,
urlcolor=blue,
filecolor=blue
]{hyperref}   
\usepackage{cleveref}
\newcommand{\cmark}{\ding{51}}%
\newcommand{\xmark}{\ding{55}}%

\defaultbibliography{thebib.bib}
\defaultbibliographystyle{jaestyle2}

\allowdisplaybreaks

\newtheorem{thm}{Theorem}[section]

\newtheorem{assumption}[thm]{Assumption}

\newtheorem{prop}[thm]{Proposition}

\newtheoremstyle{remark2}{1ex}{1ex}%
      {}
      {}
      {\bf}
      {.}
      {5pt}
      {\thmname{#1}\thmnumber{ #2}\thmnote{ \slshape{(#3)}}} 
\theoremstyle{remark2} 
\newtheorem{rem}[thm]{Remark}
\newtheorem{defn}[thm]{Definition}

\newtheorem{example}[thm]{Example}

\makeatletter

\renewenvironment{proof}[1][\textbfseries\proofname]{\par
   \pushQED{\qed}%
   \normalfont \topsep6\p@\@plus6\p@\relax
   \trivlist
   \item[\hskip\labelsep
     #1\@addpunct{:}]\ignorespaces
}{%
   \popQED\endtrivlist\@endpefalse
}

\usepackage{tikz,pgfplots}
\usetikzlibrary{arrows,shapes,positioning}
\makeatletter
\tikzset{
  edge node/.code={%
      \expandafter\def\expandafter\tikz@tonodes\expandafter{\tikz@tonodes #1}}}
\makeatother
\tikzset{
 iff/.style={implies-implies,double equal sign distance, thick, draw},
 imp/.style={-implies,double equal sign distance, thick, draw},
 mutex/.style={
   <->, thick, draw,
   edge node={node [fill=white, inner sep=1pt, sloped, allow upside down, auto=false]{$\boldsymbol\oplus$}}},
}

\tikzset{
}

\newcommand{\diff}{\mathrm{d}}

\newcommand{\va}{\bm a}
\newcommand{\vd}{\bm d}
\newcommand{\vg}{\bm g}

\newcommand{\vx}{\bm x}
\newcommand{\vy}{\bm y}
\newcommand{\vz}{\bm z}
\newcommand{\vr}{\bm r}
\newcommand{\vh}{\bm h}
\newcommand{\vt}{\bm t}

\newcommand{\mX}{\bm X}
\newcommand{\mU}{\bm U}

\newcommand{\mS}{\bm S}
\newcommand{\mT}{\bm T}
\newcommand{\mV}{\bm V}

\newcommand{\mI}{\bm I}
\newcommand{\vzero}{\bm 0}

\newcommand{\mOmega}{\bm \varOmega}

\newcommand{\VaR}{\operatorname{VaR}}
\newcommand{\CoVaR}{\operatorname{CoVaR}}
\newcommand{\CoES}{\operatorname{CoES}}
\newcommand{\MES}{\operatorname{MES}}
\newcommand{\ES}{\operatorname{ES}}
\newcommand{\OS}{\operatorname{OS}}
\newcommand{\SR}{\operatorname{SR}}

\newcommand{\E}{\operatorname{E}}
\newcommand{\Var}{\operatorname{Var}}
\newcommand{\p}{\operatorname{P}}

\newcommand{\D}{\,\mathrm{d}}

\DeclareMathOperator*{\interior}{int}
\DeclareMathOperator*{\conv}{conv}

\newcommand{\F}{\mathcal{F}}
\newcommand{\R}{\mathbb{R}}
\newcommand{\A}{\mathsf{A}}
\renewcommand{\O}{\mathsf{O}}
\newcommand{\eps}{\varepsilon}
\renewcommand{\a}{\alpha}
\renewcommand{\b}{\beta}
\renewcommand{\P}{\mathcal P}
\newcommand{\one}{\mathds{1}}

\renewcommand{\rm}{\normalfont \rmfamily}
\renewcommand{\bf}{\normalfont \bfseries}

\newcommand{\lorder}{\preceq}
\newcommand{\slorder}{\prec}

\newcommand{\lex}{\lorder_{\mathrm{lex}}}
\newcommand{\gex}{\succeq_{\mathrm{lex}}}
\newcommand{\slex}{\slorder_{\mathrm{lex}}}

\def\be{\begin{equation} \label}
\def\ee{\end{equation}}

\numberwithin{equation}{section} 

\makeatother

\usepackage[colorinlistoftodos,textsize=tiny]{todonotes}
\newcommand{\Comments}{1}
\newcommand{\mynote}[2]{\ifnum\Comments=1\textcolor{#1}{#2}\fi}
\newcommand{\mytodo}[2]{\ifnum\Comments=1%
  \todo[linecolor=#1!80!black,backgroundcolor=#1,bordercolor=#1!80!black]{#2}\fi}

\ifnum\Comments=1               
\fi

\ifnum\Comments=1               
\fi

\begin{document}

%

\title{Backtesting Systemic Risk Forecasts using Multi-Objective Elicitability\footnote{We are grateful to Immanuel Bomze for suggesting to consider multivariate scoring functions equipped with general orders. Furthermore, we would like to thank Timo Dimitriadis, R\"udiger Frey, Christoph Hanck, Jana Hlavinov\'a, Kurt Hornik, Marc-Oliver Pohle, Birgit Rudloff and Johanna F. Ziegel for detailed comments and valuable discussions. Of course, all errors and and opinions expressed in this article are solely the authors' responsibility. Yannick Hoga gratefully acknowledges support of the German Research Foundation (DFG) through grant HO 6305/2-1.}}

\author{Tobias Fissler\thanks{Vienna University of Economics and Business (WU), Department of Finance, Accounting and Statistics, Welthandelsplatz 1, 1020 Vienna, Austria, 
		e-mail:  \href{mailto:tobias.fissler@wu.ac.at}{tobias.fissler@wu.ac.at}.} \and Yannick Hoga\thanks{University of Duisburg-Essen, Faculty of Economics and Business Administration, Universit\"atsstra\ss e 12, D-45117 Essen, Germany, tel.~+49\,201\,1834365, 
		e-mail: \href{mailto:yannick.hoga@vwl.uni-due.de}{yannick.hoga@vwl.uni-due.de}.}
}
\date{\today}
\maketitle

\begin{abstract}
\noindent
Systemic risk measures such as CoVaR, CoES and MES are widely-used in finance, macroeconomics and by regulatory bodies. Despite their importance, we show that they fail to be elicitable and identifiable. This renders forecast comparison and validation, commonly summarised as `backtesting', impossible.
The novel notion of \emph{multi-objective elicitability} solves this problem. Specifically, we propose Diebold--Mariano type tests utilising two-dimensional scores equipped with the lexicographic order.
We illustrate the test decisions by an easy-to-apply traffic-light approach. We apply our traffic-light approach to DAX~30 and S\&P~500 returns, and infer some recommendations for regulators.\\



\end{abstract}

\begin{bibunit}

\section{Motivation}\label{Motivation}

Regulating financial institutions in isolation is often not sufficient to prevent financial crises due to the interdependent risks these institutions face. In particular, their losses commonly exhibit a pronounced comonotonic behaviour in the extreme tails: When one financial institution, or the market as a whole, is in distress, other institutions are much more prone to being at risk as well.
The U.S.~subprime mortgage crisis of 2008--2009, the European sovereign debt crisis of 2010--2011 and 
the Covid-19 crash of 2020 have forcefully demonstrated this fact and also the need to assess the \emph{systemic} nature of risk. As a consequence of these crises, a huge literature on measuring systemic risk has emerged over the last decade \citep{GK11,ChenIyengarMoallemi2013,AB16, Aea17, BE17, FeinsteinRudloffWeber2017}. 

Systemic risk measures are important in various contexts. First, they are important in banking regulation under the Basel framework of the \citet{BCBSBF19}, where they are vital in determining which banks are among the globally systemically important banks (G-SIBs). Such G-SIBs are then subjected to higher capital requirements. Second, in finance, systemic risk measures may be used to study spillover effects in the financial system \citep{AB16} or the build-up of asset price bubbles \citep{BRS20}. Third, they may be used to study the linkage between the financial sector and the real economy. Among others, \citet{GKP16} and \citet{BE17} show that an increase in systemic risk is predictive of future declines in real economic activity. All these examples underscore the importance of accurately measuring and predicting systemic risk.

In this paper, we revisit three influential systemic risk measures. First, we consider \citeauthor{AB16}'s \citeyearpar{AB16} conditional value-at-risk (CoVaR) and conditional expected shortfall (CoES) as extensions of the well-known value-at-risk (VaR) and expected shortfall (ES) to the realm of systemic risk. 
If $Y$ are the losses of interest and $X$ the losses of a reference position, $\CoVaR_{\a|\b}(Y|X)$ ($\CoES_{\a|\b}(Y|X)$) is the VaR (ES) of $Y$ at level $\a$, given that $X$ is ``in distress''. Here, we interpret the event that $X$ is in distress as $X$ being larger or equal than its $\b$-quantile, i.e., $\{X\ge \VaR_\b(X)\}$. 
Finally, we consider \citeauthor{Aea17}'s \citeyearpar{Aea17} marginal expected shortfall, $\MES_{\b}(Y|X)$, as the conditional mean of $Y$ given $\{X\ge \VaR_\b(X)\}$. 
Section~\ref{CoVaR, CoES and MES} introduces the exact definitions.

\citet{Bea17} distinguish between the ``source-specific approach'' and the ``global approach'' to systemic risk measurement. The source-specific approach considers individual sources of systemic risk, such as contagion risk or liquidity crises. In contrast, global measures of systemic risk potentially incorporate all mechanisms studied in the source-specific approach. \citet{Bea17} categorize CoVaR, CoES and MES under the global approach.

In practice, forecasting systemic risk measures---such as CoVaR, CoES and MES---requires adequate models for the marginals $X$ and $Y$, and for their dependence structure.
The literature has developed numerous different modelling approaches for this; see \citet{GT13} and \citet{BC19} for forecasting models for CoVaR and CoES, and \citet{BE17} and \citet{Eck18} for MES models.
Due to the importance of systemic risk measures outlined above, it is vital to develop statistical quality assessments of the various models' predictive performances. It is the main aim of this paper to provide such tools, which are referred to as `backtests' in finance. 



Backtests have two main goals. On the one hand, one may wish to assess the \emph{absolute quality} of forecasting models, also called the \emph{calibration}, akin to model validation in statistics. Following the terminology of \cite{FZG16}, we call such procedures ``traditional backtests''. Roughly speaking, they check how well a sequence of risk measure forecasts aligns with corresponding observations of losses.
Traditional backtests rely on the \textit{identifiability} of the underlying risk measure, which ensures the existence of a (possible \textit{multivariate}) function $\mV$ that uniquely ``identifies'' the true report (see Definition \ref{defn:id}).
On the other hand, the presence of several alternative prediction models for a risk measure necessitates ``comparative backtests'' \citep{FZG16} to assess their predictive accuracy relative to each other. This is akin to statistical model selection procedures.
Comparative backtests exploit the \textit{elicitability} of the underlying risk measure. This implies the existence of a \textit{real-valued} loss (or also: scoring) function $S$, which is minimised in expectation by the optimal forecast (see Definition~\ref{defn:univ score}).


Our contributions in this paper are the following:
First, we show that $\CoVaR_{\a|\b}$, $\CoES_{\a|\b}$ and $\MES_\b$ are not identifiable and elicitable as standalone risk measures (Proposition~\ref{cor:negative result}). The practical implication of this is that neither traditional nor comparative backtests can be carried out. In particular, any regulation based solely on these systemic risk measures is pointless, because neither the adequacy of the forecasts can be determined nor can different systemic risk forecasts be sensibly compared (say to a regulatory standard model).

We provide a partial remedy for this drawback by giving joint identification functions for   
$(\VaR_\b(X), \CoVaR_{\a|\b}(Y|X))$, 
$(\VaR_\b(X), \CoVaR_{\a|\b}(Y|X), \CoES_{\a|\b}(Y|X))$, and $(\VaR_\b(X), \MES_\b(Y|X))$ (Theorem~\ref{thm:joint id}). These identification functions can be used for (conditional) calibration tests in the spirit of \cite{NZ17}.
To the best of our knowledge, this entails the first traditional backtest for these systemic risk measures apart from \cite{Banulescu-RaduETAL2019}. We contrast our approach with theirs in detail in Remark~\ref{rem:comparison id func}. 
In particular, they use one-dimensional identification functions for $(\VaR_\b(X), \CoVaR_{\a|\b}(Y|X))$ and $(\VaR_\b(X), \MES_\b(Y|X))$, which fail to be strict in contrast to our two-dimensional identification functions. For the backtest of \cite{Banulescu-RaduETAL2019}, this non-strictness leads to a complete loss of power in identifying certain misspecified systemic risk forecasts (Section \ref{app:Remark 4.5} in the Supplement).
Theoretically, our results are akin to the fact that $\ES_\a(Y)$ is not identifiable on its own, but the pair $(\VaR_\a(Y), \ES_\a(Y))$ is identifiable \citep{FZ16a}.

In stark contrast to the joint elicitability of the pair $(\VaR_\a, \ES_\a)$, however, we show that the pairs $(\VaR_\b,$ $\CoVaR_{\a|\b})$, $(\VaR_\b, \MES_\b)$ and the triplet $(\VaR_\b, \CoVaR_{\a|\b}, \CoES_{\a|\b})$ \emph{fail} to be elicitable (Section~\ref{app:negative results}). So while traditional backtests for the above pairs and the triplet may be constructed by virtue of their identifiability, classical comparative backtests exploiting elicitability are not feasible.

As a remedy to this negative result, we propose the novel concept of \emph{multi-objective elicitability}, which works with \textit{multivariate} scores $\mS$ mapping to $\R^m$ equipped with a certain (partial) order relation $\lorder$. This contrasts sharply with classical $\R$-valued losses $S$. Their prevalence to date is grounded in tradition \citep{Gne11} and the fact that $\R$ is equipped with the canonical (total) order relation $\le$, which allows for straightforward comparisons of losses.
Subsection~\ref{subsec:mo-scores} introduces \emph{multi-objective scores} $\mS$ and the corresponding concepts of \emph{multi-objective consistency and elicitability}.
The terminology stems from the field of \emph{multi-objective optimisation}: 
According to \citet[p.\ v]{Ehrgott2005} it is ``a mathematical theory of optimization under multiple objectives'', and can be encountered in various fields of science, economics, logistics and engineering.
Since this novel concept to forecast evaluation may open up the avenue to a whole field of applications and research (which is underpinned by further instances; see Example \ref{app:examples}),
we give a concise general outline of the theory, using partial orders on $\R^m$ or even infinite-dimensional real vector spaces.

For the systemic risk forecasts we consider here, scores mapping to $\R^2$ equipped with the lexicographic (total) order---described in 
Subsection~\ref{subsec:lexico}---are sufficient. In particular, the performance of different systemic risk forecasts must be ranked with regard to the lexicographic order.
Specifically, Theorem~\ref{thm:mo el} shows that $(\VaR_\b, \CoVaR_{\a|\b})$,
$(\VaR_\b, \CoVaR_{\a|\b}, \CoES_{\a|\b})$ and $(\VaR_\b, \MES_\b)$ are multi-objective elicitable, and it provides classes of strictly multi-objective consistent scores.
We outline in Section~\ref{sec:tests} how these scores can be used for comparative backtests of Diebold--Mariano type. These comparative backtests are different from---in our case infeasible---``standard'' comparative backtests in that they build on the newly introduced notion of multi-objective elicitability (with scores mapping to $\R^2$ equipped with the lexicographic order) instead of the ``standard'' notion of elicitability (with scores mapping to $\R$ equipped with the canonical order $\leq$). In particular, some systemic risk forecast is now preferable to some other forecast when the $\R^2$-valued score of the former is smaller (with regard to the lexicographic order) than that of the latter. Thus, financial institutions may build on this result to improve their prediction models for $(\VaR_\b, \CoVaR_{\a|\b})$, $(\VaR_\b, \CoVaR_{\a|\b}, \CoES_{\a|\b})$ and $(\VaR_\b, \MES_\b)$, which is crucial for an adequate assessment of the diverse risks faced by these institutions.

Since the multi-objective scores of Theorem~\ref{thm:mo el} take values in $\R^2$ equipped with the lexicographic order, some particularities arise for statistical hypothesis tests. While simple ``two-sided'' null hypotheses of equal predictive performance can be tested with a classical Wald-test, particular caution must be taken when testing for superior predictive ability. 
Due to the particularities of the lexicographic order, a straightforward ``one-sided'' composite null hypothesis would be insensitive to the systemic risk measure forecast, ignoring the primary goal of the backtesting procedure.
Therefore, we suggest to use ``one and a half''-sided composite null hypotheses, testing for superior predictive ability in the systemic risk component and equal performance in the auxiliary $\VaR_\b(X)$ component. Section~\ref{sec:tests} provides details, including an adaptation of the Basel framework's traffic-light approach to systemic risk backtests.

An empirical application in Section~\ref{Empirical Application} demonstrates the viability of the comparative backtest. There, we consider daily log-losses of the DAX~30 with daily log-losses of the S\&P~500 as a reference quantity. We compare systemic risk forecasts derived from a benchmark Gaussian copula model with those produced by a $t$-copula model, where in both models the correlation parameter of the copula is driven by generalised autoregressive score (GAS) dynamics \citep{CKL13}.
We find that the predictive performance of the $t$-copula is superior with $p$-values close to $3\%$, which is consistent with its popularity in empirical work. One conclusion from our empirical analysis is that fairly long samples are required to validly distinguish between different forecasts, because the effective sample sizes in comparing systemic risk forecasts are (almost by definition) reduced. Thus, the one year evaluation period for (univariate) VaR and ES forecasts in the Basel framework of the \citet{BCBSBF19} is, in our view, insufficient for systemic risk forecasts.

The paper closes with a discussion and outlook (Section~\ref{Summary and Outlook}).
Besides the parts already mentioned above, 
the Supplement provides proofs for the results of Section~\ref{Scoring functions, identification functions and (conditional) elicitability} (Section~\ref{sec:Proofs Section 3}) and  further background material on multi-objective elicitability (Section \ref{app:mo-el}). 
All other proofs are relegated to Section \ref{app:Proofs}. Section~\ref{app:Monte Carlo Simulations} investigates the finite-sample properties of our comparative backtests in simulations.
The \texttt{R} code to reproduce all numerical experiments is available online.

Throughout the paper, we indicate vectors with bold letters. We highlight the distinction between row and column vectors only when it is essential, and use the symbol $'$ to indicate the transpose of a vector or matrix.

\section{Formal definition of CoVaR, CoES and MES}
\label{CoVaR, CoES and MES}

Fix some non-atomic probability space $(\Omega, \mathfrak A, \p)$ where all random objects are defined.
Using standard notation, let $L^0(\R^d)$, $d=1,2$, be the space of all $\R^d$-valued random vectors on $(\Omega, \mathfrak A, \p)$. 
Furthermore, for $p\in[1,\infty)$, let $L^p(\R^d)\subseteq L^0(\R^d)$ be the collection of random vectors whose components possess a finite $p$th moment.
For $\mX\in L^0(\R^d)$ let $F_{\mX}$ be its joint distribution. Then define for $p\in\{0\}\cup[1, \infty)$ the collection $\F^p(\R^d):= \{F_{\mX}\colon \mX \in L^p(\R^d)\}$.
We overload notation and identify any $F\in \F^0(\R^d)$ with its cumulative distribution function (cdf) $\R^d \to [0,1]$.

Our systemic risk measures of interest---$\CoVaR$, $\CoES$ and $\MES$---are maps from $L^0(\R^2)$ (or $L^1(\R^2)$ for MES) to $ \R^* := (-\infty,\infty]$. They are law-determined, meaning that their values for $(X,Y)$ and $(\tilde X, \tilde Y)$ coincide if $F_{X,Y} = F_{\tilde X, \tilde Y}$. Hence, we can consider them as risk-functionals on $\F^0(\R^2)$ (or $\F^1(\R^2)$). 
Similarly, the popular univariate risk measure $\VaR_\b$, $\b\in[0,1]$, is a law-determined map $L^0(\R)\to[-\infty,\infty]$.
In the rest of the paper, we will frequently overload notation and identify these law-determined risk measures with their induced risk functionals. As such, we will use the terms `risk measure' and `risk functional' interchangeably.

Let $(X,Y)\in L^0(\R^2)$ be a two-dimensional random vector. Here, $Y$ stands for the losses of a position of interest (with the sign convention that positive values are losses and negative values are gains) and $X$ is a univariate reference position or aggregate of a reference system, having the same sign convention.
Denote by $F_{X,Y}$ their joint distribution function and by $F_X$ and $F_Y$ their marginals, respectively.
Recall that for $\b\in[0,1]$ the $\b$-quantile of $F_X$ is the closed interval $q_\b(F_X) = \{x\in\R\colon F_X(x-)\le \b \le F_X(x)\}$, where $F_X(x-) := \lim_{t\uparrow x}F_X(x-)$. Then, $\VaR_\b(X)$ is the lower $\b$-quantile of $F_X$, i.e.,
$\VaR_\b(X):= \VaR_\b(F_X):= \inf q_\b(F_X)$. 
For $\b\in(0,1)$, $\VaR_\b$ is always finite.
Our sign convention is such that the larger the risk measure of a position, the riskier it is deemed. 
Hence, we typically choose a probability level of $\b$ close to 1 for $\VaR_\b$, such as $\b=0.95$ or $\b=0.99$.

\citet{AB16} define $\CoVaR_\b$ as the $\b$-quantile of the conditional distribution function $F_Y(\ \cdot\mid X= \VaR_{\b}(X)) = \p\{Y\leq\cdot\mid X = \VaR_{\b}(X)\}$.
The conditioning event in this definition is problematic for several reasons. First, it may have probability zero (which is the case when $F_X$ is continuous). Second, it does not fully capture the tail-risk of $F_X$. Third, since the roles of $Y$ and $X$ are asymmetric by construction, one may want to consider different probability levels to specify the event of `being in distress'. Thus, we follow \citet{GT13} and \citet{Banulescu-RaduETAL2019} in redefining $\CoVaR_{\a|\b}\colon L^0(\R^2)\to\R$ for $\a\in(0,1)$, $\b\in[0,1)$ as
\begin{equation}\label{eq:CoVaR}
	\CoVaR_{\a|\b}(Y| X):=\CoVaR_{\a|\b}(F_{X,Y}):=\VaR_\a(F_{Y\mid X \geq \VaR_\b(X)}),
\end{equation}
where $F_{Y\mid X \geq \VaR_\b(X)}=\p\{Y\leq\cdot\mid X \geq \VaR_{\b}(X)\}$.
For $\beta=0$, we simply have $\CoVaR_{\a|0}(Y|X) := \VaR_\a(Y)$, and if $\b=\a$ we simply write $\CoVaR_{\a}(Y| X) = \CoVaR_{\a|\a}(Y| X)$.

Since $\CoVaR_{\a|\b}(Y| X)$ is merely a quantile of the distribution $F_{Y\mid X \geq \VaR_\b(X)}$, it inherits the same defects as $\VaR_\a(Y)$. 
That is, it ignores tail risks beyond the quantile level $\a$ and it fails to be coherent, particularly defying the rationale of advantageous diversification effects \citep{Aea99,MS14}.
The well-known \textit{Expected Shortfall} at level $\a\in(0,1)$, $\ES_\a(Y):= \ES_\a(F_Y):=\frac{1}{1-\a}\int_\a^1 \VaR_\gamma(Y) \diff \gamma \in \R^*$, does not suffer from these defects. 
This motivates \citet{AB16} to introduce the conditional Expected Shortfall (CoES).
As for $\CoVaR_{\a|\b}$, we modify the conditioning event and formally introduce for $\a\in(0,1)$, $\b\in[0,1)$,
$\CoES_{\a|\b}\colon L^0(\R^2)\to \R^*$ via
\begin{align}\label{eq:CoES}
&\CoES_{\a|\b}(Y|X):= \CoES_{\a|\b}(F_{X,Y})
:= \frac{1}{1-\a}\int_\a^1 \CoVaR_{\gamma | \beta}(Y|X)\diff \gamma .
\end{align}
If $F_{X,Y}$ is continuous, then
$\CoES_{\a|\b}(Y|X) = \E[Y |  Y\ge \CoVaR_{\a|\b}(Y|X) ,\ X\ge \VaR_\b(X)]$.
Again, $\CoES_{\a|0}(Y|X) = \ES_\a(Y)$, and we write $\CoES_\a(Y|X) := \CoES_{\a|\a}(Y|X)$.

Finally, we consider the marginal Expected Shortfall (MES) of \citet{Aea17}, which measures the expectation of $Y$ when $X$ is in distress, i.e., when $X$ is in its right tail. 
Specifically, we introduce for $\b\in[0,1)$ the map $\MES_\b\colon L^1(\R^2)\to\R$,
\begin{equation}\label{eq:MES}
\MES_{\b}(Y|X) := \MES_{\b}(F_{X,Y}) := \CoES_{0|\b}(Y|X) 
= \int_0^1 \CoVaR_{\gamma|\beta}(Y|X) \diff \gamma .
\end{equation}
Again, $\MES_{0}(Y|X) = \E[Y]$, and $\MES_{\b}(Y|X) = \E[Y |  X\ge \VaR_\b(X)]$ if $F_{X,Y}$ is continuous. 
If $F_X$ is discontinuous, Remark ~\ref{rem:erw} proposes a novel correction term which generalises the three measures considered in this paper.

\section{(Conditional) identifiability and multi-objective elicitability}
\label{Scoring functions, identification functions and (conditional) elicitability}

We present the theory in this section in all generality to serve as a basis for future research on multi-objective scores. Therefore, we work with general functionals which do not necessarily have the interpretation of risk functionals. All proofs are in Section~\ref{sec:Proofs Section 3}.

\subsection{Notation, basic definitions and results}
\label{subsec:notation}

Adopting the decision-theoretic terminology of \cite{Gne11}, we denote by $\A$ an action domain. This is the space of plausible forecasts, which can be finite for categorial forecasts, $\R$ or $\R^k$ for point forecasts, or a set of distributions for probabilistic forecasts. Moreover, let $\O$ be an observation domain---a set where verifying observations materialise---with $\O = \R^d$ as a leading example. 
Denote by $\F^0(\O)$ the set of all probability distributions on $\O$. Let $\F\subseteq  \F'$ be subclasses of $\F^0(\O)$. We consider a general, possibly set-valued functional $\mT\colon\F'\to\P(\A)$ with $\P(\A)$ the power set of $\A$. Later on, $\mT$ will have the interpretation of a risk functional.
Note that $\F'$ is the class of distributions, where our functional $\mT$ is defined on, and $\F\subseteq  \F'$ is the subclass on which $\mT$ will be identifiable/elicitable (see Definitions~\ref{defn:univ score} and~\ref{defn:id}). For instance, for $\mT(F_{X,Y})=(\VaR_\b(F_X),$ $\CoVaR_{\a|\b}(F_{X,Y}) )$, we have $\F'=\F^0(\R^2)$ and $\F$ is given in Theorem~\ref{thm:joint id}~(i)/Theorem~\ref{thm:mo el}~(ii).

We adopt the \emph{selective} notion of forecasts discussed in \cite{FFHR2021} where one is content with correctly specifying a single element $\vt \in \mT(F) \subseteq \A$ as opposed to specifying the entire set $\mT(F)$. (If one is interested in \emph{exhaustive} forecasts---i.e., in forecasts for the whole set $\mT(F)$---one can change the action domain to $\P(\A)$.)
If $\mT$ attains singletons only, we identify the value of $\mT(F)$ with its unique element.
This identification allows us to treat point-valued functionals as set-valued functionals without loss of generality. We mention that the risk functionals to be considered in Section~\ref{Structural results for CoVaR, CoES and MES} are all point-valued.

A function $\vg\colon \A\times \O\to\R^{\mathcal I}$, where $\mathcal I$ is an index set, is called $\F$-\textit{integrable} if for all components $g_i$, $i\in\mathcal I$, it holds that $\int |g_i(\vr,\vy)|\,\diff F(\vy)<\infty$ for all $\vr\in \A$, $F\in\F$.
If $\vg$ is $\F$-integrable, we define the map $\bar \vg\colon \A \times \F \to\R^{\mathcal I}$, $\bar \vg(\vr, F) := \int g_i(\vr,\vy)\,\diff F(\vy)$ for $\vr\in \A$, $F\in\F$.
A similar convention and notation is used for maps $\va\colon\O\to\R^{\mathcal I}$.
We start with the classical definition of elicitability and consistent scoring functions, mapping to $\R$.

\begin{defn}\label{defn:univ score}
An $\F$-integrable map $S\colon \A\times \O\to\R$ is an \emph{$\F$-consistent scoring function} for $\mT$ if $\bar S(\vt,F)\le \bar S(\vr,F)$ for all $\vt\in \mT(F)$, all $\vr\in\A$ and for all $F\in\F$. 
It is a \emph{strictly} $\F$-consistent scoring function for $\mT$
if, additionally, $\bar S(\vt,F)= \bar S(\vr,F)$ implies that $\vr\in \mT(F)$.
$\mT$ is \emph{elicitable on $\F$} if there is a strictly $\F$-consistent scoring function for $\mT$.
\end{defn}

\begin{defn}\label{defn:id}
An $\F$-integrable map $\mV\colon \A\times \O\to\R^m$ is an \emph{$\F$-identification function} for $\mT$ if $\bar \mV(\vt,F)=\vzero$ for all $\vt\in \mT(F)$ and for all $F\in\F$.
It is a \emph{strict} $\F$-identification function for $\mT$ if, additionally, for all $F\in\F$ and for all $\vr\in \A$, $\bar \mV(\vr,F)=\vzero$ implies that $\vr\in \mT(F)$.
$\mT$ is \emph{identifiable on $\F$} if there is a strict $\F$-identification function for $\mT$.
\end{defn}

Suppose in a risk management context that $\mT$ corresponds to the $\VaR_{\a}$-functional and $y_t$ are the observed losses. Subject to mild conditions on $\F$, $\VaR_\a$ is elicitable, where the `pinball loss' $S(r,y) = (\one\{y>r\} - 1+\a)(r-y)$ is strictly $\F$-consistent. This allows to compare competing VaR forecasts $(r_{t,(1)})_{t=1, \ldots, n}$ and $(r_{t,(2)})_{t=1, \ldots, n}$ via their empirical average score differences $\overline{d}_n = \overline{S}_{1n}-\overline{S}_{2n} = \frac{1}{n} \sum_{t=1}^n S(r_{t,(1)},y_t) - S(r_{t,(2)},y_t)$. A negative (positive) sign of $\overline{d}_n$ indicates superiority (inferiority) of $(r_{t,(1)})_{t=1, \ldots, n}$ over $(r_{t,(2)})_{t=1, \ldots, n}$. Identifiability, on the other hand, opens the way to test for calibration by checking, e.g., if the test statistic $\frac{1}{n}\sum_{t=1}^n \mV(r_{t,(i)},y_t)$ is sufficiently close to $\vzero$ or not.
E.g., when $\mT$ corresponds to $\VaR_\a$ checking calibration amounts to checking if the empirical VaR-violation rate is roughly $1-\a$, which can be done in terms of $V(r,y) = \one\{y>r\} - (1-\a)$. These examples demonstrate the importance of elicitability and identifiability for comparing and evaluating (risk) forecasts in practice.

Under regularity conditions, the notions of elicitability and identifiability are equivalent for point-valued functionals mapping to $\R$ \citep{SteinwartPasinETAL2014}.
There are also important functionals which fail to be elicitable and identifiable, most prominently the variance and expected shortfall \citep{Gne11}. 
In such situations, the notions of \emph{conditional} elicitability and \emph{conditional} identifiability can be helpful.
Following the concept presented in \citet[Section 6]{FZ16a}, we slightly adapt \citeauthor{EKT15}'s \citeyearpar{EKT15} original definition of conditional elicitability. We also introduce the corresponding counterpart of conditional identifiability.

\begin{defn}\label{defn:cond el}\label{defn:cond id}
Consider two functionals $\mT_j:\F'\to \P(\A_j)$, $j=1,2$, and let $\F\subseteq \F'$. \begin{enumerate}[\rm (i)]
\item
$\mT_{2}$ is \textit{conditionally elicitable with} $\mT_{1}$ on $\F$, if $\mT_1$ is elicitable on $\F$ and $\mT_2$ is elicitable on 
\(
	\mathcal{F}_{\vr_1}:=\left\{F\in\mathcal{F}\colon\ \vr_1\in \mT_1(F)\right\} 
\)
for any $\vr_1\in \A_1$.
\item
$\mT_{2}$ is \textit{conditionally identifiable with} $\mT_{1}$ on $\F$, if $\mT_1$ is identifiable on $\F$ and $\mT_2$ is identifiable on 
\(
	\mathcal{F}_{\vr_1}:=\left\{F\in\mathcal{F}\colon\ \vr_1\in \mT_1(F)\right\} 
\)
for any $\vr_1\in \A_1$.
\end{enumerate}
\end{defn}

It is easy to see that the variance is conditionally elicitable and conditionally identifiable with the mean, and that $\ES_\a$ is conditionally elicitable and conditionally identifiable with $\VaR_\a$ on appropriate classes of distributions, respectively \citep{EKT15}.
The pairs (mean, variance) and $(\VaR_\a, \ES_\a)$ even turn out to be elicitable and identifiable \citep{FZ16a}. For identifiability, this is an instance of the following proposition, which is already stated without a proof in the discussion of \cite{FZ16a}.

\begin{prop}\label{prop:conditional id}
If $\mT_2$ is conditionally identifiable with $\mT_1$ on $\F$, then the pair $(\mT_1,\mT_2)$ is identifiable on $\F$.
\end{prop}

Of course, an analogue to Proposition \ref{prop:conditional id} for (conditional) elicitability would be desirable, and it has been stated as an open problem in the discussion of \cite{FZ16a}.
Unfortunately, the answer is negative: 
While Section \ref{app:cond id and el} establishes the conditional elicitability of $\CoVaR_{\a|\b}$, $(\CoVaR_{\a|\b}, \CoES_{\a|\b})$ and $\MES_\b$ all with $\VaR_\b$, the corresponding pairs and the triplet generally fail to be elicitable; see Section~\ref{app:negative results}.


\subsection{Multi-objective scores, consistency and elicitability}
\label{subsec:mo-scores}

To overcome the structural drawback of elicitability in comparison to identifiability, in particular the lack of an analogue to Proposition \ref{prop:conditional id}, we introduce the novel notion of \emph{multi-objective scoring functions} and the corresponding concepts of \emph{multi-objective consistency} and \emph{multi-objective elicitability}.
It is inspired by the fundamental observation that identification functions are generally multivariate.

To be more precise, the dimension $m$ of the identification function $\mV$ usually coincides with the dimension $k$ of the functional. 
If $k=m=1$ an identification function is often induced by the derivative of a consistent scoring function $S$. Also, the antiderivative of an (oriented) identification function yields a consistent score, thus roughly establishing a one-to-one correspondence between the class of identification functions and the class of consistent scoring functions.
For a $k$-dimensional functional, the gradient of a consistent score is $\R^k$-valued and naturally induces an identification function, subject to smoothness conditions.
However, not every $k$-dimensional identification function possesses an antiderivative for $k\ge2$. This is due to integrability conditions asserting that if it was integrable, the corresponding Hessian of the stipulated antiderivative would need to be symmetric (see also Example~\ref{footnote:double quantile} for an illustration).
This rules out the one-to-one relation in the multivariate setting, giving rise to a \emph{gap} between the class of consistent scoring functions and the one of identification functions.
This gap and its consequences on estimation are discussed in detail by \cite{DFZ2020}.

We sidestep this integrability constraint by introducing the concept of \textit{multivariate} scoring functions.
Indeed, it is this multivariate structure of identification functions which facilitates the straightforward proof of Proposition~\ref{prop:conditional id}. Therefore, we mimic this multi-dimensionality for scores, letting them map to $\R^m$, where usually $m=k$, or even more generally to some real vector space $\R^{\mathcal I}$, where the index set ${\mathcal I}$ may be finite, countable or even uncountable.
In the application of the general theory developed here to the systemic risk measures, it suffices to consider $\R^2$-valued scores; see Theorem~\ref{thm:mo el} in Section~\ref{Structural results for CoVaR, CoES and MES}.
The motivation for defining a classical score as a univariate map to $\R$ (see Definition~\ref{defn:univ score}) is grounded in tradition on the one hand. On the other hand, $\R$ is equipped with the canonical (total) order relation $\le$, which allows for straightforward comparisons of forecasts by checking whether the empirical average score differences satisfy $\overline{d}_n\geq0$ or $\overline{d}_n\leq0$.
Consistency in Definition~\ref{defn:univ score} ultimately relies on the existence of an order. 

Hence, we need to equip $\R^{\mathcal I}$ with a vector partial order $\lorder$ and write $(\R^{\mathcal I}, \lorder)$. 
A binary relation $\lorder$ on $\R^{\mathcal I}$ is a \emph{partial order} if it is 
reflexive ($\forall \vx \in \R^{\mathcal I}$,  $\vx\lorder \vx$),
antisymmetric ($\forall \vx,\vy\in \R^{\mathcal I}$ if $\vx\lorder \vy$ and $\vy\lorder \vx$, then $\vx=\vy$),
transitive ($\forall \vx,\vy, \vz \in \R^{\mathcal I}$ if $\vx\lorder \vy$ and $\vy\lorder \vz$, then $\vx\lorder \vz$).
Two elements $\vx,\vy\in \R^{\mathcal I}$ are \emph{comparable} if $\vx\lorder \vy$ or $\vy\lorder \vx$. A \emph{total} order is a partial order where all elements are comparable.
A partial order $\lorder$ on $\R^{\mathcal I}$ induces a \emph{strict} partial order $\slorder$ on $\R^{\mathcal I}$ as follows. For all $\vx,\vy\in \R^{\mathcal I}$ it holds that $\vx\slorder \vy$ if and only if $\vx\lorder \vy$ and $\vx\neq \vy$.
A partial order $\lorder$ on $\R^{\mathcal I}$ is a \emph{vector partial order} if it is compatible with addition and positive scaling. 
That is, if for all $\vx,\vy, \vz\in\R^{\mathcal I}$ and for all $\lambda \in (0,\infty)$ it holds that 
$\vx \lorder \vy$ implies that $\vx + \vz \lorder \vy + \vz$, and
$\vx \lorder \vy$ implies that $\lambda \vx \lorder \lambda \vy$.
(In the sequel, we always mean a vector partial order whenever we write ``partial order''.)
The canonical choice is the componentwise order defined for $\vx=(x_i)_{i \in \mathcal I},\vy = (y_i)_{i \in \mathcal I}\in\R^{\mathcal I}$ as $\vx \lorder \vy$ if and only if $x_i\le y_i$ for all $i\in \mathcal I$. 
The use of a partial order $\lorder$ on $\R^{\mathcal I}$ leads to two different generalisations of Definition \ref{defn:univ score}. These are due to the fact that in a partial order $\vx \lorder \vy$ implies that $\vy \nprec \vx$, but the reverse implication fails if $\vx$ and $\vy$ are not comparable.

\begin{defn}\label{defn:multiv score partial}
\begin{enumerate}[(i)]
\item
An $\F$-integrable map $\mS\colon \A\times \O\to(\R^{\mathcal I}, \lorder)$ is a \emph{strongly} multi-objective $\F$-consistent scoring function for $\mT\colon \F\to\P(\A)$ if 
$\bar \mS(\vt,F)\lorder \bar \mS(\vr,F)$
for all $\vt\in \mT(F)$, all $\vr\in\A$ and all $F\in\F$.
It is a strictly strongly multi-objective $\F$-consistent scoring function for $\mT$ if, additionally, 
$\bar \mS(\vt,F)= \bar \mS(\vr,F)$ implies that $\vr\in \mT(F)$.
$\mT$ is \emph{strongly multi-objective-elicitable on $\F$} (with respect to the order $\lorder$ on $\R^{\mathcal I}$) if there is a strictly strongly multi-objective $\F$-consistent scoring function for $\mT$.
\item
An $\F$-integrable map $\mS\colon \A\times \O \to(\R^{\mathcal I}, \lorder)$ is a \emph{weakly} multi-objective $\F$-consistent scoring function for $\mT\colon \F\to\P(\A)$  if 
$\bar \mS(\vr,F)\nprec \bar \mS(\vt,F)$
for all $\vt\in \mT(F)$, all $\vr\in\A$ and all $F\in\F$.
It is a strictly weakly
multi-objective $\F$-consistent scoring function for $\mT$  if, additionally, 
for all $F\in\F$ and for any $\vr_0\in\A$ it holds that if 
$\bar \mS(\vr,F)\nprec \bar \mS(\vr_0,F)$
for all $\vr \in\A$, then  $\vr_0\in \mT(F)$.
$\mT$ is \emph{weakly multi-objective-elicitable on $\F$} (with respect to the order $\lorder$ on $\R^{\mathcal I}$) if there is a strictly weakly multi-objective $\F$-consistent scoring function for $\mT$.
\end{enumerate}
\end{defn}

As discussed, a partial order on $\R^{\mathcal I}$ is sufficient to define multi-objective consistency. 
Using the notation $\bar \mS(B,F):= \{\bar \mS(\vr,F)\colon \vr\in B\}$ for $F\in\F$ and $B\subseteq \A$, the definition of multi-objective consistency
 does not require comparability of all elements $\bar \mS(\A,F)$ for some fixed $F\in\F$. Weak multi-objective consistency solely implies that 
all elements of  $\bar \mS(\mT(F),F)$ are minimal in $\bar \mS(\A,F)$. 
The strict version additionally ensures that all minimal elements of $\bar \mS(\A,F)$ are in $\bar \mS(\mT(F),F)$.
Strong multi-objective consistency additionally means that the elements of $\bar \mS(\mT(F),F)$ are not only minimal in $\bar \mS(\A,F)$, but that they are the (unique) least element of $\bar \mS(\A,F)$. Due to the uniqueness of a least element, $\bar \mS(\mT(F),F)$ is a singleton.
The strict version of strong multi-objective consistency additionally means that $\bar \mS(\mT(F),F)\not\subset\bar \mS(\A\setminus \mT(F),F)$.
See Figure \ref{fig:StrongWeak} for an illustration of the two situations.
Clearly, (strict) strong consistency implies (strict) weak consistency. 
If we omit the qualifiers ``weak'' or ``strong'', we refer to the strong version.
\begin{figure}[t!]
	\centering
		\includegraphics[width=\textwidth]{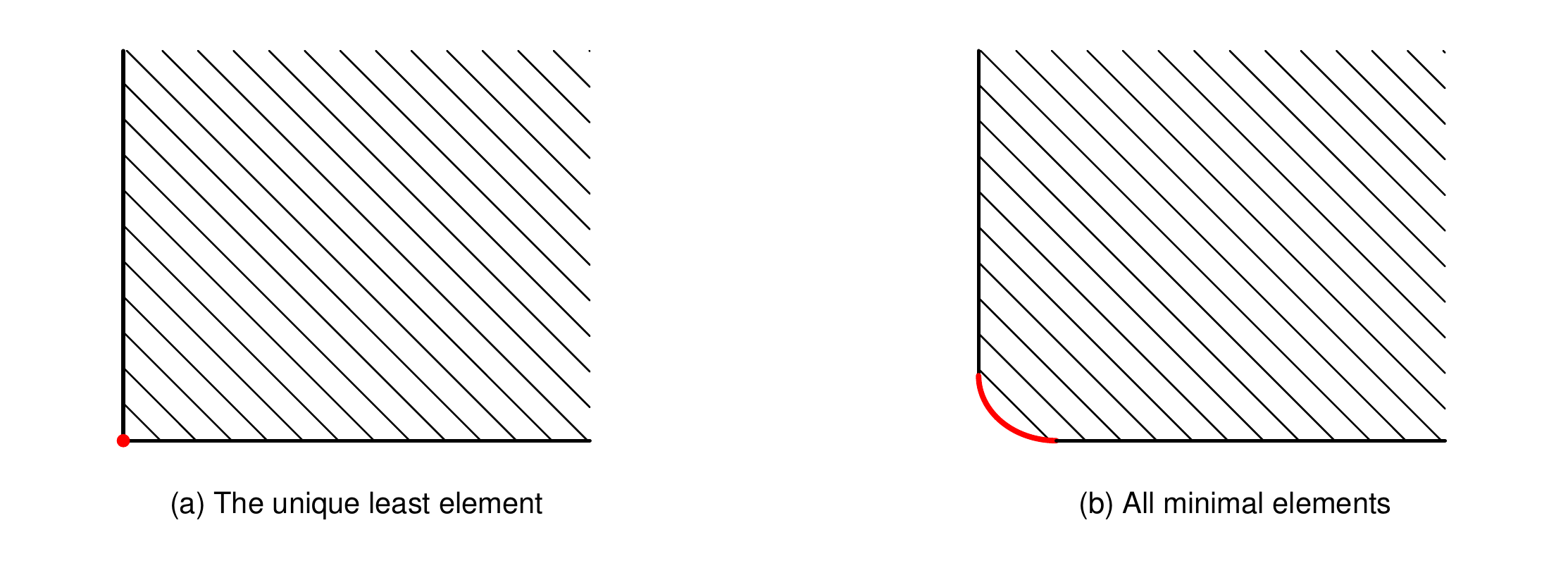}
	\caption{In both panels, the shaded areas correspond to $\bar \mS(\A,F)$ of a multi-objective score $\mS$ mapping to $\R^2$ equipped with the componentwise order. The red sets correspond to $\bar \mS(\mT(F),F)$. In panel (a), this is the unique least element, illustrating the situation of strict strong multi-objective consistency. In panel (b), this set corresponds to all minimal elements, depicting the situation of strict weak multi-objective elicitability.}
	\label{fig:StrongWeak}
\end{figure}

Multi-objective consistency translates the essence of consistency as a `truth serum', or being incentive compatible, to the multivariate realm. 
No matter what forecast $\vr\in\A$ an agent issues, they would be better off (or at least not worse off) in expectation when issuing a 
correctly specified functional value $\vt\in \mT(F)$. (In the weak version, if an agent issued some $\vt\in \mT(F)$, they would not be better off with some other forecast.)
Strictness means that any action $\vr\in \A\setminus \mT(F)$ leads to a worse outcome in expectation than the truth $\vt\in \mT(F)$. 
So this minimal requirement of honouring truthful forecasts is preserved.

\begin{rem}
To the best of our knowledge, the notion of multi-objective scoring functions with the related concepts is novel to the forecast evaluation literature. However, it has some connections to the concept of \emph{forecast dominance} introduced by \citet[Section 3.2]{EhmETAL2016}.  
Let $(S_i)_{i\in\mathcal I}$ be some class of univariate (strictly) $\F$-consistent scoring functions $S_i\colon \A\times \O\to\R$ for some functional $\mT\colon\F\to\P(\A)$.
Adapting their definition slightly, we say that a forecast $\vr\in\A$ dominates $\tilde \vr\in\A$ relative to $(S_i)_{i\in\mathcal I}$ if
$\bar S_i(\vr,F) \le \bar S_i(\tilde \vr,F)$ 
for all $F\in\F$ and for all $i\in\mathcal I$.
In our terminology, we could phrase this forecast ranking in terms of a multivariate score $\mS\colon \A\times \O \to (\R^{\mathcal I}, \lorder)$, $\mS(\vr,\vy) = \big(S_i(\vr,\vy)\big)_{i\in\mathcal I}$, where $\lorder$ is the usual componentwise partial order. 
By virtue of mixture representations, \cite{EhmETAL2016} impressively demonstrate that for quantiles and expectiles, one can rephrase forecast dominance with respect to (nearly) all consistent scoring functions equivalently in terms of extremal or elementary scores $(S_\theta)_{\theta\in \R}\subseteq (S_i)_{i\in\mathcal I}$. 
For this situation, we can easily construct a multi-objective score $\mS = (S_\theta)_{\theta\in \R}$, mapping to the function space $\R^\R$ equipped with the componentwise order. \cite{EhmETAL2016} provide numerous instances of comparable and incomparable forecast rankings.
This example provides an easy construction for a (strictly) multi-objective $\F$-consistent score, $\mS = (S_\theta)_{\theta\in \R}$, if $\mT$ is elicitable.
\end{rem}

For verifying observations $(\vy_t)_{t=1, \ldots, n}$, our Definition~\ref{defn:multiv score partial} suggests to compare two sequences of forecasts $(\vr_{t,(i)})_{t=1, \ldots, n}$ ($i=1,2$) via their average scores $\overline{\mS}_{1n}=\frac{1}{n} \sum_{t=1}^n \mS(\vr_{t,(1)},\vy_t)$ and $\overline{\mS}_{2n}=\frac{1}{n} \sum_{t=1}^n \mS( \vr_{t,(2)},\vy_t)$. However, while using classical univariate scores always leads to a conclusive forecast ranking (ignoring questions of statistical significance for a moment), the presence of a \emph{partial} order may lead to inconclusive rankings, particularly if neither of the two forecasts is correctly specified. For instance, if $\overline{\mS}_{1n}=(1, 4)$ and $\overline{\mS}_{2n}=(4, 1)$ (e.g., for the bivariate systemic risk scores of Theorem~\ref{thm:mo el}), then neither $\overline{\mS}_{1n}\lorder\overline{\mS}_{2n}$ nor $\overline{\mS}_{2n}\lorder\overline{\mS}_{1n}$ in the componentwise order on $\mathbb{R}^2$.

\subsection{Multi-objective elicitability with respect to the lexicographic order \\
and conditional elicitability}
\label{subsec:lexico}

To overcome the issue of inconclusive forecast rankings, we equip $\R^{\mathcal I}$ with a total order. 
In a total order, all elements are comparable. Hence, the notions of weak and strong multi-objective consistency / elicitability coincide and a distinction is redundant.
A promising choice is the lexicographic order.
For our purposes, and to ease the exposition, it suffices to consider $\R^2$ equipped with the lexicographic order $\lex$. On $\R^2$ it holds that $(x_1,x_2)\lex (y_1,y_2)$ if $x_1<y_1$ or if ($x_1=y_1$ and $x_2\le y_2$). 
The lexicographic order is widely used for preference relations in microeconomics. In particular, it is well-known for the fact that it cannot be represented by a real-valued utility function \cite[Chapter 3.C]{MWG1995}.

There are at least two reasons for our choice of the lexicographic order. First, as a total order, it allows for conclusive forecast rankings, which we exploit in our comparative backtests in Section~\ref{sec:tests}. For instance, if as above $\overline{\mS}_{1n}=(1, 4)$ and $\overline{\mS}_{2n}=(4, 1)$ for the systemic risk scores, then the inconclusive ranking is resolved because $\overline{\mS}_{1n}\lex \overline{\mS}_{2n}$. Second, the lexicographic order opens the way to the following analogue of Proposition~\ref{prop:conditional id}.

\begin{thm}\label{thm:conditional el}
If $\mT_2$ is conditionally elicitable with $\mT_1$ on $\F$ and $\mT_1(F)$ is a singleton for all $F\in\F$,
then the pair $(\mT_1,\mT_2)$ is multi-objective elicitable on $\F$ with respect to the lexicographic order $\lex$ on $\R^2$.
\end{thm}

Theorem~\ref{thm:conditional el} is almost a direct analogue of Proposition~\ref{prop:conditional id}, with the intriguing exception that $\mT_1$ is assumed to be a singleton on $\F$ in Theorem~\ref{thm:conditional el}; see Remark~\ref{rem:relaxation} for further details. We illustrate how the construction of Theorem~\ref{thm:conditional el} leads to strictly consistent multi-objective scores mapping to $(\R^2, \lex)$ for the cases of (mean, variance) and $(\VaR_\a, \ES_\a)$ in Subsection~\ref{app:examples}.
There, we also provide further examples of conditionally elicitable functionals failing to be elicitable in the traditional sense, but which are multi-objective elicitable with respect to $(\R^2, \lex)$. Section~\ref{Dimensionality considerations} elaborates on how the multi-objective scores can be considered a ``generalised antiderivative'' of a $k$-dimensional identification function with relaxed symmetry conditions. Section~\ref{More on multi-objective elicitability with respect to the lexicographic order} provides further details on comparing misspecified forecasts under multi-objective elicitability, and on the sensitivity of multi-objective scores with respect to increasing information sets. Finally, Section~\ref{subsec:CxLS} revisits a powerful necessary condition for identifiability and elicitability, namely the Convex Level Sets (CxLS) property, for multi-objective scores.

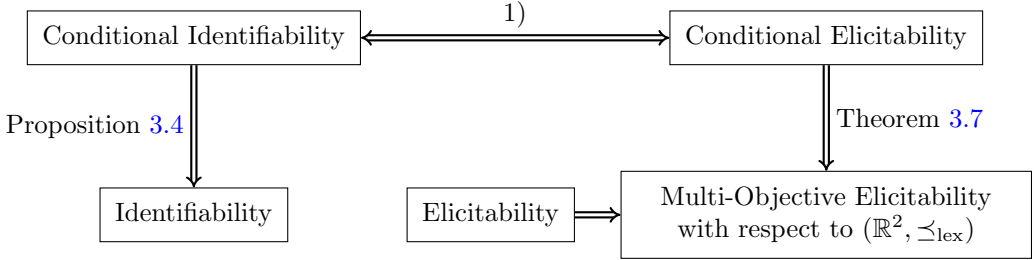
\begin{figure}
\centering
\small
\begin{tikzpicture}[every edge/.style={imp}]
  \matrix[nodes={inner sep=2mm},
  row sep=0.7cm,column sep=0.6cm] {
    \node (c-id) [rectangle, draw] {Conditional Identifiability}; &
																									&
    \node (c-el) [rectangle, draw] {Conditional Elicitability}; \\ \\
		\node (id)   [rectangle, draw] {Identifiability}; &
		\node (el)  [rectangle, draw] {Elicitability};
																									&
    \node (mo-el) [rectangle, draw, text width=5cm, align=center] {Multi-Objective Elicitability with respect to $(\mathbb{R}^2, \lex)$}; \\
  };
  \path
	(c-id) edge[implies-implies,double equal sign distance] node[above] {1)} (c-el)
	(c-el) edge[imp] node[right] {Theorem~\ref{thm:conditional el}} (mo-el)
	(c-id) edge[imp] node[left] {Proposition~\ref{prop:conditional id}} (id)
	(el) edge[imp] (mo-el)
  ;
\end{tikzpicture}
\caption{Illustration of the most important structural results for real-valued (risk) functionals $\mT_1$ and $\mT_2$. The equivalence 1) follows from \citet{SteinwartPasinETAL2014} under some regularity conditions.}
\label{fig:overview meth}
\end{figure}

The proof of Theorem \ref{thm:conditional el} explicitly exploits the asymmetric structure of the lexicographic order which fits well with the asymmetric notion of conditional elicitability, where the roles of $\mT_1$ and $\mT_2$ may not be changed. 
In particular, in the setup of Theorem \ref{thm:conditional el}, the pair $(\mT_2,\mT_1)$ is generally not multi-objective elicitable on $\F$ with respect to the lexicographic order $\lex$ on $\R^2$.
More generally, we suspect that under the conditions of Theorem \ref{thm:conditional el} the lexicographic order is the only order on $\R^2$ which renders $(\mT_2,\mT_1)$ multi-objective elicitable.

In the rest of the paper, we focus on multi-objective scores mapping to $(\R^2, \lex)$. We use the results of Section~\ref{Scoring functions, identification functions and (conditional) elicitability}, summarised in Figure~\ref{fig:overview meth}, extensively to prove our structural results for the systemic risk measures in the next section.


\section{Structural results for CoVaR, CoES and MES}
\label{Structural results for CoVaR, CoES and MES}

\subsection{CoVaR, CoES and MES fail to be identifiable or elicitable}

The following proposition shows that the three risk measures $\CoVaR_{\a|\b}$, $\CoES_{\a|\b}$ and $\MES_{\b}$ generally fail to be identifiable or elicitable on sufficiently rich classes of bivariate distributions $\F\subseteq \F^0(\R^2)$. It is proven by showing that the CxLS property (which remains necessary for multi-objective scores; see Proposition~\ref{prop:CxLS}) is violated in each case.

\begin{prop}\label{cor:negative result}
For $\a, \b\in(0,1)$, $\CoVaR_{\a|\b}$, $\CoES_{\a|\b}$ and $\MES_{\b}$ are neither identifiable nor elicitable on any class $\F\subseteq \F^0(\R^2)$ containing all bivariate normal distributions along with their mixtures.
\end{prop}

Proposition~\ref{cor:negative result} casts doubt on traditional and comparative backtesting approaches for CoVaR, CoES and MES as standalone systemic risk measures. Thus, these measures should not be used for regulatory purposes on their own, because forecasts for them can neither be verified for their adequacy nor can they be sensibly compared to improve their modeling.

Proposition~\ref{cor:negative result} can also be shown on other classes $\F\subseteq \F^0(\R^2)$ that are sufficiently rich, e.g., the classes containing all measures with finite support (see Remark~\ref{rem:CxLS1}). Importantly, $\F$ must be convex and must not only consist of distributions with independent marginals.

\subsection{Joint identifiability results}

Section~\ref{app:cond id and el} establishes the conditional identifiability and conditional elicitability of $\CoVaR_{\a|\b}(Y|X)$, $(\CoVaR_{\a|\b}(Y|X), \CoES_{\a|\b}(Y|X))$, and $\MES_{\b}(Y|X)$ with $\VaR_\b(X)$, respectively. 
This, in combination with Proposition \ref{prop:conditional id}, immediately yields the following joint identifiability results, where for $\a,\b\in(0,1)$ and $p\in\{0\}\cup[1,\infty]$ we use the notation
\begin{align}\nonumber
	\F_{\a}^p(\R)	&:=\big\{F\in \F^p(\R) \colon F\big(\VaR_\a(F)\big)=\a\big\},\\\label{eq:notation}
	\F_{(\a)}^p(\R)	&:=\big\{F\in  \F_{\a}^p(\R)\colon F\big(\VaR_\a(F) + \eps\big)>\a \ \text{for all } \eps>0\big\}, \\ \nonumber
	\F_{(\b)}^0(\R^2)&:=\big\{F_{X,Y}\in\F^0(\R^2)\colon F_X \in \F_{(\b)}^0(\R)\big\}.
\end{align}
\begin{thm}\label{thm:joint id}
Let $\a,\b\in(0,1)$.
Consider the strict $\F_{(\b)}^0(\R^2)$-identification function for $F_{X,Y}\mapsto \VaR_{\b}(F_X)$, $V^{\VaR}\colon \R\times\R^2\to\R$, $V^{\VaR}\big(v,(x,y)\big) = \one\{x\le v\} - \b$.
\begin{enumerate}[\rm (i)]
\item
On $\{F_{X,Y}\in \F^0(\R^2)\colon F_X\in \F^0_{(\b)}(\R), \ F_{Y|X\ge\VaR_\b(X)}\in \F^0_{(\a)}(\R)\}$, the pair $F_{X,Y}\mapsto  (\VaR_\b(F_X),$ $\CoVaR_{\a|\b}(F_{X,Y}) )$ is identifiable with a strict identification function $\mV^{(\VaR, \CoVaR)}\colon \R^2\times \R^2\to\R^2$,
\begin{align}\label{eq:V_CoVaR}
			\mV^{(\VaR, \CoVaR)}\big((v,c),(x,y)\big) 
			&=\begin{pmatrix}
			\one\{x\le v\} - \b\\ 
			\one\{x>v\} \big[\one\{y\le c\} - \a\big]
			\end{pmatrix}.
\end{align}
\item
On $\{F_{X,Y}\in \F^0(\R^2)\colon F_X\in \F^0_{(\b)}(\R), \ F_{Y|X\ge\VaR_\b(X)}\in \F^1_{(\a)}(\R)\}$, the triplet $F_{X,Y}\mapsto  (\VaR_\b(F_X), \CoVaR_{\a|\b}(F_{X,Y}), \CoES_{\a|\b}(F_{X,Y}) )$ is identifiable with a strict identification function $\mV^{(\VaR, \CoVaR, \CoES)}\colon \R^3\times \R^2\to\R^3$,
\begin{equation}\label{eq:V_CoES}
			\big((v,c,e),(x,y)\big) 
			\mapsto\begin{pmatrix}
			\one\{x\le v\} - \b\\ 
			\one\{x>v\} \big[\one\{y\le c\} - \a\big]\\
			\one\{x>v\} \Big[e - \frac{1}{1-\a}\big(y\one\{y>c\} + c(\one\{y\le c\} - \a)\big)			\Big]
			\end{pmatrix}.
\end{equation}
\item
On $\{F_{X,Y}\in \F^0(\R^2)\colon F_X\in \F^0_{(\b)}(\R), \ F_{Y|X\ge\VaR_\b(X)}\in \F^1(\R)\}$, the pair $F_{X,Y}\mapsto  (\VaR_\b(F_X),$ $\MES_{\b}(F_{X,Y}) )$ is identifiable with a strict identification function $\mV^{(\VaR, \MES)}\colon \R^2\times \R^2\to\R^2$,
\begin{align}\label{eq:V_MES}
			\mV^{(\VaR, \MES)}\big((v,\mu),(x,y)\big) 
			&=\begin{pmatrix}
			\one\{x\le v\} - \b\\ 
			\one\{x>v\} \big[\mu-y\big]
			\end{pmatrix}.
\end{align}
\end{enumerate}
\end{thm}

Following \cite{NZ17}, Theorem~\ref{thm:joint id} can readily be used to assess the absolute forecast quality via joint (\emph{Wald}-type) calibration tests. These either test the null hypothesis of \textit{unconditional} calibration, $\E[\mV(\vr_t, (X_t,Y_t))]=\vzero$ for all $t\in \mathbb N$, or the more informative null of \emph{conditional} calibration, $\E[\mV(\vr_t, (X_t,Y_t))\,|\,\mathfrak F_{t-1}]=\vzero$ for all $t\in \mathbb N$.
Here, the $\sigma$-algebra $\mathfrak F_{t-1}$ represents the information available to the forecaster at time $t-1$.
Recall that the null of conditional calibration is equivalent to 
$\E[\varphi_{t-1}\mV(\vr_t, (X_t,Y_t))]=\vzero$ for all $\mathfrak F_{t-1}$-measurable random vectors $\varphi_{t-1}$ for all $t\in \mathbb N$.
Unless $\mathfrak F_{t-1}$ is particularly simple, this approach is statistically not feasible, because a finite selection $\varphi_{1,t-1}, \ldots, \varphi_{\ell, t-1}$ of $\mathfrak F_{t-1}$-measurable random vectors is required in practice. Note that using only a finite collection amounts to testing a broader null than the original one of conditional calibration.

\begin{rem}\label{rem:comparison id func}
It is worth pointing out the improvements upon the substantial work of \cite{Banulescu-RaduETAL2019}, who develop---inter alia---traditional backtests of systemic risk measures. Translating their notation into ours, they propose in their equation (4) a one-dimensional identification function for the pair $(\VaR_\b(X), \CoVaR_{\a|\b}(Y|X))$ of the form $V\colon \R^2\times \R^2\to\R$,
\be{eq:Banulescu-Radu CoVaR id}
V\big((v,c),(x,y)\big) = \one\{x> v\} \one\{y> c\} - (1-\a)(1-\b).
\ee
\cite{Banulescu-RaduETAL2019} point out that this an identification function since
$
\bar V\big((\VaR_\b(X),$ $\CoVaR_{\a|\b}(X|Y)),F_{X,Y}\big) =0.
$
However, it fails to be strict, since, so long as $(1-\a')(1-\b') = (1-\a)(1-\b)$, we have $\bar V\big((\VaR_{\b'}(X), \CoVaR_{\a'|\b'}(X|Y)),F_{X,Y}\big)=0$. Thus, the specific null of `correct' unconditional calibration $\E[V(\vr_t, (X_t, Y_t))]=0$ with $V$ from \eqref{eq:Banulescu-Radu CoVaR id} is too broad, since it is satisfied for all $\vr_t=(\VaR_{\b'}(X),\CoVaR_{\a'|\b'}(X|Y))$ with $(1-\a')(1-\b') = (1-\a)(1-\b)$. This implies only trivial power against such alternatives for the corresponding Wald-test for calibration. This is in stark contrast to the Wald-test employing our two-dimensional identification function in \eqref{eq:V_CoVaR}, where $\E[\mV^{(\VaR, \CoVaR)}(\vr_t, (X_t, Y_t))]=\vzero$ \textit{if and only if} $\vr_t=(\VaR_\b(X),\CoVaR_{\a|\b}(X|Y))$. We refer to Section~\ref{app:Remark 4.5} for a simulation study illustrating the possible detrimental effects of the non-strictness of \eqref{eq:Banulescu-Radu CoVaR id}. We stress that the non-strictness is not only problematic in the uncountably many cases where $(1-\a')(1-\b') = (1-\a)(1-\b)$, but in \textit{all} cases where VaR is overpredicted and CoVaR underpredicted or the other way around. This is because by overpredicting (underpredicting) VaR and underpredicting (overpredicting) CoVaR, the different biases cancel out in the one-dimensional identification function in \eqref{eq:Banulescu-Radu CoVaR id}, thus leading to a loss of power in identifying misspecified forecasts; see Figure~\ref{fig:Figure_comp}.
Intuitively speaking, two moment conditions via a two-dimensional identification function are necessary to uniquely identify the two-dimensional pair $(\VaR_\b(X), \CoVaR_{\a|\b}(Y|X))$. 
This is in line with other two-dimensional strict identification functions for two-dimensional functionals in the literature, e.g.~for Value at Risk and Expected Shortfall, see \cite{FZ16a, NZ17}.

\end{rem}

\subsection{Multi-objective elicitability results}

Recall that $\VaR_{\b}(Y)$, $(\VaR_\b(Y), \ES_\b(Y) )$ and $\E(Y)$ are all elicitable. Surprisingly, 
their conditional counterparts $(\VaR_\b(X), \CoVaR_{\a|\b}(Y|X) )$, $(\VaR_\b(X),$ $\CoVaR_{\a|\b}(Y|X), \CoES_{\a|\b}(Y|X) )$ and $(\VaR_\b(X), \MES_\b(Y|X) )$ fail to be elicitable despite being identifiable (Section~\ref{app:negative results}). 
This is due to integrability conditions, causing an extreme gap between the class of strict identification functions and the class of strictly consistent scoring functions, which turns out to be empty.
Thus, comparative backtests cannot be implemented using a scalar strictly consistent scoring function. 
Furthermore, the conditional elicitability results of Section~\ref{app:cond id and el} can hardly be used for forecast comparisons, unless the $\VaR_\b(X)$ forecasts are the same and correctly specified. However, the conditional elicitability in combination with Theorem~\ref{thm:conditional el} immediately yields the following novel joint multi-objective elicitability results with respect to the lexicographic order $\lex$ on $\R^2$. These results can readily be used for comparing systemic risk forecasts as detailed in Section~\ref{sec:tests}.
We again use the notation introduced in \eqref{eq:notation}.

\begin{thm}\label{thm:mo el}
Let $\a,\b\in(0,1)$.
\begin{enumerate}[\rm (i)]
\item
On $\F \subseteq \F_{\b}^0(\R^2)$, the score $S^{\VaR}\colon \R\times\R^2\to\R$,
\begin{equation}\label{eq:Score VaR}
S^{\VaR}\big(v,(x,y)\big) = \big(\one\{x\le v\} - \b\big)h(v) - \one\{x\le v\}h(x) +a^{\VaR}(x,y)
\end{equation}
is strictly $\F$-consistent for $\F\ni F_{X,Y}\mapsto \VaR_{\b}(F_X)$,
if $h\colon\R\to\R$ is strictly increasing and for all $v\in\R$, the function $(x,y)\mapsto a^{\VaR}(x,y)-\one\{x\le v\}h(x)$ is $\F$-integrable.
\item
On $\F\subseteq \{F_{X,Y}\in \F^0(\R^2)\colon F_X\in \F^0_{(\b)}(\R), \ F_{Y|X\ge\VaR_\b(X)}\in \F^0_{\a}(\R)\}$, the pair $\F\ni F_{X,Y}\mapsto (\VaR_\b(F_X), \CoVaR_{\a|\b}(F_{X,Y}))$ is multi-objective elicitable with respect to $(\R^2, \lex)$. 
A strictly $\F$-consistent multi-objective scoring function $\mS^{(\VaR, \CoVaR)}\colon \R^2\times \R^2\to(\R^2, \lex)$ is given by
\begin{align}\label{eq:S_CoVaR}
			\mS^{(\VaR, \CoVaR)}\big((v,c),(x,y)\big) 
			&=\begin{pmatrix}
			S^{\VaR}\big(v,(x,y)\big)\\ 
			S_v^{\CoVaR}\big(c,(x,y)\big)
			\end{pmatrix}, \\[0.5em]
			\nonumber
			S_v^{\CoVaR}\big(c,(x,y)\big) 
			&= \one\{x>v\} \Big[\big(\one\{y\le c\} - \a\big) g(c) - \one\{y\le c\}g(y) + a(y)\Big] \\ \nonumber
			&\quad + a^{\CoVaR}(x,y),
\end{align}
where 
$g\colon\R\to\R$ is strictly increasing and $S_v^{\CoVaR}$ is $\F$-integrable for all $v\in\R$.
\item
On $\F\subseteq \{F_{X,Y}\in \F^0(\R^2)\colon F_X\in \F^0_{(\b)}(\R), \ F_{Y|X\ge\VaR_\b(X)}\in \F^1_{\a}(\R)\}$, the triplet $\F\ni F_{X,Y}\mapsto (\VaR_\b(F_X), \CoVaR_{\a|\b}(F_{X,Y}), \CoES_{\a|\b}(F_{X,Y}))$ is multi-objective elicitable with respect to $(\R^2, \lex)$.
A strictly $\F$-consistent multi-objective scoring function $\mS^{(\VaR, \CoVaR, \CoES)}\colon \R^3\times \R^2\to(\R^2, \lex)$ is given by
\begin{align}\label{eq:S_CoES}
			&\mS^{(\VaR, \CoVaR, \CoES)}\big((v,c,e),(x,y)\big) 
			=\begin{pmatrix}
			S^{\VaR}\big(v,(x,y)\big)\\ 
			S_v^{(\CoVaR, \CoES)} \big((c,e),(x,y)\big)
			\end{pmatrix},\\[0.5em]
			\nonumber
			&S_v^{(\CoVaR, \CoES)}\big((c,e),(x,y)\big) = 
 \one\{x>v\}\Big[\big(\one\{y\le c\} - \a\big) g(c) - \one\{y\le c\}g(y) 
 \\ \nonumber
&
\quad +\phi'(e)\Big(e - \tfrac{1}{1-\a}\big(y\one\{y>c\} + c(\one\{y\le c\} - \a)\big)\Big) - \phi(e)+ a(y) \Big] + a^{\CoES}(x,y),
\end{align}
where 
 $g\colon\R\to\R$ is increasing, $\phi\colon\R\to\R$ is strictly convex with subgradient $\phi'<0$, 
 and $S_v^{(\CoVaR, \CoES)}$ is $\F$-integrable for all $v\in\R$.
 \item
On $\F\subseteq\{F_{X,Y}\in \F^0(\R^2)\colon F_X\in \F^0_{(\b)}(\R), \ F_{Y|X\ge\VaR_\b(X)}\in \F^1(\R)\}$, the pair $F_{X,Y}\mapsto (\VaR_\b(F_X), \MES_{\b}(F_{X,Y}))$ is multi-objective elicitable with respect to ${(\R^2, \lex)}$.
A strictly $\F$-consistent multi-objective scoring function $\mS^{(\VaR, \MES)}\colon \R^2\times \R^2\to(\R^2, \lex)$ is given by
\begin{align}\label{eq:S_MES}
			\mS^{(\VaR, \MES)}\big((v,\mu),(x,y)\big) 
			&=\begin{pmatrix}
			S^{\VaR}\big(v,(x,y)\big)\\ 
			S_v^{\MES} \big(\mu,(x,y)\big)
			\end{pmatrix}, \\[0.5em]
			\nonumber
			S_v^{\MES}\big(\mu,(x,y)\big) 
			&= \one\{x>v\} \Big[\phi'(\mu)(\mu - y) - \phi(\mu) + a(y)\Big] + a^{\MES}(x,y),
\end{align}
where 
 $\phi\colon\R\to\R$ is strictly convex with subgradient $\phi'$ and $S_v^{\MES}$ is $\F$-integrable for all $v\in\R$.
\end{enumerate}
\end{thm}

The choice of the functions $a, a^{\VaR},  a^{\CoVaR}, a^{\CoES}, a^{\MES}$ in Theorem \ref{thm:mo el} is inessential. It may only influence the integrability of the scoring functions on the one hand and it can control the sign of the scores on the other hand.
For the remaining functions in Theorem \ref{thm:mo el}, one could choose the standard choices, e.g., the identity for $h$ and $g$ in \eqref{eq:Score VaR}, \eqref{eq:S_CoVaR} and \eqref{eq:S_CoES}, giving rise to the common `pinball loss', or $\phi(y) = y^2$ in \eqref{eq:S_MES} leading to the square loss. 
For further possible choices, especially for $\phi$ in \eqref{eq:S_CoES}, we refer to Subsection \ref{Description of Tests}.

Table~\ref{tab:overview} summarises the results of Section~\ref{Structural results for CoVaR, CoES and MES}. For purposes of comparison, the final three rows display the properties of $\VaR_{\b}(Y)$ and $\ES_{\b}(Y)$. The multi-objective elicitability of the pair $(\VaR_\b(Y), \ES_\b(Y) )$ follows from Example \ref{example: lex functionals b}. Table~\ref{tab:overview} highlights the structural difference between $\VaR_{\b}(Y)$, $\E(Y)$, $(\VaR_\b(Y), \ES_\b(Y) )$ on the one hand and their conditional counterparts $(\VaR_\b(X), \CoVaR_{\a|\b}(Y|X) )$, $(\VaR_\b(X), \MES_\b(Y|X) )$, and $(\VaR_\b(X), \CoVaR_{\a|\b}(Y|X),$ $\CoES_{\a|\b}(Y|X))$ on the other hand. While elicitability in the usual sense holds for the former, it fails for the latter. This highlights the importance of the newly introduced concept of multi-objective elicitability, which allows for comparative backtests. Section~\ref{sec:tests} shows how to implement comparative backtests with the multi-objective scores of Theorem~\ref{thm:mo el}.

\begin{table}
	\caption{\label{tab:overview}Overview of properties of (systemic) risk measures. \cmark (\xmark) indicates that property does (does not) apply.} 
	
	\centering
	\small
		\begin{tabular}{lccc}
			\toprule
	Risk Measure	 	&  Identifiability  &	Elicitability				& Multi-objective \\
									& 									&											&Elicitability 	\\
						\midrule
	$\CoVaR_{\a|\b}(Y|X)$                                          & \xmark & \xmark & \xmark \\
	$\CoES_{\a|\b}(Y|X)$                                           & \xmark & \xmark & \xmark \\
	$\MES_{\a|\b}(Y|X)$                                            & \xmark & \xmark & \xmark \\
	\midrule
	$(\VaR_\b(X), \CoVaR_{\a|\b}(Y|X) )$                      & \cmark & \xmark & \cmark \\
	$(\VaR_\b(X), \CoVaR_{\a|\b}(Y|X), \CoES_{\a|\b}(Y|X) )$  & \cmark & \xmark & \cmark \\
	$(\VaR_\b(X), \MES_\b(Y|X) )$                             & \cmark & \xmark & \cmark \\
	\midrule
	$\VaR_{\a}(Y)$                                            & \cmark & \cmark & \cmark \\
	$\ES_{\a}(Y)$                                             & \xmark & \xmark & \xmark \\
	$(\VaR_\a(Y), \ES_\a(Y) )$                                & \cmark & \cmark & \cmark \\
			\bottomrule
		\end{tabular}\qquad\qquad\qquad\qquad\qquad\qquad\qquad\qquad\qquad\qquad\qquad\qquad\qquad\qquad\qquad\qquad\qquad\qquad\qquad\qquad\qquad\qquad\qquad\qquad\qquad\qquad\qquad\qquad\qquad\qquad\qquad\qquad\qquad\qquad\qquad\qquad\qquad\qquad\qquad\qquad\qquad\qquad\
\end{table}

\section{Diebold--Mariano tests for multi-objective scores}
\label{sec:tests}

\subsection{Two-sided tests}\label{Two-Sided Tests}

\citet{DM95} propose to use formal hypothesis tests to account for sampling uncertainty in forecast comparisons. These so-called Diebold--Mariano (DM) tests are widely used in empirical forecast comparisons and continue to be studied in the theoretical literature. However, consistent with the extant notion of strict consistency, DM tests have hitherto relied on scalar scoring functions. Thus, here we show how to use our two-dimensional multi-objective scores from Theorem~\ref{thm:mo el} in DM tests, with a special focus on the implications caused by the lexicographic order.

To that end, denote by $\mS=(S_1, S_2)^\prime$ one of the multi-objective scores of Theorem~\ref{thm:mo el}. 
Let $\{\vr_{t,(1)}\}_{t=1,\ldots,n}$ and $\{\vr_{t,(2)}\}_{t=1,\ldots,n}$ be the appertaining competing sequences of forecasts (e.g., if $\mS=\mS^{(\VaR,\CoVaR)}$, then $\vr_{t,(i)}=(\widehat{\VaR}_{t,(i)},\ \widehat{\CoVaR}_{t,(i)})$ for $i=1,2$). 
The verifying observations are $\{(X_t, Y_t)\}_{t=1,\ldots,n}$. We compare the two forecasts via the (bivariate) score differences $\vd_t:=(d_{1t}, d_{2t})^\prime:=\mS(\vr_{t,(1)}, (X_t, Y_t))-\mS(\vr_{t,(2)}, (X_t, Y_t))$. The two-sided null hypothesis is that both forecasts predict equally well on average, i.e.,
	$H_0^{=} \colon \E[\overline{\vd}_n]=\vzero$ for all $n=1,2,\ldots$,
where $\overline{\vd}_n:=(\overline{d}_{1n}, \overline{d}_{2n})^\prime:=\frac{1}{n}\sum_{t=1}^{n}\vd_t$. 
(Along the lines of \cite{GW06}, one can also test the conditional null hypothesis 
$H_0^{*=} \colon \E[ \vd_t\mid \mathfrak F_{t-1}]=\vzero$ for all $t=1,2,\ldots$,
where the $\sigma$-algebra $\mathfrak F_{t-1}$ contains all information available at time $t-1$.)
We test $H_0^{=}$ using the Wald-type test statistic
\begin{equation}
\label{eq:T_n}
	\mathcal{T}_n=n \overline{\vd}_n^\prime\widehat{\mOmega}_n^{-1} \overline{\vd}_n,
\end{equation}
where $\widehat{\mOmega}_n$ 
is some consistent estimator of the variance-covariance matrix $\mOmega_n=\Var(\sqrt{n}\overline{\vd}_n)$ 
under the null hypothesis (i.e., in the componentwise norm, $\|\widehat{\mOmega}_n-\mOmega_n\|\overset{\p}{\longrightarrow}0$, as $n\to\infty$). To account for possible autocorrelation in the sequence $\{\vd_t\}_{t=1,\ldots,n}$, one can use

\begin{multline}\label{eq:hat Omega_n}
	\widehat{\mOmega}_n=\begin{pmatrix}\widehat{\sigma}_{11,n} & \widehat{\sigma}_{12,n} \\ \widehat{\sigma}_{12,n} & \widehat{\sigma}_{22,n}\end{pmatrix}=\frac{1}{n}\sum_{t=1}^{n}(\vd_t-\overline{\vd}_n)(\vd_t-\overline{\vd}_n)^\prime \\ 
	+ \frac{1}{n}\sum_{h=1}^{m_n}w_{n,h}\sum_{t=h+1}^{n}\Big[(\vd_t-\overline{\vd}_n)(\vd_{t-h}-\overline{\vd}_n)^\prime + (\vd_{t-h}-\overline{\vd}_n)(\vd_t-\overline{\vd}_n)^\prime\Big],
\end{multline}
where $m_n\rightarrow\infty$ is a sequence of integers satisfying $m_n=o(n^{1/4})$, and $w_{n,h}$ is a uniformly bounded scalar triangular array with $w_{n,h}\rightarrow1$, as $n\to\infty$, for all $h=1,\ldots,m_n$; see \citet{Whi01} for detail.
Under the assumption that $\{\vd_t\}_{t=1,\ldots,n}$ does not exhibit autocorrelation under the null, $m_n$ can be set to 0 such that 
$\widehat{\mOmega}_n$ is simply the sample variance-covariance matrix.
\begin{thm}\label{thm:DM}
Suppose that $\mOmega_n\longrightarrow\mOmega$, as $n\to\infty$, where $\mOmega$ is positive definite. Then, under technical Assumption~\ref{ass:DM} (see Supplement Section~\ref{app:Proofs}), it holds under $H_0^{=}$ that
\[
	\sqrt{n}\widehat{\mOmega}_n^{-1/2}\overline{\vd}_{n}\overset{d}{\longrightarrow}N(\vzero, \mI_{2\times2}),\qquad\text{as }n\to\infty,
\]
where $\mI_{2\times2}$ denotes the $(2\times2)$-identity matrix. In particular, $\mathcal{T}_n\overset{d}{\longrightarrow}\chi_{2}^{2}$, as $n\to\infty$, where $\chi_2^{2}$ denotes a $\chi^2$-distribution with 2 degrees of freedom.
\end{thm}

Thus, we reject $H_0^{=}$ at significance level $\nu$, if $\mathcal{T}_n>\chi_{2,1-\nu}^2$, where $\chi_{2,1-\nu}^2$ is the $(1-\nu)$-quantile of the $\chi_2^2$-distribution. A typical non-rejection region in terms of $\overline{d}_{1n}$ and $\overline{d}_{2n}$ is sketched in Figure~\ref{fig:RejReg}~(a), and has the well-known ellipse shape. For brevity, we leave out a formal investigation of our test under the alternative. We mention, however, that consistency of our test can be established along the lines of \citet{GW06}. 

\begin{figure}
	\centering
		\includegraphics[width=\textwidth]{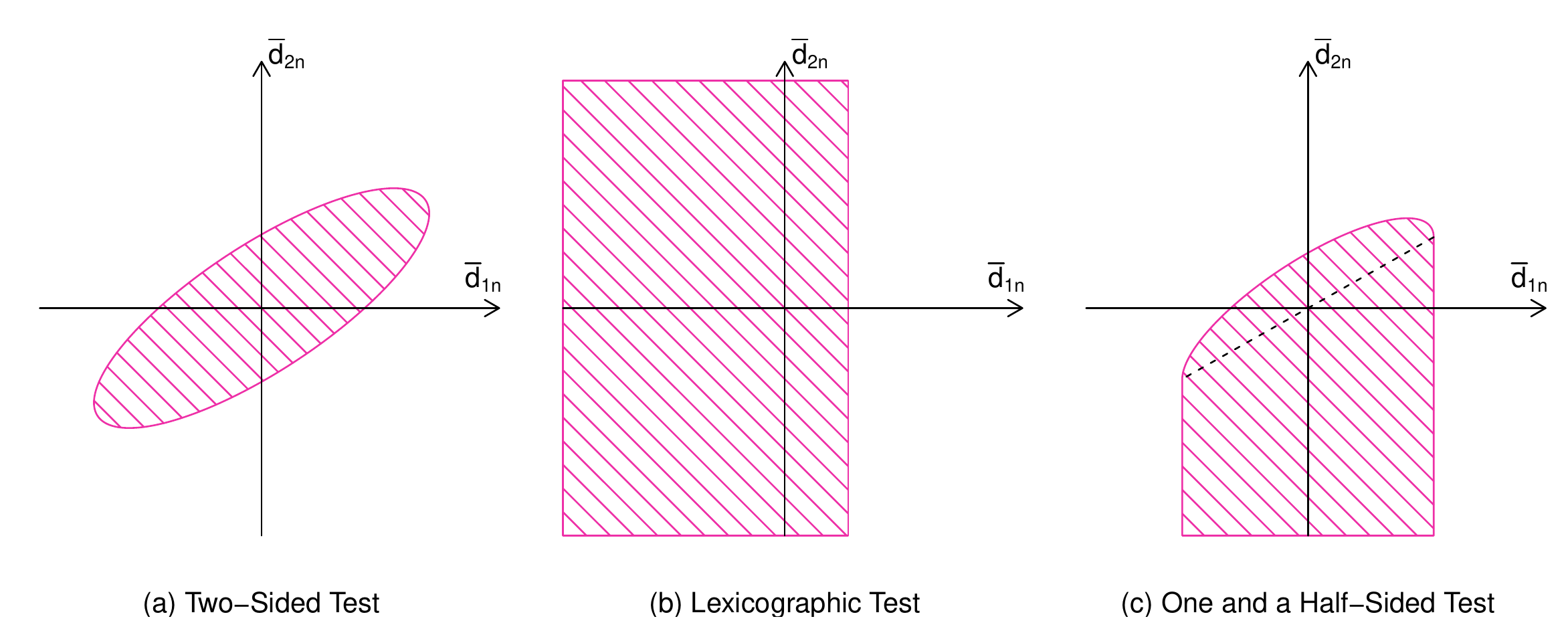}
	\caption{Non-rejection regions for two-sided DM test in (a), lexicographic DM test in (b), and one and a half-sided DM test in (c).}
	\label{fig:RejReg}
\end{figure}

\subsection{One and a half-sided tests}\label{One-Sided Tests}

In several contexts, it may be desirable to perform a one-sided comparative backtest to establish the superiority of risk forecasts $\{\vr_{t,(2)}\}_{t=1,\ldots,n}$ over some benchmark forecasts $\{\vr_{t,(1)}\}_{t=1,\ldots,n}$. 
These different forecasts could stem from two different internal models of a financial institution (say, a legacy model and and extension thereof). It could also be the case that the benchmark forecasts originate from a regulatory standard model, and---in line with the conservative backtesting approach of \cite{FZG16}---the financial institution has the onus of proof to show the  superiority of its internal model over the standard model.
In such situations, it is tempting to test the null hypothesis $\E[\overline{\vd}_n]\lex\vzero$, which is equivalent to 
\begin{equation}\label{eq:lex}
	\E[\overline{d}_{1n}]<0 \qquad\text{or}\qquad\Big(\E[\overline{d}_{1n}]=0\quad\text{and}\quad\E[\overline{d}_{2n}]\leq0 \Big).
\end{equation}
Here, the goal would be to reject the null of better or, at least, equally good benchmark forecasts as evidence of the superiority of $\{\vr_{t,(2)}\}_{t=1,\ldots,n}$. However, as under standard conditions (see Theorem~\ref{thm:DM}), $\sqrt{n}(\overline{d}_{1n}, \overline{d}_{2n})^\prime$ has an asymptotic bivariate normal distribution, the probability that $\overline{d}_{1n}=0$ and $\overline{d}_{2n}\leq0$ vanishes for large sample sizes. Thus, testing \eqref{eq:lex} amounts to a test of the null $\E[\overline{d}_{1n}]<0$, i.e., that the benchmark VaR forecasts are superior. This null can be tested via $\mathcal{T}_{1n}=\sqrt{n}\overline{d}_{1n}/\widehat{\sigma}_{11,n}^{1/2}$, where we reject \eqref{eq:lex} at significance level $\nu\in(0,1)$ when $\mathcal{T}_{1n}>\Phi^{-1}(1-\nu)$. The corresponding non-rejection region is sketched in Figure~\ref{fig:RejReg}~(b). In particular, we would reject $\E[\overline{\vd}_n]\lex\vzero$ solely based on the predictive performance of the $\VaR(X_t)$ component. In other words, once $\E[\overline{d}_{1n}]<0$ is rejected such that the internal VaR forecasts are superior, the internal forecasts $\vr_{t,(1)}$ are preferable in the lexicographic order \textit{irrespective} of the quality of the systemic risk component. 
Since this would entirely ignore the motivation of the backtest, we disregard this one-sided backtesting approach. 

As a compromise, we suggest to test the following \textit{`one and a half-sided'} null hypothesis.
Since the two systemic risk forecasts only play a role for the ranking in the lexicographic order when $\E[\overline{d}_{1n}]=0$, our suggested null hypothesis takes the form
\begin{equation*}
	H_0^{\lex} \colon \E[\overline{d}_{1n}]=0\quad\text{and}\quad\E[\overline{d}_{2n}]\leq0\qquad\text{for all }n=1,2,\ldots.
\end{equation*}
We demonstrate below how to interpret a rejection of this null. The region in $\mathbb{R}^2$ pertaining to $H_0^{\lex}$ is the lower part of the vertical axis in Figure~\ref{fig:RejReg}~(c).

Obviously, $H_0^{\lex}$ is the union of all $H_0^{(c)}$ with $c\le0$, where 
\[	
	H_0^{(c)} \colon \E[\overline{d}_{1n}]=0\quad\text{and}\quad \E[\overline{d}_{2n}]=c\qquad\text{for all }n=1,2,\ldots
\]
For each individual $c\le0$, this can be tested using the Wald-type test statistic $\mathcal{T}_{n}^{(c)}$, where $\mathcal{T}_{n}^{(c)}$ is defined similarly as $\mathcal{T}_n$, only with $\vd_t$ replaced by $\vd_t^{(c)}=(d_{1t}, d_{2t}-c)^\prime$. (Note that this substitution leaves $\widehat{\mOmega}_n$ unaffected.) Thus, we reject $H_0^{\lex}$ if and only if 
\begin{equation}\label{eq:rejection}
	\mathcal{T}_n^{(c)}> \chi_{2,1-\widetilde{\nu}}^2\qquad\text{for all }c\leq0,
\end{equation}
where $\widetilde{\nu}\in(0,1)$. The area associated with the appertaining non-rejection region is shaded in pink in Figure~\ref{fig:RejReg}~(c). The rejection condition \eqref{eq:rejection} is of course equivalent to 
\(
	\mathcal{T}_n^{\OS}
	:=
	\inf_{c\le 0} \mathcal{T}_n^{(c)}
	= \mathcal{T}_n^{(c^*)}
	>\chi_{2,1-\widetilde{\nu}}^2,
\)
where the solution $c^*:= \min\big\{0, \overline{d}_{2n} - (\widehat{\sigma}_{12,n}/\widehat{\sigma}_{11,n})\overline{d}_{1n}\big\}$ follows from a simple quadratic minimisation problem.
Hence,
\begin{equation}\label{eq:equiv NR}
	\mathcal{T}_n^{\OS}=n\Big(\overline{d}_{1n}, \max\big\{\overline{d}_{2n}, (\widehat{\sigma}_{12,n}/\widehat{\sigma}_{11,n})\overline{d}_{1n}\big\}\Big)\widehat{\mOmega}_n^{-1}\begin{pmatrix}\overline{d}_{1n}\\ \max\big\{\overline{d}_{2n}, (\widehat{\sigma}_{12,n}/\widehat{\sigma}_{11,n})\overline{d}_{1n}\big\}\end{pmatrix}.
\end{equation}

To illustrate this graphically, note that the line defined by $z_{2}=(\widehat{\sigma}_{12,n}/\widehat{\sigma}_{11,n})z_{1}$---indicated by the dashed line in Figure~\ref{fig:RejReg}~(c)---passes through the extremal (negative and positive) horizontal points of the ellipse. Thus, if $\overline{d}_{2n}>(\widehat{\sigma}_{12,n}/\widehat{\sigma}_{11,n})\overline{d}_{1n}$, \eqref{eq:equiv NR} parametrizes the upper half of the tilted ellipse, and if $\overline{d}_{2n}\leq(\widehat{\sigma}_{12,n}/\widehat{\sigma}_{11,n})\overline{d}_{1n}$ the area below. In our numerical experiments, we use $\mathcal{T}_n^{\OS}$ to test $H_0^{\lex}$.

The next proposition shows that rejecting $H_0^{\lex}$ when $\mathcal{T}_n^{\OS}>\chi_{2,1-\widetilde{\nu}}^2$ leads to a test of level $\nu=1/2\big[1+\widetilde{\nu}-F_{\chi_1^2}(\chi_{2,1-\widetilde{\nu}}^{2})\big]$, where $F_{\chi_1^2}$ denotes the cdf of a $\chi_{1}^{2}$-distribution. 

\begin{prop}\label{prop:OS}
Under the conditions of Theorem~\ref{thm:DM}, rejecting $H_0^{\lex}$ if $\mathcal{T}_n^{\OS}>\chi_{2,1-\widetilde{\nu}}^2$ leads to an asymptotic size $\nu$-test with $\nu=1/2\big[1+\widetilde{\nu}-F_{\chi_1^2}(\chi_{2,1-\widetilde{\nu}}^{2})\big]$. 
That is,
\[
	\sup_{c\leq0}\lim_{n\to\infty}\p\Big\{H_0^{\lex}\ \text{is rejected based on}\ \mathcal{T}_{n}^{\OS}\ \Big\vert\ H_0^{(c)}\ \text{holds}\Big\}=\nu.
\]
\end{prop}

\begin{rem}\label{rem:sig level corr}
For a test of $H_0^{\lex}$ with desired significance level of $\nu$, one can determine the required level $\widetilde{\nu}$ from $\nu=1/2\big[1+\widetilde{\nu}-F_{\chi_1^2}(\chi_{2,1-\widetilde{\nu}}^{2})\big]$ using standard root-finding algorithms. E.g., if $\nu=1\%$ / $\nu=5\%$ / $\nu=10\%$, then $\widetilde{\nu}=1.60\%$ / $\widetilde{\nu}=7.66\%$ / $\widetilde{\nu}=14.9\%$.
\end{rem}

\begin{rem}\label{rem:degen}
To compare systemic risk forecasts on an equal footing, it may be desirable to compare them based on the same marginal model being used and, hence, based on the same VaR forecasts $\widehat{\VaR}_t=\widehat{\VaR}_{t,(1)}=\widehat{\VaR}_{t,(2)}$. In this case, differences in predictive ability can be attributed solely to the different dependence models; see, e.g., \citet{NZ19+} and \citet{Hog20a+}.
(Note that Theorem~\ref{thm:DM} no longer applies for identical VaR forecasts, since the limit of the covariance matrix $\mOmega$ is only positive semi-definite.)
When $\widehat{\VaR}_t=\widehat{\VaR}_{t,(1)}=\widehat{\VaR}_{t,(2)}$, testing for $\E[\overline{d}_{1n}]=0$ is redundant, and $H_0^{=}$ and $H_0^{\lex}$ reduce to $\E[\overline{d}_{2n}]=0$ and $\E[\overline{d}_{2n}]\leq0$, respectively. These hypotheses can be tested using $\mathcal{T}_{2n}=\sqrt{n}\overline{d}_{2n}/\widehat{\sigma}_{22,n}^{1/2}$.
Under $\E[\overline{d}_{2n}]=0$, $\mathcal{T}_{2n}$ converges to an $N(0,1)$-limit when restricting the assumptions of Theorem~\ref{thm:DM} to the sequence $d_{2t}$. Hence, we would reject $\E[\overline{d}_{2n}]=0$ (or $\E[\overline{d}_{2n}]\leq0$) at significance level $\nu\in(0,1)$, if $|\mathcal{T}_{2n}|>\Phi^{-1}(1-\nu/2)$ (or $\mathcal{T}_{2n}>\Phi^{-1}(1-\nu)$).
\end{rem}

\begin{figure}
	\centering
		\includegraphics[width=0.85\textwidth]{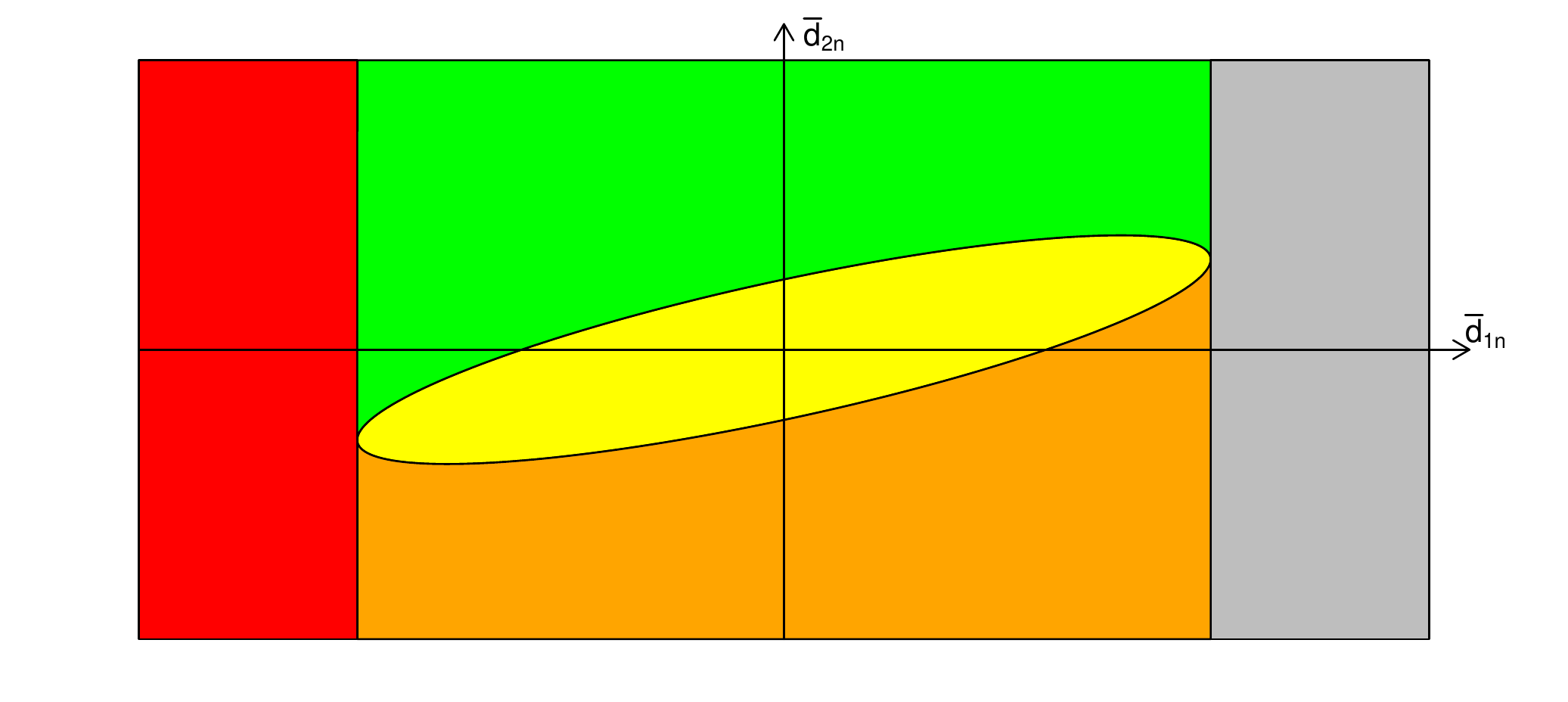}
	\caption{Interpretation of test decisions. Yellow zone corresponds to ellipse. Red (grey) zone corresponds to the rectangle to the left (right) of the ellipse. Green (orange) zone corresponds to area immediately above (below) the ellipse.}
	\label{fig:DecHeur}
\end{figure}

Figure~\ref{fig:DecHeur} presents a schematic for interpreting test decisions on $H_0^{\lex}$ depending on the values of $\overline{d}_{1n}$ and $\overline{d}_{2n}$. In a regulatory context, it extends the three-zone  traffic-light classification of the \citet[Paragraph~99.39]{BCBSBF19}: 
\begin{quote}
\label{quotation:traffic-light}
{
``The green zone corresponds to backtesting results that do not themselves suggest a problem with the quality or accuracy of a bank's model. The yellow zone encompasses results that do raise questions in this regard, but where such a conclusion is not definitive. The red zone indicates a backtesting result that almost certainly indicates a problem with a bank's risk model.''}
\end{quote}
The union of the yellow and the orange areas in Figure~\ref{fig:DecHeur} corresponds to the non-rejection region depicted in Figure~\ref{fig:RejReg}~(c). By symmetry, the union of the yellow and the green areas is the non-rejection region associated with the null
\[
	H_0^{\gex} \colon \E[\overline{d}_{1n}]=0\quad\text{and}\quad\E[\overline{d}_{2n}]\geq0\qquad\text{for all }n=1,2,\ldots.
\]
Thus, for $(\overline{d}_{1n},\overline{d}_{2n})'$ in the yellow area we can neither reject $H_0^{\lex}$ nor $H_0^{\gex}$ (at significance level $\nu=1/2\big[1+\widetilde{\nu}-F_{\chi_1^2}(\chi_{2,1-\widetilde{\nu}}^{2})\big]$, respectively). This implies that there is no evidence of differences in predictive ability between the internal and the benchmark model (at level $\widetilde{\nu}$).
From a regulatory perspective, the bank's internal model is `at the boundary' of what can be deemed acceptable. 
Hence, we suggest heightened attention and close monitoring by the regulator, as indicated by the yellow colour.

When $(\overline{d}_{1n},\overline{d}_{2n})'$ falls into the orange area, then---while the VaR forecasts are of comparable quality---the benchmark model provides superior systemic risk forecasts. Here, as indicated by the orange colour, a revision of the internal systemic risk forecasts is called for, while the internal marginal model is in order.

In contrast, the green area indicates superior systemic risk forecasts of the internal model, with VaR forecasts being comparably accurate. In this case, the financial institution should be allowed to use its internal model for risk forecasting purposes. The green colour indicates a pass of the backtest.

In the red region, the VaR forecasts of the internal model are deemed inferior to the benchmark predictions, whence there is no basis to compare the systemic risk forecasts. 
(The red region corresponds to the rejection region of the null $\E[\overline{d}_{1n}]\ge0$ at level $\nu'=\frac12 - \frac12 F_{\chi_1^2}(\chi_{2,1-\widetilde{\nu}}^{2})$. E.g., for $\nu = 5\%$ we get $\widetilde \nu = 7.66\%$ and $\nu'= 1.17\%$.)
Here, the bank should not be allowed to use its own marginal model (the traffic light is red), but instead should be required to use the benchmark model for the marginals to ensure a fair assessment of the systemic risk forecasts. For this comparison, where the VaR forecasts are identical, one can focus solely on the systemic risk component by using $\mathcal{T}_{2n}$; see Remark~\ref{rem:degen}. For the comparison via $\mathcal{T}_{2n}$, we suggest to adopt a similar decision heuristic as in \citet{FZG16}: Neither $H_0^{\lex}$ nor $H_0^{\gex}$ can be rejected when $|\mathcal{T}_{2n}|\leq \Phi^{-1}(1-\nu)$ (corresponding to our yellow zone). When $\mathcal{T}_{2n}> \Phi^{-1}(1-\nu)$ ($\mathcal{T}_{2n}< \Phi^{-1}(\nu)$) such that $H_0^{\lex}$ ($H_0^{\gex}$) can be rejected, the internal systemic risk forecasts are superior (inferior), corresponding to our green (red) zone.

Similarly as for the red region, there are no grounds for meaningful systemic risk forecast comparisons in the grey area, since the internal model's VaR forecasts are superior. 
(Formally, it corresponds to the rejection region of the null $\E[\overline{d}_{1n}]\le0$ at level $\nu'=\frac12 - \frac12 F_{\chi_1^2}(\chi_{2,1-\widetilde{\nu}}^{2})$.)
In this case, there is no cause for action on the end of the bank (hence the grey colour), but the regulator should rather adopt the bank's VaR model as a basis for comparing the systemic risk forecasts. As before, the subsequent comparison of systemic risk forecasts should be carried out by using $\mathcal{T}_{2n}$, since the VaR forecasts are identical. 

Section~\ref{app:Monte Carlo Simulations} investigates the finite-sample properties of $\mathcal{T}_{n}$, $\mathcal{T}_n^{\OS}$, and $\mathcal{T}_{2n}$ under $H_0^{=}$ and $H_0^{\lex}$ in detail. Here, we only summarize the main findings. First, size is adequate already for $n=500$, which is encouraging since effective sample sizes in risk forecast comparisons are small. Second, power increases markedly in $n$. Third, comparisons for (CoVaR, CoES) are slightly more powerful than those for CoVaR alone, most likely due to the increased informational content of the CoES component. Fourth, as expected for one-sided tests, departures from $H_0^{\lex}$ and $H_0^{=}$ are easier to detect for the former. Fifth, it is in general easier to detect differences in predictive ability of the systemic risk component, when the VaR component is identical (instead of only comparable) across forecasts. Intuitively, the inclusion of comparable VaR forecasts dilutes the power of the test in the systemic risk component.

While the schematic in Figure~\ref{fig:DecHeur} is motivated by a regulatory framework, we stress that it can be used in the context of any comparative backtest between different models. 
Such a general example is provided in Section~\ref{Empirical Application}.

\section{Empirical application}
\label{Empirical Application}

Consider daily log-losses $X_{-r+1},\ldots,X_n$ on the S\&P~500 and log-losses $Y_{-r+1},\ldots,Y_n$ on the DAX~30 from 2000--2020, where the data are taken from \href{https://www.wsj.com/market-data/quotes}{www.wsj.com/market-data/quotes} (ticker symbols: SPX and DX:DAX). So if $P_{Z,t}$ denotes the stock index value at time $t$, then $Z_t=-\log(P_{Z,t}/P_{Z,t-1})$ ($Z\in\{X,Y\}$). We only keep those observations where data on both indexes are available, giving us $n+r=5\,193$ observations. Here, for $\alpha=\b=0.95$, we compare rolling-window (VaR, CoVaR, CoES) forecasts for the series $\{(X_t,Y_t) \}_{t=1,\ldots,n}$, where $r=1\,000$ denotes the moving window length. 
The choice of $X$ and $Y$ amounts to considering the risk for large losses of the DAX~30, given that the world's leading stock index---the S\&P~500---is in distress. To promote flow in this section, we often refer to the simulation setup in Section~\ref{app:Monte Carlo Simulations} for details on the time series models and the risk forecast computation.

For short-term risk management purposes, \emph{conditional} (systemic) risk measure forecasts are more informative than unconditional ones.
Conditional risk measures are based on the conditional distribution of $(X_t, Y_t)$, that is, $F_{(X_t,Y_t) \mid \mathfrak F_{t-1}}(x,y) = \p\{X_t\le x, Y_t\le y \mid \mathfrak F_{t-1}\} =: \p_{t-1}\{X_t\le x, Y_t\le y\}$ for $x,y\in\R$. 
Here, the filtration $\mathfrak F_{t-1}$ is generated by the information available to a forecaster at time $t-1$.
These are usually past observations $(X_{t-1},Y_{t-1}), (X_{t-2}, Y_{t-2}), \ldots$, and possibly additional exogenous information. Here, we assume $\mathfrak{F}_{t-1} = \sigma\big((X_{t-1},Y_{t-1}), (X_{t-2}, Y_{t-2}), \ldots\big)$, such that we forecast the \emph{conditional} risk measures
$\VaR_{t}(X_t) = \VaR_{\b}(F_{X_t\mid\mathfrak F_{t-1}})$, 
$\CoVaR_{t}(Y_t\vert X_t) = \CoVaR_{\a|\b}(F_{(X_t,Y_t) \mid \mathfrak F_{t-1}})$, and
$\CoES_{t}(Y_t\vert X_t) = \CoES_{\a|\b}(F_{(X_t,Y_t) \mid \mathfrak F_{t-1}})$.
For notational brevity, we suppress the dependence of the risk measures on the risk levels $\alpha$ and $\beta$, which we fix at $\alpha=\beta=0.95$.


We consider two different methods for $(\VaR_{t}(X_t), \CoVaR_{t}(Y_t\vert X_t), \CoES_{t}(Y_t\vert X_t))$ forecasting. The first method uses a simple GARCH(1,1) model for $X_t$ and $Y_t$ and---as a dependence model for the respective innovations $\varepsilon_{x,t}$ and $\varepsilon_{y,t}$---a Gaussian copula driven by GAS dynamics. That is, we assume $(\varepsilon_{x,t},\varepsilon_{y,t})\mid\mathfrak{F}_{t-1}$ to have Gaussian copula density $c(\,\cdot\,; \rho_t)$ with time-varying correlation parameter $\rho_t\in(-1,1)$ following GAS dynamics. Details on the precise specification are in Subsection~\ref{DGP}, where the same model (with $\vartheta=\infty$ in Equation \eqref{eq:GAS}) is used in the simulations. Both the marginal and the dependence model are regularly used as benchmark models in forecast comparisons.

The second method uses the GJR--GARCH(1,1) model of \citet{GJR93}. The GJR--GARCH model possesses an additional parameter that allows positive and negative shocks of equal magnitude to have a different effect on volatility. As a dependence model, we now use a $t$-copula driven by GAS dynamics, similarly as in Equation \eqref{eq:GAS}. For both models we remain agnostic regarding the specific distribution of the $\varepsilon_{x,t}$ and $\varepsilon_{y,t}$---both in the model estimation (by using the robust Gaussian quasi-maximum likelihood estimator for the GARCH-type marginal models) and the risk forecasting stage (by using their empirical cdfs $\widehat{F}_x$ and $\widehat{F}_y$ in computing the risk measures). For details on how the risk predictions are calculated, we refer to Subsection~\ref{Risk Forecasts}.

We now compare two sequences of rolling-window predictions. For each rolling window of length $r=1\,000$, we re-fit the two models on a daily basis. 
This gives us $n=4\,193$ forecasts $\{\vr_{t,(1)}\}_{t=1,\ldots,n}$ from the GARCH model with Gaussian copula, and $\{\vr_{t,(2)}\}_{t=1,\ldots,n}$ from the GJR--GARCH with $t$-copula. We interpret the $\vr_{t,(1)}$ as generic benchmark forecasts, which are to be improved upon by the $\vr_{t,(2)}$. In an oversight context, the $\vr_{t,(1)}$ may be some regulatory benchmark forecasts and the $\vr_{t,(2)}$ forecasts from the bank's internal model. Or from a banking perspective, the $\vr_{t,(1)}$ may be forecasts issued from a trading desk's legacy model and the $\vr_{t,(2)}$ are forecasts from a refined version thereof. The latter context is more realistic here, because in the regulatory framework the focus is more often on the relation of between an individual bank's returns and the market as a whole. In any case, the methodology remains the same. We compare forecasts based on the verifying observations $\{(X_t,Y_t)\}_{t=1,\ldots,n}$.

\begin{figure}[t!]
	\centering
		\includegraphics[width=\textwidth]{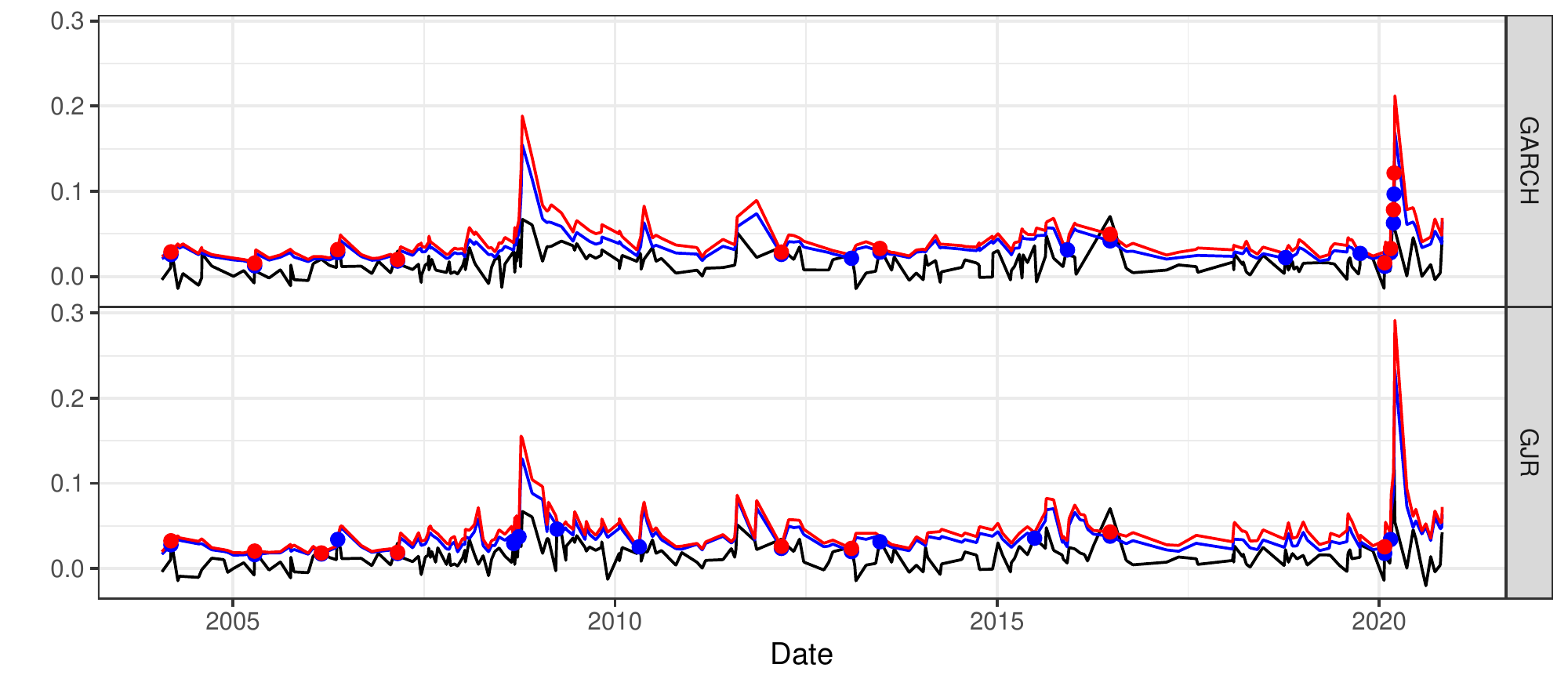}
	\caption{Top: DAX log-losses on days of VaR violation of S\&P~500 (black). CoVaR (CoES) forecasts are shown as the blue (red) line. Violation of CoVaR (CoES) forecast is indicated by a blue (red) dot. All forecasts from GARCH with Gaussian copula.
	Bottom: Same as top, only with forecasts from GJR--GARCH with $t$-copula.}
	\label{fig:Violations}
\end{figure}

Figure~\ref{fig:Violations} shows CoVaR and CoES forecasts from both models, where the top panel corresponds to the GARCH model with Gaussian copula and the bottom panel to the GJR--GARCH with $t$-copula. Specifically, the panels show the CoVaR (blue) and CoES (red) forecasts for the DAX log-losses (black) on days where the S\&P~500 exceeds its VaR forecast. 
Note that due to the different marginal models (and, hence, the different VaR forecasts), the black lines differ slightly in the upper and lower panel of Figure~\ref{fig:Violations}.
By definition, the S\&P~500 should only exceed its VaR forecast on 5\% of all trading days, i.e., on $0.05 \cdot 4193=209.65$ days in our out-of-sample period. 
With 213 (218) VaR violations, our marginal GARCH(1,1) model (GJR--GARCH(1,1) model) is close to the ideal frequency. By definition, we expect our CoVaR forecasts to be not exceeded on 95\% of these 213 days (218 days) with a VaR violation. With 15 and 16 exceedances (blue dots), which correspond to non-exceedance frequencies of 93.0\% and 92.7\%, the Gaussian copula and the $t$-copula model are reasonably close to the 95\%-benchmark. However, for the Gaussian copula, the CoVaR exceedances seem to cluster more, such as during the beginning of the Covid-19 pandemic in early 2020 (top panel of Figure~\ref{fig:Violations}). Such violation clusters are undesirable from a risk management perspective, providing some informal evidence in favour of the $t$-copula model. We investigate this more formally in the following.
Note that both panels of Figure~\ref{fig:Violations} indicate marked spikes in systemic risk during the financial crisis of 2008--2009, during the European sovereign debt crisis in the first half of the 2010's, and---most markedly---in 2020 as a consequence of the Covid-19 pandemic.

As pointed out above, we regard the GARCH model with Gaussian copula as our benchmark model. So we now want to test $H_0^{\lex}$ for (VaR, CoVaR) and (VaR, CoVaR, CoES). Due to the less clustered CoVaR exceedances and the compelling empirical evidence in favour of GJR--GARCH models \citep{GJR93,BEK11} and GAS-$t$-copula models \citep{CKL13,BC19}, we expect the GJR--GARCH model with $t$-copula to produce lower scores, i.e., better risk forecasts, possibly leading to a rejection of $H_0^{\lex}$. We carry out the tests using $\mathcal{T}_n^{\OS}$ with the scores of Equations \eqref{eq:S CoVaR appl} and \eqref{eq:S CoES appl} having 0-homogeneous score differences, and with $\widehat{\mOmega}_n$ from \eqref{eq:hat Omega_n} (with $m_n=0$); see Supplement Subsection~\ref{Description of Tests} for details. For VaR and ES forecasts, scoring functions giving 0-homogeneous score differences are typically recommended, since they allow for `unit-consistent' and powerful comparisons \citep{NZ17}. We confirm the latter in our simulations for systemic risk forecasts as well, thus justifying our choice. Let $\overline{\vd}_n=(\overline{d}_{1n}, \overline{d}_{2n})^\prime$ be defined as in Subsection~\ref{Two-Sided Tests}. Indeed, computing $\overline{\vd}_n$, we find that the GJR--GARCH with $t$-copula produces lower scores for both the (VaR, CoVaR) and the (VaR, CoVaR, CoES) forecasts. The score differences are even statistically significant at the 5\%-level: The $p$-values for the $\mathcal{T}_n^{\OS}$-based Wald test are 2.9\% for (VaR, CoVaR) and 3.0\% for (VaR, CoVaR, CoES).

\begin{figure}[t!]
	\centering
		\includegraphics[width=\textwidth]{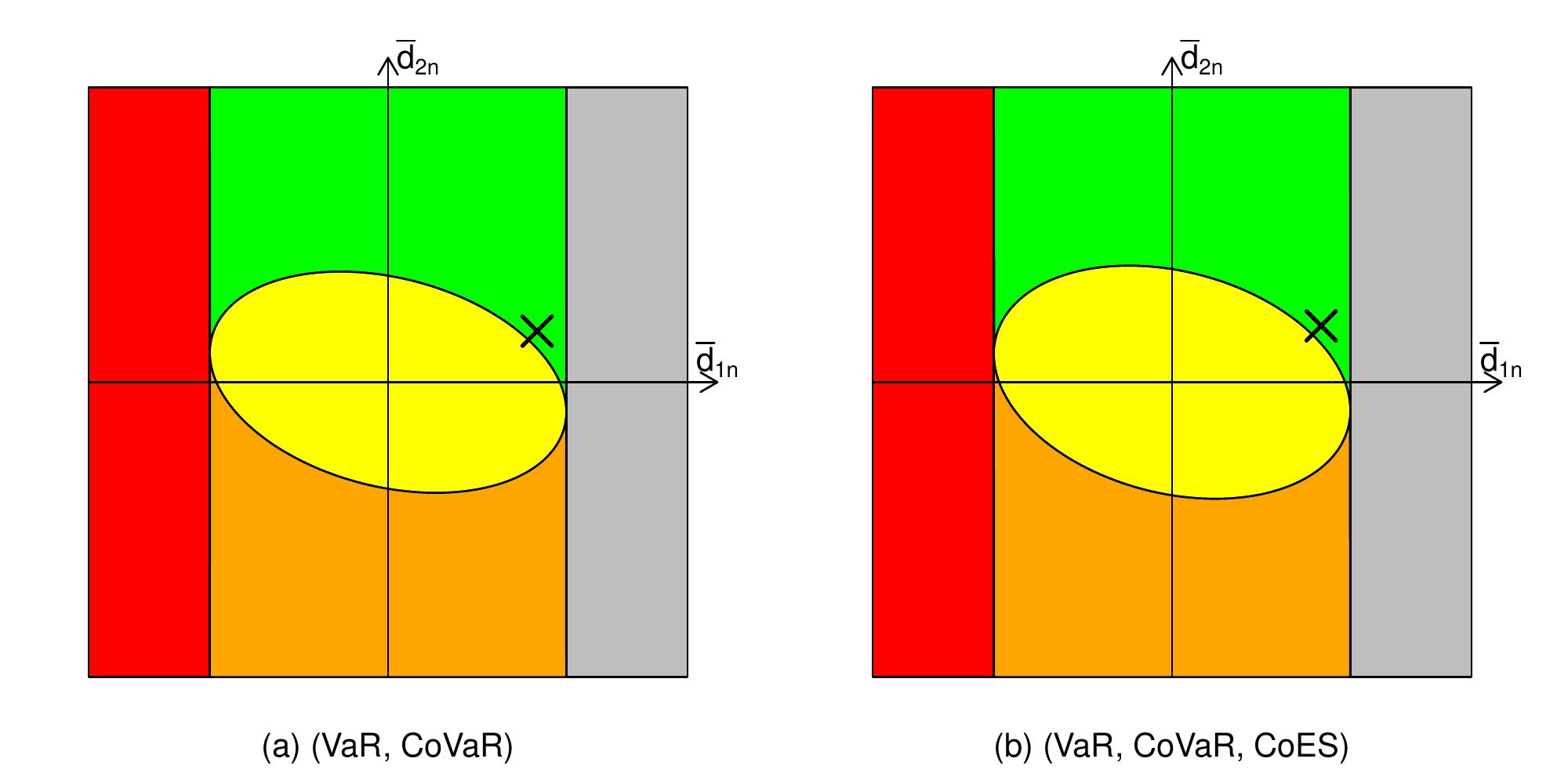}
	\caption{Panel~(a): Values of $\overline{d}_{1n}$ and $\overline{d}_{2n}$ for (VaR, CoVaR) forecasts, indicated by $\times$. Panel~(b): Values of $\overline{d}_{1n}$ and $\overline{d}_{2n}$ for (VaR, CoVaR, CoES) forecasts, indicated by $\times$.}
	\label{fig:DecAppl}
\end{figure}

Figure~\ref{fig:DecAppl} illustrates the test decisions. Panel~(a) shows the results for the (VaR, CoVaR) comparison. Both $\overline{d}_{1n}$ and $\overline{d}_{2n}$ are positive, favouring the forecasts $\vr_{t,(2)}$ from the GJR--GARCH model with $t$-copula. Additionally, the pair $\overline{\vd}_n=(\overline{d}_{1n}, \overline{d}_{2n})^\prime$ lies above the yellow ellipse, which traces the contour of a bivariate normal distribution with probability content $(100-7.66)\%=92.34\%$. This confirms the significance of the score difference at the $5\%$-level; recall Remark~\ref{rem:sig level corr}. The results for the (VaR, CoVaR, CoES) forecasts in panel~(b) are qualitatively similar.
Adopting our traffic-light interpretation,
the GJR--GARCH model with $t$-copula would be deemed an adequate risk forecasting model.

Nonetheless, the borderline significance of this example shows that discriminating between (systemic) risk forecasts requires long samples for the given parameter choice of $\b=0.95$. 
This is as expected, because by only considering observations with one extreme component, the effective sample size is massively reduced to roughly $n(1-\beta)$. Indeed, for both models, the effective out-of-sample period for comparing systemic risk forecasts is reduced from a length of 4193 to just over 200. 
The practical implication is that one should allow for sufficiently large samples to prove the superiority of the internal model over the benchmark. 
Alternatively, one could adapt a slightly looser definition of distress in the reference asset $X$, e.g., by setting $\b=0.9$.
The current evaluation period for VaR forecasts specified in the Basel framework by the \citet{BCBSBF19} is one year, amounting to roughly 250 daily returns. Even for evaluating VaR and ES forecasts, the horizon of 250 days has been called into question for being too short \citep{DLS20}, and this is only magnified for systemic risk forecasts. So whatever evaluation period for systemic risk measures is eventually settled on in a regulatory context, it likely needs to be far in excess of one year.

\section{Discussion and outlook}
\label{Summary and Outlook}

To our knowledge, this is the first paper to come up with comparative backtests for the systemic risk measures CoVaR, CoES and MES, which are crucial inputs in financial, macroeconomic and regulatory applications. Model selection procedures based on our results may enhance modelling attempts of these quantities in financial institutions. Moreover, the fact that our notion of multi-objective elicitability serves as a `truth serum' implies that the regulatory framework can be improved by enticing financial institutions to accurately model systemic risk.
In terms of calibration tests, we are only aware of one more paper: 
\cite{Banulescu-RaduETAL2019} introduce such tests, which unfortunately hinge on non-strict identification functions. This may lead to a severe loss of power under the alternative as demonstrated in Remark~\ref{rem:comparison id func} and Section~\ref{app:Remark 4.5}.
Due to the strictness of our identification functions, such a phenomenon cannot happen in our context. 

The novel concept of multi-objective elicitability is likely to be fruitful also in applications beyond the proposed DM-backtests for systemic risk measures. In Subsection~\ref{app:examples}, we provide more examples of interesting situations where conditional elicitability holds, yet classical joint elicitability fails. By virtue of Theorem~\ref{thm:conditional el}, one can now construct incentive compatible elicitation mechanisms for these functionals, or can come up with $M$-estimation procedures in a regression context.

Beyond the confines of finance, we anticipate many other interesting applications of our backtests. For instance, in economics, \citet{ABG19} and \citet{Aea21} have recently drawn attention to tail risks and their interconnections by popularizing the Growth-at-Risk, which is simply the VaR of GDP growth. This literature has led to an increase in the use of risk forecast evaluation methods in economics \citep[e.g.,][]{BS21}. Hence, the backtests developed in this paper should be relevant for future macroeconomic applications as well.

%


\end{bibunit}

\begin{bibunit}

\section*{Supplement}

\begin{appendix}

\setcounter{page}{1}

\section{Proofs for Section~\ref{Scoring functions, identification functions and (conditional) elicitability}}
\label{sec:Proofs Section 3}

\begin{proof}[{\textbf{Proof of Proposition~\ref{prop:conditional id}:}}]
Let $\mV_1\colon\A_1\times \O \to\R^{m_1}$ be a strict $\F$-identification function for $\mT_1$ 
and, for any $\vr_1\in\A_1$, let $\mV_{2,\vr_1}\colon\A_2\times \O\to\R^{m_2}$ 
be a strict $\F_{\vr_1}$-identification function for $\mT_2$.
We claim that
\[
\mV \colon (\A_1\times \A_2)\times \O \to \R^{m_1+m_2}, \qquad \big((\vr_1,\vr_2),\vy\big)\mapsto \big( \mV_1(\vr_1,\vy), \mV_{2,\vr_1}(\vr_2,\vy)\big)
\]
is a strict $\F$-identification function for $(\mT_1,\mT_2)$.
Consider some $F\in\F$ and $(\vr_1,\vr_2)\in\A_1\times \A_2$.
First, assume $(\vr_1, \vr_2) \in \mT_1(F)\times \mT_2(F)$. Then $\bar \mV_1(\vr_1,F)=\vzero$ and $\bar \mV_{2,\vr_1}(\vr_2,F)=\vzero$. 
Second, assume $\bar \mV_1(\vr_1,F)=\vzero$ and $\bar \mV_{2,\vr_1}(\vr_2,F)=\vzero$. Due to the former equality and the strictness of $\mV_1$, it holds that $\vr_1 \in \mT_1(F)$. Now we can invoke the strictness of $\mV_{2,\vr_1}$ on $\F_{\vr_1}$ to conclude that $\vr_2 \in \mT_2(F)$.
\end{proof}

\begin{proof}[{\textbf{Proof of Theorem \ref{thm:conditional el}:}}]
Let $S_1\colon\A_1\times \O \to\R$ be a strictly $\F$-consistent scoring function for $\mT_1$. 
For any $\vr_1\in\A_1$ let $S_{2,\vr_1}\colon\A_2\times \O \to\R$ 
be a strictly $\F_{\vr_1}$-consistent scoring function for $\mT_2$.
We claim that
\be{eq: S lex}
\mS \colon (\A_1\times \A_2)\times \R^d \to (\R^2, \lex), \qquad \big((\vr_1,\vr_2),\vy\big)\mapsto \big( S_1(\vr_1,\vy), S_{2,\vr_1}(\vr_2,\vy)\big)
\ee
is a strictly multi-objective $\F$-consistent scoring function for $(\mT_1,\mT_2)$ with respect to $\lex$.
To see this, consider some $F\in\F$ with $(\vt_1,\vt_2)\in \mT_1(F) \times \mT_2(F)$ and $(\vr_1,\vr_2)\in\A_1\times \A_2$.
From the strict $\F$-consistency of $S_1$ it follows that
\begin{equation}\label{eq:proof inequ 1}
\bar S_1(\vt_1,F)\le \bar S_1(\vr_1,F),
\end{equation}
where the inequality is strict if $\vr_1\notin \mT_1(F)$.
If the inequality is strict, we can already conclude that 
\(
\bar \mS(\vt_1, \vt_2,F) \slex \bar \mS(\vr_1, \vr_2,F).
\)
Otherwise, 
it must be the case that $\vr_1 \in \mT_1(F)$. 
Since by assumption $\mT_1(F)$ is a singleton, $\vr_1 = \vt_1$.
Invoking the strict $\F_{\vt_1}$-consistency of $S_{2,\vt_1}$, we obtain that 
\begin{equation}\label{eq:proof inequ 2}
\bar S_{2,\vt_1}(\vt_2,F)\le \bar S_{2,\vt_1}(\vr_2,F),
\end{equation}
where the inequality is strict if $\vr_2\notin \mT_2(F)$.
If the inequality is strict, this implies that 
\(
\bar \mS( (\vt_1, \vt_2), F) \slex \bar \mS( (\vr_1, \vr_2), F).
\)
Otherwise, equality 
implies that $\vr_2 \in \mT_2(F)$.
\end{proof}

\begin{rem}\label{rem:relaxation}
Theorem~\ref{thm:conditional el} is almost a direct analogue of Proposition~\ref{prop:conditional id}, with the intriguing exception that $\mT_1$ is assumed to be a singleton on $\F$ in Theorem \ref{thm:conditional el}.
An inspection of the proof reveals that this is needed to show the inequality in \eqref{eq:proof inequ 2}. If at this stage of the proof one has that $\vt_1,\vr_1\in \mT_1(F)$ and $\vt_1\neq \vr_1$, it is \textit{per se} not clear (and in general also not the case) that $S_{2,\vt_1} = S_{2,\vr_1}$ from which the inequality in \eqref{eq:proof inequ 2} would follow.
What would be sufficient, however, to dispense with the assumption on $\mT_1$ and to still show \eqref{eq:proof inequ 2} is the following: It is possible to choose a family of strictly $\F_{\vr_1}$-consistent scores $S_{2,\vr_1}$, $\vr_1\in\A_1$, for $\mT_2$ such that for all $F\in\F$ and for all $\vt_1, \vr_1\in \mT_1(F)$ we have equality for the expected scores, $\bar S_{2,\vt_1}(\cdot, F) = \bar S_{2,\vr_1}(\cdot, F)$. 
In many practical examples, this can be achieved; see, e.g., Example \ref{example: lex functionals b} \eqref{item:VaR ES}.
\end{rem}

\section{More background on multi-objective elicitability}
\label{app:mo-el}

\subsection{Further examples of multi-objective elicitable functionals}
\label{app:mo-scores}
\label{app:examples}

We provide more examples of functionals that are conditionally elicitable, yet fail to be elicitable in the traditional sense. Examples \eqref{item:PI1} and \eqref{item:PI2} are concerned with $\a$-prediction intervals. In line with \citet[Section 4]{FFHR2021} for some random variable $Y$ with distribution $F$, this is any interval $[a,b]\subseteq \R$ such that $\p(Y\in[a,b]) = F(b) - F(a-) \ge \a$.
Of course, each distribution $F$ has a multitude of different $\a$-prediction intervals. We consider two possible attempts of coming up with \emph{interpretable} $\a$-prediction intervals.

Two univariate scores $S, \tilde S\colon \A\times \O \to\R$ are \emph{equivalent} if there exists a positive constant $c>0$ and a function $a\colon \O \to\R$ such that 
$\tilde S(\vr,\vy) = cS(\vr,\vy) + a(\vy)$. Ignoring integrability issues for a moment, it is clear that equivalence preserves (strict) consistency.
If $S$ and $a$ are $\F$-integrable, so is $\tilde S$. However, by a convenient choice of $a$, the integrability conditions of $\tilde S$ can be milder than those of $S$. E.g., while $S(r,y) = (r-y)^2$, $r,y\in\R$, is a strictly $\F^2(\R)$-consistent scoring function for the mean, the equivalent score $\tilde S(r,y) = r^2 - 2ry$ is strictly $\F^1(\R)$-consistent score for the mean.
A straightforward generalisation is obvious for multi-objective scores of the form \eqref{eq: S lex}, i.e.,
\begin{equation}\label{eq:app (3.1)}
\mS \colon (\A_1\times \A_2)\times \R^d \to (\R^2, \lex), \qquad \big((\vr_1,\vr_2),\vy\big)\mapsto \big( S_1(\vr_1,\vy), S_{2,\vr_1}(\vr_2,\vy)\big)'.
\end{equation}  

\begin{example}\label{example: lex functionals b}
\begin{enumerate}[(a)]
\item
The variance functional is conditionally elicitable with the mean functional on $\F^4(\R)$. 
A strictly multi-objective $\F^4(\R)$-consistent score $\mS\colon \R^2\times\R\to (\R^2, \lex)$ is given by \eqref{eq:app (3.1)}, where $S_1(r_1,y) = (r_1 - y)^2$ and 
\(
S_{2,r_1}(r_2,y) = \big[r_2 - (y-r_1)^2\big]^2.
\)
Alternatively, to achieve a strictly multi-objective $\F^2(\R)$-consistent score, one can consider the equivalent version $\tilde S_{2,r_1}(r_2,y) = r_2^2 - 2r_2(y-r_1)^2$, or more generally $\tilde S_{2,r_1}(r_2,y) = -\phi(r_2) + \phi'(r_2)[r_2 - (y-r_1)]$, where $\phi$ is strictly convex with subgradient $\phi'$.
\item
\label{item:VaR ES}
For $\a\in(0,1)$, $\ES_\a$ is conditionally elicitable with the $\a$-quantile, $q_\a$, on $\F^2(\R)$.
A strictly multi-objective $\F^2(\R)$-consistent score $\mS\colon \R^2\times\R\to (\R^2, \lex)$ is given via \eqref{eq:app (3.1)}, where $S_1(r_1,y) =  (\one\{y\le r_1\} - \a)(r_1 - y)$ and 
\(
S_{2,r_1}(r_2,y) = \Big\{r_2 - \frac{1}{1-\a}\big[\one\{y> r_1\}y + r_1(\one\{y\le r_1\}-\a)\big] \Big\}^2.
\)
It is straightforward to verify that this family of scores satisfies the condition discussed in Remark \ref{rem:relaxation}. 
Similar to (a), we can achieve strict multi-objective $\F^1(\R)$-consistency upon replacing $S_{2,r_1}$ by $\tilde S_{2,r_1}(r_2,y) = r_2^2 - \frac{2r_2}{1-\a}\big[\one\{y> r_1\}y + r_1(\one\{y\le r_1\}-\a)\big]$, or more generally by $\tilde S_{2,r_1}(r_2,y) = -\phi(r_2) + \phi'(r_2)\Big\{r_2 - \frac{1}{1-\a}\big[\one\{y> r_1\}y + r_1(\one\{y\le r_1\}-\a)\big]\Big\}$, where $\phi$ is strictly convex with subgradient $\phi'$.
\item
\label{item:PI1}
Fix $\a\in(0,1)$.
Motivated by prediction intervals induced by quantiles of the form $[q_\b^-(F),$ $q_{\a+\b}^-(F)]$ for $\b\in[0,1-\a]$ we consider the following form of $\a$-prediction intervals.
Let $T_1\colon \F\to\R$ be elicitable on $\F\subseteq \F_{\textrm{inc}}\subset \F^0(\R)$, where all $F\in \F_{\textrm{inc}}$ are strictly increasing cdfs, and let $T_2\colon \F\to (-\infty,\infty]$ be the right endpoint of the shortest $\a$-prediction interval $[T_1(F), T_2(F)]$ whose left endpoint is given via $T_1$. 
\citet[Proposition 4.10]{FFHR2021} asserts that $(T_1,T_2)$ is not elicitable, subject to smoothness conditions, but \citet[Proposition A.1]{FFHR2021} yields that $T_2$ is conditionally elicitable with $T_1$ on $\F$.
A strictly multi-objective $\F$-consistent scoring function $\mS\colon \R^2\times \R\to (\R^2, \lex)$ is given by \eqref{eq:app (3.1)}, where $S_1(r_1,y)$ is a strictly $\F$-consistent score for $T_1$ and  
\[
S_{2,r_1}(r_2,y) =  \one\{y\ge r_1\}\Big[\big(\one\{y\le r_2\} - \a\big)\big(g(x) - g(y)\big)\Big] - \one\{y<r_1\}\a g(r_2),
\]
where $g\colon(-\infty,\infty]\to\R$ is strictly increasing and $\F$-integrable.
\item
\label{item:PI2}
A different approach to constructing interpretable $\a$-prediction intervals is to specify the midpoint in terms of an elicitable functional, such as the mean or the median.
Again, let $\a\in(0,1)$ and let $T_1\colon \F\to\R$ be elicitable on $\F\subseteq \F_{\textrm{inc}}$ and let $T_2\colon \F\to [0,\infty)$ be half of the length  of the shortest $\a$-prediction interval $[T_1(F) -T_2(F), T_1(F) +T_2(F)]$ with midpoint $T_1$.
\citet[Proposition 4.12]{FFHR2021} asserts that $(T_1,T_2)$ is not elicitable, subject to smoothness conditions, but \citet[Proposition A.3]{FFHR2021} yields that $T_2$ is conditionally elicitable with $T_1$ on $\F$.
A strictly multi-objective $\F$-consistent scoring function $\mS\colon \R^2\times \R\to (\R^2, \lex)$ is given via \eqref{eq:app (3.1)}, where $S_1(r_1,y)$ is a strictly $\F$-consistent score for $T_1$ and  
\[
S_{2,r_1}(r_2,y) =  \big(\one\{|y-r_1|\le r_2\} - \a\big) \big[g(r_2) - g(|y-r_1|)\big],
\]
where $g\colon[0,\infty)\to\R$ is strictly increasing with $g(|\cdot - \,r_1|)$ being $\F$-integrable for all $r_1\in\R$.
\item
For $\a\in(0,1)$, let $\F^0_{(\a)}(\R)$ as in \eqref{eq:notation}.
Define the right $(1-\a)$-tail of a distribution $F\in \F^0_{(\a)}(\R)$ as $T_\a(F)\in\F^0(\R)$ where $T_\a(F)(x) = \one\{x> \VaR_\a(F)\} (F(x) - \a)/(1-\a)$.
Then, $T_\a$ is conditionally elicitable with $\VaR_\a$ on $\F^0_{(\a)}(\R)$. 
A strictly multi-objective $\F^0_{(\a)}(\R)$-consistent score $\mS\colon \R\times \F^0_{(\a)}(\R)\times \R\to(\R^2, \lex)$ is given via \eqref{eq:app (3.1)}, where $S_1(r,y) = (\one\{y\le r\} - \a)[g(r) - g(y)]$ with strictly increasing and bounded $g\colon \R\to\R$, and the second component of $\mS$ is given by
\(
S_{2,r}(F,y) = \one\{y>r\}R(F,y).
\)
Here, $R\colon \F^0(\R)\times\R\to\R$ is a strictly proper scoring rule. This result follows from \citet[Theorem 5]{Gne11a}; see also Theorem \ref{thm:reweighting} and \cite{HolzmannKlar2017} for related approaches.
\end{enumerate}
\end{example}

\subsection{Dimensionality considerations}\label{Dimensionality considerations}

Consider an elicitable functional $\mT\colon \F\to \P(\A)$ with strictly $\F$-consistent score $S\colon \A\times \O \to\R$, where $ \A\subseteq \R^k$. 
Under sufficient smoothness conditions on the score $S(\cdot,\vy)$ and the expected score $\bar S(\cdot, F)$, first order conditions yield that the $k$-dimensional gradient of $S$ is an $\F$-identification function for $\mT$ (which is not necessarily strict).
Under the conditions of Theorem \ref{thm:conditional el}, first order conditions of the corresponding strictly multi-objective $\F$-consistent score \eqref{eq:app (3.1)}
also induce an $\F$-identification function whose dimension coincides with $\mT$. 
To illustrate this, suppose that $\mT_1$ and $\mT_2$ are univariate, such that $k_1 + k_2 = k=2$. 
 We have the scores $\mS((r_1,r_2), \vy) = \big(S_1(r_1,\vy), S_{2,r_1}(r_2,\vy)\big)' \in (\R^2, \lex)$. At the optimum $(t_1,t_2)\in T_1(F) \times T_2(F)$ we obtain the first order condition
$\partial \bar S_1(r_1,F)=0$ and $\partial \bar S_{2,r_1}(r_2,F)=0$.
Hence, if the derivatives exist and tacitly assuming that integration and differentiation commute, $\mV ( (r_1,r_2), \vy) = \big(\partial S_1(r_1,\vy), \partial  S_{2,r_1}(r_2,\vy)\big)'$ constitutes a two-dimensional $\F$-identification function for $(T_1,T_2)$.
This means we can consider the two-dimensional score \eqref{eq:app (3.1)} as a generalised antiderivative of a $k$-dimensional identification function where the symmetry conditions imposed by the Hessian are massively relaxed, which renders the existence of such an object possible. Thus, multi-objective scores provide a means to close the gap between identification and scoring functions for multivariate functionals.

\subsection{More on multi-objective elicitability with respect to the lexicographic order}\label{More on multi-objective elicitability with respect to the lexicographic order}

\begin{rem}\label{rem:comparison of misspecified forecasts}
Resuming with the discussion right before Theorem \ref{thm:conditional el}, the use of the lexicographic order on $\R^2$ allows to compare any forecasts, also misspecified ones. 
As for the classical univariate concept of strict consistency, strong consistency stays silent about the ranking of possibly misspecified forecasts, which is, however, the more realistic scenario \citep{Patton2020}.
For univariate functionals, consistency implies order-sensitivity under mild conditions \citep{BelliniBignozzi2015, Lambert2019}: If two forecasts are both smaller or larger than the true functional value, the one closer to the true value achieves an expected score at most as large as the other forecast.
For multivariate forecasts, there are various generalisations of order-sensitivity \citep[see, e.g.,][]{LambertETAL2008,FisslerZiegel2019}. For multi-objective scores similar order-sensitivity results would be desirable. We suspect that the component\-wise order-sensitivity concept would be particularly promising in that regard.
\end{rem}

\begin{rem}
On the level of the prediction space setting \citep{GneitingRanjan2013}, 
\cite{HolzmannEulert2014} establish that consistent scoring functions are sensitive with respect to increasing information sets. That is, when comparing two ideal forecasts based on nested information sets, the more informed forecast outperforms the less informed one on average.
This is an instance of the more general calibration--resolution principle of \cite{Pohle2020}, showing that minimising a consistent scoring function amounts to ``jointly maximizing information content and minimizing systematic mistakes''.
Using the same arguments as in \cite{HolzmannEulert2014} and \cite{Pohle2020} (basically exploiting the tower property of conditional expectations), one can establish a similar principle for multi-objective consistent scores mapping to $(\R^2, \lex)$.
\end{rem}

\subsection{The necessity of the Convex Level Sets property}
\label{subsec:CxLS}

Here, we revisit a powerful necessary condition for identifiability and elicitability, namely the Convex Level Sets (CxLS) property, for multi-objective scores. We use the same notation as in Section~\ref{Scoring functions, identification functions and (conditional) elicitability} of the main paper, in particular we let $\F\subseteq \F'\subseteq \F^0(\O)$.
Following the terminology of \citet[Section 3]{FFHR2021} we say that a functional $\mT \colon \F'\to\P(\A)$, satisfies the \emph{selective CxLS property} on $\F$ if 
for all $F_0,F_1\in\F$ and for all $\lambda\in(0,1)$ such that $F_\lambda:= (1-\lambda)F_0 + \lambda F_1\in \F$ we have
\(
\mT(F_0)\cap \mT(F_1) \subseteq \mT(F_\lambda).
\)
$\mT$ satisfies the \emph{selective CxLS* property} on $\F$ if 
for all $F_0,F_1\in\F$ and for all $\lambda\in(0,1)$ such that $F_\lambda:= (1-\lambda)F_0 + \lambda F_1\in \F$ 
$ \mT(F_0)\cap \mT(F_1) \neq \emptyset$ implies that $\mT(F_0)\cap \mT(F_1) =\mT(F_\lambda)$.
We shall frequently omit the term ``selective'' and will just speak of the CxLS and the CxLS* property.
Clearly, the CxLS* property implies the CxLS property. If $\mT$ attains singletons only on $\F$, then the two properties coincide. They then take the familiar form that
$\mT(F_0)= \mT(F_1)$ implies $\mT(F_0) =\mT(F_\lambda)$.
It is well known that the CxLS property is necessary for identifiability (which we state for the sake of completeness).
Proposition 3.4 in \cite{FFHR2021} establishes that the CxLS* property is necessary for univariate elicitability. The following generalises this result to strong multi-objective elicitability.

\begin{prop}\label{prop:CxLS}
Consider a functional $\mT\colon \F'\to\P(\A)$. 
\begin{enumerate}[\rm (i)]
\item
If $\mT$ is identifiable on $\F\subseteq\F'$, it satisfies the CxLS property on $\F$.
\item
If $\mT$ is strongly multi-objective elicitable on $\F\subseteq\F'$, it satisfies the CxLS* property on $\F$.
\end{enumerate}
\end{prop}

\begin{proof}[{\textbf{Proof of Proposition \ref{prop:CxLS}:}}]
For (i), let $\mV\colon\A\times \O\to \R^m$ be a strict identification function. 
Let $F_0,F_1\in\F$ such that $F_\lambda\in\F$ for some $\lambda \in(0,1)$. For $\vt\in \mT(F_0)\cap \mT(F_1)$ it holds that
\(
\vzero = (1-\lambda)\bar \mV(\vt,F_0) + \lambda \bar \mV(\vt,F_1) = \bar \mV(\vt,F_\lambda).
\)
Due to the strictness of $\mV$, we have $\vt\in \mT(F_\lambda)$.\\
For (ii), the proof follows along the lines of the proof of Proposition 3.4 in \cite{FFHR2021}.
With the same set-up as above, consider some strictly strongly multi-objective $\F$-consistent scoring function $\mS\colon\A\times \O \to (\R^{\mathcal I}, \lorder)$. For $i\in\{0,1\}$, the strict strong multi-objective $\F$-consistency of $\mS$ implies that for any $\vr\in\A$ we have
$\bar \mS(\vt,F_i) - \bar \mS(\vr,F_i) =\vzero$ for $\vr \in \mT(F_i)$ and for $\vr \notin \mT(F_i)$ the score difference is $\slorder \vzero$.
Hence, $\bar \mS(\vt,F_\lambda) - \bar \mS(\vr,F_\lambda)$ equals
\(
(1-\lambda) \big(\bar \mS(\vt,F_0) - \bar \mS(\vr,F_0) \big) + \lambda \big(\bar \mS(\vt,F_1) - \bar \mS(\vr,F_1) \big),
\)
which is $\vzero$ for $\vr \in \mT(F_0) \cap \mT(F_1)$ and otherwise $\slorder \vzero$.
\end{proof}

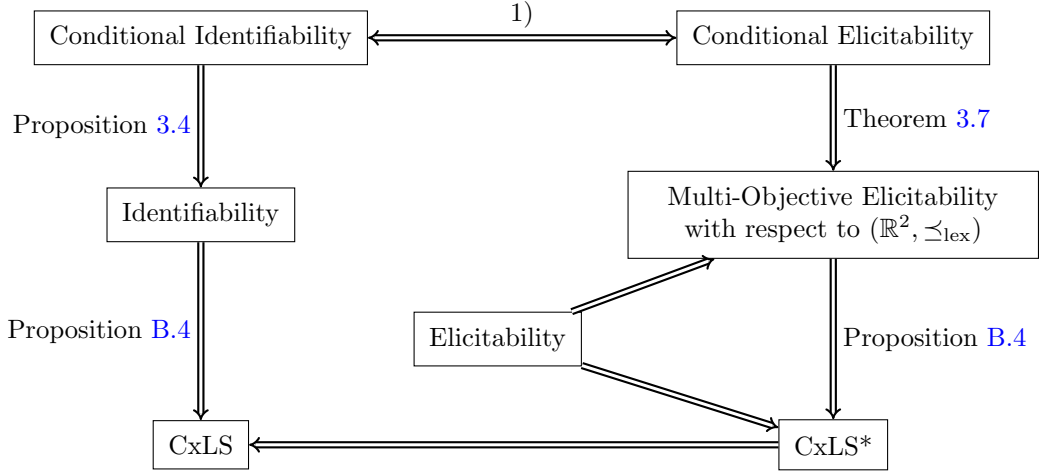
\begin{figure}
\centering
\small
\begin{tikzpicture}[every edge/.style={imp}]
  \matrix[nodes={inner sep=2mm},
  row sep=0.7cm,column sep=0.6cm] {
    \node (c-id) [rectangle, draw] {Conditional Identifiability}; &
																									&
    \node (c-el) [rectangle, draw] {Conditional Elicitability}; \\ \\
		\node (id)   [rectangle, draw] {Identifiability}; &
																									&
    \node (mo-el) [rectangle, draw, text width=5cm, align=center] {Multi-Objective Elicitability with respect to $(\mathbb{R}^2, \lex)$}; \\
		& 
		\node (el)  [rectangle, draw] {Elicitability}; &
		&\\
		\node (cx)  [rectangle, draw] {CxLS}; &
																									&
    \node (cxS) [rectangle, draw] {CxLS*}; \\
  };
  \path
	(c-id) edge[implies-implies,double equal sign distance] node[above] {1)} (c-el)
	(el) edge[imp] (cxS)
	(c-el) edge[imp] node[right] {Theorem~\ref{thm:conditional el}} (mo-el)
	(c-id) edge[imp] node[left] {Proposition~\ref{prop:conditional id}} (id)
	(mo-el) edge[imp] node[right] {Proposition~\ref{prop:CxLS}} (cxS)
	(cxS) edge[imp] (cx)
	(id) edge[imp] node[left] {Proposition~\ref{prop:CxLS}} (cx)
	(el) edge[imp] (mo-el)
  ;
\end{tikzpicture}
\caption{Illustration of the most important structural results for functionals $\mT_1:\mathcal{F}\rightarrow\mathcal{P}(\mathsf{A}_1)$ and $\mT_2:\mathcal{F}\rightarrow\mathcal{P}(\mathsf{A}_2)$. The equivalence 1) holds under some regularity conditions if $\mT_1$ and $\mT_2$ are point-valued mappings to $\mathbb{R}$ \citep{SteinwartPasinETAL2014}.}
\label{fig:overview meth app}
\end{figure}

Proposition \ref{prop:CxLS} implies that strong multi-objective elicitability shares an important feature with the usual univariate notion of elicitability. 
Showing that the CxLS* is violated is the standard procedure to rule out elicitability. Consequently, functionals violating the CxLS* property such as variance or ES also fail to be strongly multi-objective elicitable. This applies in particular to scoring functions mapping to $(\R^2, \lex)$, since $\lex$ is a total order.

On the other hand, weak multi-objective elicitability appears to be much more flexible than its strong counterpart. 
Indeed, the following consideration shows that the CxLS property is not necessary for weak multi-objective elicitability:
Suppose we have a strictly weakly multi-objective $\F$-consistent scoring function $\mS$ for $\mT$ mapping to $(\R^2, \lorder)$, where $\lorder$ is the componentwise order. Suppose further that there are $F_0,F_1\in\F$ such that $F_\lambda = (1-\lambda)F_0 + \lambda F_1\in \F$ for all $\lambda\in(0,1)$. Assume there is some $\vt\in \mT(F_0)\cap \mT(F_1)$ and some $\vr \notin \mT(F_0)\cup \mT(F_1)$
such that $\bar \mS(\vr,F_0) - \bar \mS(\vt,F_0) = (-5,1)'\nprec(0,0)'$ and $\bar \mS(\vr,F_1) - \bar \mS(\vt,F_1) = (1,-5)'\nprec(0,0)'$.  
Then, invoking the linearity of the integral, $\bar \mS(\vr,F_{1/2}) - \bar \mS(\vt,F_{1/2}) = (-2,-2)'\slorder (0,0)'$.
Therefore, $\vt\notin \mT(F_{1/2})$ and the CxLS property is violated (and thus also the CxLS* property).

A further investigation of weak multi-objective elicitability therefore seems to be an interesting field. Such an investigation will be deferred to future work.

As a complete summary of the results of this section and Section~\ref{Scoring functions, identification functions and (conditional) elicitability}, Figure~\ref{fig:overview meth app} complements Figure~\ref{fig:overview meth} in the main paper by including the implications of Proposition~\ref{prop:CxLS}.

\section{Conditional identifiability and conditional elicitability results}
\label{app:cond id and el}

Here, we show that $\CoVaR_{\a|\b}(Y|X)$, $\MES_\b(Y|X)$ and $(\CoVaR_{\a|\b}(Y|X), \CoES_{\a|\b}(Y|X) )$ are all conditionally elicitable and conditionally identifiable with $\VaR_\b(X)$, subject to mild assumptions on the corresponding class of bivariate distributions.

\begin{rem}\label{rem:erw}
Our definitions of the systemic risk measures in Section~\ref{CoVaR, CoES and MES} rely on the conditional cdf $F_{Y\mid X \geq \VaR_\b(X)}$; see \eqref{eq:CoVaR}.
If $F_X$ is continuous, the event $\{X\geq \VaR_{\b}(X)\}$ has probability $1-\b$, such that conditioning on the right tail is akin to the classical definition of expected shortfall at level $\b$, as $\ES_\b(X) = \E[X\,|\,X\geq \VaR_{\b}(X)]$. 
However, if $X=\VaR_\b(X)$ with positive probability, then $\p\{X\geq \VaR_{\b}(X)\}>1-\b$, such that $\{X\geq \VaR_{\b}(X)\}$ loses its interpretation as a tail event. (E.g., for $X$ with $\p\{X=0\} = 1-\p\{X=1\} = p<1$ it holds that $\VaR_\b(X) =1$ for all $\b\in[p,1]$. Hence, $\p\{X\ge \VaR_\b(X)\} = 1-p >1-\b$ if $p<\b$.)
To allow for point masses of $F_X$ at $\VaR_\b(X)$, one may use a similar correction term as for $\ES_\b$; see \citet[Proposition 8.15]{MFE15}.
Specifically, one may adapt the definition of $F_{Y\mid X \geq \VaR_\b(X)}$ as follows. 
For $\b\in[0,1)$ and $y\in\R$ we define the cdf
\begin{multline}\label{eq:new F}
F_{Y\mid X \succcurlyeq \VaR_\b(X)}(y) 
:=
\frac{1}{1-\b}\Big[\p\big\{Y\le y, X > \VaR_\b(X)\big\} \\  
\qquad + \p\big\{Y\le y\mid X = \VaR_\b(X)\big\}\big(1-\b-\p\big\{X>\VaR_\b(X)\big\}\big)\Big].
\end{multline}
If $\p\big\{X>\VaR_\b(X)\big\} = 1-\b$, the second summand vanishes and the first summand simplifies to $F_{Y|X\ge\VaR_\b(X)}(y)$. 
If $\p\big\{X>\VaR_\b(X)\big\} < 1-\b$, then $\p\big\{X=\VaR_\b(X)\big\} >0$, such that the conditional probability $\p\big\{Y\le y\mid X = \VaR_\b(X)\big\}$ can be defined elementarily.
Note that $F_{Y\mid X \succcurlyeq\VaR_0(X)} = F_Y$.
In particular, one may replace $F_{Y\mid X \geq \VaR_\b(X)}$ in \eqref{eq:CoVaR} with $F_{Y\mid X \succcurlyeq \VaR_\b(X)}(y)$ to obtain alternative definitions of $\CoVaR_{\a|\b}(Y|X)$, $\CoES_{\a|\b}(Y|X)$ and $\MES_\b(Y|X)$. All results of Section~\ref{app:cond id and el}, Section~\ref{app:negative results} and Section~\ref{Structural results for CoVaR, CoES and MES} go through unaltered with this change. In fact, we shall prove all these results with the systemic risk measures defined in terms of $F_{Y\mid X \succcurlyeq \VaR_\b(X)}(y)$. We mention that the definitions coincide when $F_X$ is continuous, because then $F_{Y|X\ge\VaR_\b(X)}(\cdot)\equiv F_{Y\mid X \succcurlyeq \VaR_\b(X)}(\cdot)$ for $\b\in[0,1)$.
\end{rem}

We start with a description of the core rationale behind the conditional identifiability and elicitability results.
To this end, first recall a powerful result about weighted scoring and identification functions.

\begin{thm}[\cite{Gne11a} Theorem 5; \cite{Fissler2017} Proposition 2.3.5]\label{thm:reweighting}
Let $\F\subseteq \F' \subseteq  \F^0(\R^d)$ and let $\mV\colon\A\times \R^d\to\R^m$ and $S\colon \A\times\R^d\to\R$ be a strict $\F$-identification function and a strictly $\F$-consistent scoring function for $\mT\colon\F'\to \P(\A)$. Let $w\colon \R^d\to[0,\infty)$ be a non-negative measurable weight function. For any $F\in \F^0(\R^d)$ such that $\bar w(F)\in(0,\infty)$ define the probability measure
\begin{equation*}
F^{(w)}(\diff \vx) := \frac{w(\vx) F(\diff \vx)}{\bar w(F)}, \qquad \vx \in\R^d.
\end{equation*}
Define the class $\F^{(w)} = \{F\in\F^{0}(\R^d)\colon F^{(w)}\in \F\}$ and the functional $\mT^{(w)}\colon \F^{(w)} \to \P(\A)$, $F\mapsto  \mT^{(w)}(F):= \mT(F^{(w)})$. Then
$\mV^{(w)}(\vx,\vy) := w(\vx) \mV(\vx,\vy)$ a strict $\F^{(w)}$-identification function and 
$S^{(w)}(\vx,\vy) := w(\vx) S(\vx,\vy)$
is a strictly $\F^{(w)}$-consistent scoring function for $\mT^{(w)}$.
\end{thm}

Secondly, consider a class of bivariate distributions $F_{X,Y}\in\F^0(\R^2)$, where $F_{X,Y}\mapsto \VaR_\b(F_X)$ is elicitable and identifiable. This is, e.g., the case for $\F_{(\b)}^0(\R^2)$ defined in \eqref{eq:notation}.
Conveniently, for any $F_{X,Y}\in \F_{(\b)}^0(\R^2)$, the distribution $F_{Y| X\succcurlyeq \VaR_\b(X)}$ defined in \eqref{eq:new F} reduces to $F_{Y|X>\VaR_\b(X)}$. 
Crucially, this conditional distribution can be retrieved as the marginal distribution
\[
F_{Y}^{(w)}(\diff y):= \int_\R F_{X,Y}^{(w)}(\diff x, \diff y), \qquad y\in\R
\]
with weight function $w(x,y) = \one\{x>\VaR_\b(F_X)\}$. Hence, Theorem \ref{thm:reweighting} along with well-known results about the elicitability and identifiability of $\VaR_\a$, the mean \citep{Gne11}, and those for $(\VaR_\a, \ES_\a)$ \citep{FZ16a}
implies that $\CoVaR_{\a|\b}(Y|X)$, $\MES_\b(Y|X)$ and $ (\CoVaR_{\a|\b}(Y|X), \CoES_{\a|\b}(Y|X) )$ are identifiable and elicitable on a suitable restriction of the class
$\F_{(\b),v} := \big\{F_{X,Y} \in \F_{(\b)}^0(\R^2) \colon \VaR_\b(F_X) = v\big\}$, where $v\in\R$.

\section{Negative results about joint elicitability}
\label{app:negative results}

This section shows that the pairs $ (\VaR_\b(X), \CoVaR_{\a|\b}(Y|X) )$, $ (\VaR_\b(X), \MES_\b(Y|X) )$ and the triplet $ (\VaR_\b(X), \CoVaR_{\a|\b}(Y|X), \CoES_{\a|\b}(Y|X) )$ fail to be elicitable, subject to regularity conditions, despite the fact that they are identifiable and thus also satisfy the CxLS property.
This negative result is in stark contrast to the fact that the unconditional counterparts,
$\VaR_{\a}(Y)$, $\E(Y)$ and $ (\VaR_\a(Y), \ES_\a(Y) )$, are all elicitable.

The reason for this intriguing finding lies in integrability conditions: If the functionals were elicitable and admitted strictly consistent scoring functions, which are twice continuously differentiable in expectation for sufficiently smooth distributions, their Hessian would need to be symmetric. Osband's Principle \cite[Theorem 3.2]{FZ16a} asserts that, under some richness conditions on the underlying class of distributions $\F$, for any consistent score $S(\vr,y)$ and any strict identification function $\mV(\vr,y)$, there exists a matrix-valued function $\vh(\vr)$ such that 
$\nabla \bar S(\vr,F) = \vh(\vr) \bar \mV(\vr, F)$ for all $\vr$ and for all $F\in\F$.
Using the strict identification functions from Theorem~\ref{thm:joint id} it is possible to show that such a function $\vh(\vr)$ cannot exist for the three functionals of interest here.

This is an instance of a more general phenomenon: 
\cite{DFZ2020} argue that there is generally a \emph{gap} between the classes of strictly consistent scoring functions and strict identification functions for multivariate functionals, which is due to these integrability conditions. The following example illustrates this.

\begin{example}\label{footnote:double quantile}
Consider the double quantile functional $\mT = (q_\a, q_\b)$, $0<\a<\b<1$ with the standard identification function $\mV((r_1,r_2),y) = (\one\{y\le r_1\} - \a, \one\{y\le r_2\} - \b)'$. Then the matrix-valued function $\vh$ induced by a consistent scoring function is necessarily diagonal with non-negative diagonal entries such that $h_{11}$ is only a function of $r_1$ and $h_{22}$ is only a function of $r_2$. Put differently, any consistent scoring function for the double quantile is necessarily the sum of two consistent scoring functions for the individual quantiles. On the other hand, it is easy to see that for any matrix-valued function $\vh$ such that $\vh(\vr)$ has full rank everywhere, the product $\vh(\vr)\mV(\vr,y)$ is a strict identification function for $\mT$. Clearly, the latter class is considerably larger.
\end{example}

In the case at hand, this gap is extreme: While there exist strict identification functions, the respective classes of strict scoring functions are empty.
We are only aware of one more instance in the literature of such an extreme gap \cite[Subsection 4.3]{FFHR2021}.

As mentioned, we need to impose certain richness conditions along with smoothness conditions on the class of bivariate distributions $\F\subset \F^0(\R^2)$. 
Then, the negative results hold on any class containing such an $\F$. 
We state the particular conditions needed for $\F$ explicitly, noting that the class of bivariate normal distributions (excluding perfect correlation) along with all mixtures always satisfies these conditions.

\begin{assumption}
\label{ass:G}
\begin{enumerate}[\rm (a)]
\item\label{ass:11}
For any $F_{X,Y}\in\F$ with marginals $F_X$ and $F_Y$, the distribution $F_XF_Y$ is in $\F$, having the same marginal distributions, but independent components.
\item\label{ass:12}
$\F$ is convex.
\item\label{ass:13}
All $F_{X,Y}\in\F$ are continuously differentiable and admit a strictly positive continuous density $f_{X,Y}$.
\end{enumerate}
\end{assumption}

In the following, let $\conv(M)$ denote the convex hull of some set $M$, and $\interior(M)$ its interior.

\begin{prop}\label{prop:CoVaR}
Let $\a,\b\in(0,1)$ and denote by $\mT$ the functional mapping $F_{X,Y}$ to $ (\VaR_\b(F_X),$ $\CoVaR_{\a|\b}(F_{X,Y}) )$.
Let $\F\subseteq \F^0(\R^2)$ satisfy Assumption \ref{ass:G} and suppose that 
\begin{enumerate}[\rm (i)]
\item\label{ass:15}
$\A := \mT(\F) = \{\mT(F)\colon F\in\F\} \subseteq \R^2$ has a non-empty and simply connected interior $ \interior(\A)$.
\item\label{ass:16}
For $\mV = \mV^{(\VaR, \CoVaR)}$ defined in \eqref{eq:V_CoVaR} it holds that for all $(v,c)\in\interior(\A)$ there are $F_1,F_2,F_3\in\F$ such that
\(
\vzero\in \interior\big(\conv\big\{\bar \mV\big((v,c), F_1\big), 
\bar \mV\big((v,c), F_2\big),
\bar \mV\big((v,c), F_3\big)\big\}\big).
\)
\item\label{ass:17}
For all $F_{X,Y}\in\F$ and for all $(v,c)\in\interior(\A)$ there is some $\tilde F_{X,Y}\in\F$ such that 
	$F_{Y}(c) = \tilde F_{Y}(c)$, 
	$F_{X}(v) = \tilde F_{X}(v)$,
	$f_{Y}(c) = \tilde f_{Y}(c)$,
	$f_{X}(v) \neq \tilde f_{X}(v)$.
\end{enumerate}
Then, for any $\F\subseteq \F'\subseteq\F^0(\R^2)$, there is no strictly $\F$-consistent scoring function $S\colon \A\times \R^2\to\R$ for $\mT$, such that the expected score $\bar S(\cdot,F)\colon\A\to\R$ is twice partially differentiable with continuous second order derivatives for all $F\in\F$.
\end{prop}

It is worth emphasising that it is the \emph{expected} score in Proposition \ref{prop:CoVaR} which is assumed to be twice differentiable with continuous second order derivatives, and not necessarily the score itself.
E.g., while the usual score for $\VaR_\a$, $S(x,y) = (\one\{y\le x\} - \a)(x-y)$ clearly fails to even be differentiable everywhere, the expected score $\bar S(x,F)$ is twice continuously differentiable if $F$ is continuously differentiable, with $\bar S'(x,F) = F(x) - \a$ and $\bar S''(x,F) = F'(x) = f(x)$.
Therefore, requiring this smoothness condition on the expected score seems to be reasonable on $\F$ of the form considered in the proposition.
The other two negative results are of a similar form.

\begin{prop}\label{prop:CoES}
Let $\a,\b\in(0,1)$ and denote by $\mT$ the functional mapping $F_{X,Y}$ to $ (\VaR_\b(F_X),$ $\CoVaR_{\a|\b}(F_{X,Y}), \CoES_{\a|\b}(F_{X,Y}) )$.
Let $\F\subseteq \F^0(\R^2)$ satisfy Assumption \ref{ass:G} and suppose that 
\begin{enumerate}[\rm (i)]
\item\label{ass:35}
$\A := \mT(\F) = \{\mT(F)\colon F\in\F\} \subseteq \R^3$ has a non-empty and simply connected interior $ \interior(\A)$.
\item\label{ass:36}
For $\mV = \mV^{(\VaR, \CoVaR, \CoES)}$ defined in \eqref{eq:V_CoES} it holds that for all $(v,c,\mu)\in\interior(\A)$ there are $F_1,\ldots,F_4\in\F$ such that 
\(
\vzero\in \interior\big(\conv\big\{\bar \mV\big((v,c,\mu), F_1\big), \ldots, 
\bar \mV\big((v,c,\mu), F_4\big)\big\}\big).
\)
\item\label{ass:37}
For all $F_{X,Y}\in\F$ and for all $(v,c,\mu)\in\interior(\A)$ there are some $\tilde F_{X,Y}, \hat F_{X,Y}\in\F$ such that 
	$F_{Y}(c) = \tilde F_{Y}(c) = \hat F_{Y}(c)$, 
	$ \int_c^\infty y f_Y(y)\diff y= \int_c^\infty y \tilde f_Y(y)\diff y =  \int_c^\infty y \hat f_Y(y)\diff y$, 
	$F_{X}(v) = \tilde F_{X}(v) = \hat F_{X}(v)$,
	$f_{Y}(c) = \tilde f_{Y}(c) \neq \hat f_{Y}(c)$,
	$f_{X}(v) \neq \tilde f_{X}(v) = \hat f_{X}(v)$.
\end{enumerate}
Then, for any $\F\subseteq \F'\subseteq\F^0(\R^2)$, there is no strictly $\F'$-consistent scoring function $S\colon \A\times \R^2\to\R$ for $\mT$, such that the expected score $\bar S(\cdot,F)\colon\A\to\R$ is twice partially differentiable with continuous second order derivatives for all $F\in\F$.
\end{prop}

\begin{prop}\label{prop:MES}
Let $\a,\b\in(0,1)$ and denote by $\mT$ the functional mapping $F_{X,Y}$ to $ (\VaR_\b(F_X),$ $\MES_{\b}(F_{X,Y}) )$.
Let $\F\subseteq \F^0(\R^2)$ satisfy Assumption \ref{ass:G} and suppose that 
\begin{enumerate}[\rm (i)]
\item\label{ass:25}
$\A := \mT(\F) = \{\mT(F)\colon F\in\F\} \subseteq \R^2$ has a non-empty and simply connected interior $ \interior(\A)$.
\item\label{ass:26}
For $\mV = \mV^{(\VaR, \MES)}$ defined in \eqref{eq:V_MES} it holds that for all $(v,\mu)\in\interior(\A)$ there are $F_1,F_2,F_3\in\F$ such that $
\vzero\in \interior\big(\conv\big\{\bar \mV\big((v,\mu), F_1\big), 
\bar \mV\big((v,\mu), F_2\big),
\bar \mV\big((v,\mu), F_3\big)\big\}\big).$
\item\label{ass:27}
For all $F_{X,Y}\in\F$ and for all $(v,\mu)\in\interior(\A)$ there is some $\tilde F_{X,Y}\in\F$ such that 
	$\E(F_Y)= \E(\tilde F_{Y})$, 
	$F_{X}(v) = \tilde F_{X}(v)$,
	$f_{X}(v) \neq \tilde f_{X}(v)$.
\end{enumerate}
Then, for any $\F\subseteq \F'\subseteq\F^0(\R^2)$, there is no strictly $\F$-consistent scoring function $S\colon \A\times \R^2\to\R$ for $\mT$, such that the expected score $\bar S(\cdot,F)\colon\A\to\R$ is twice partially differentiable with continuous second order derivatives for all $F\in\F$.
\end{prop}

\section{Proofs}\label{app:Proofs}

We first prove Proposition~\ref{cor:negative result} via Proposition~\ref{prop:CxLS}. Specifically, we show that $\CoVaR_{\a|\b}$, $\CoES_{\a|\b}$ and $\MES_{\b}$ violate the CxLS property on any class $\F\subseteq \F^0(\R^2)$ containing all bivariate normal distributions along with their mixtures. Recall from Section~\ref{app:cond id and el} that we take the systemic risk measures to be defined in terms of $F_{Y\mid X \succcurlyeq \VaR_\b(X)}$ from \eqref{eq:new F} here.

\begin{rem}\label{rem:CxLS1}
Before proving Proposition~\ref{cor:negative result}, we outline the core rationale why the systemic risk measures of interest violate the CxLS property on a class $\F$ under very mild conditions.
Suppose $\F$ contains two bivariate distributions $F^0, F^1$ of the form 
$F^0(x,y) = C\big(F_1^0(x), F_2(y)\big)$, $F^1(x,y) = C\big(F_1^1(x), F_2(y)\big)$.
Here, $C\colon [0,1]^2\to[0,1]$ is a copula and $F_1^0, F_1^1, F_2$ are univariate marginal distributions.
Let $(X^0, Y^0)\sim F^0$ and $(X^1,Y^1)\sim F^1$. 
If $F_1^0$ and $F_1^1$ are continuous, then
\[
F_{Y^0\mid X^0 \succcurlyeq \VaR_\b(X^0)}(y)
= F_{Y^1\mid X^1 \succcurlyeq \VaR_\b(X^1)}(y)
= \Big[F_2(y) - C\big(\b, F_2(y)\big)\Big]/(1-\b).
\]
Hence, clearly all functionals of $F_{Y^0\mid X^0 \succcurlyeq \VaR_\b(X^0)}$ and $F_{Y^1\mid X^1 \succcurlyeq \VaR_\b(X^1)}$ coincide.
Next, consider the convex combination $F^\lambda(x,y):= (1-\lambda)F^0(x,y) + \lambda F^1(x,y)$ and let $(X^\lambda, Y^\lambda)\sim F^\lambda$ where $\lambda\in (0,1)$.
To simplify the argument, suppose that $F_1^0$ and $F_1^1$ are strictly increasing (otherwise, the argument can be adapted).
If $\VaR_\b(X^0)<\VaR_\b(X^1)$, then $\VaR_\b(X^0)<\VaR_\b(X^\lambda)<\VaR_\b(X^1)$. Hence, $\b^0 := F_1^0( \VaR_\b(X^\lambda)) > \b >F_1^1( \VaR_\b(X^\lambda)) =: \b^1$.
We obtain
\[
F_{Y^\lambda\mid X^\lambda \succcurlyeq \VaR_\b(X^\lambda)}(y)
= (1-\lambda)\frac{F_2(y) -  C\big(\beta^0, F_2(y)\big)}{1-\b^0}
+ \lambda\frac{F_2(y) -  C\big(\beta^1, F_2(y)\big)}{1-\b^1}.
\]
If $C$ is not the independence copula, then different choices of $u_1\in(0,1)$ generally result in different mappings $[0,1]\ni u_2 \mapsto \big[u_2 - C(u_1,u_2)\big]/(1-u_1)$. Therefore, $F_{Y^\lambda\mid X^\lambda \succcurlyeq \VaR_\b(X^\lambda)}$ will in general be different from $F_{Y^0\mid X^0 \succcurlyeq \VaR_\b(X^0)}$.
Hence, a violation of the CxLS property can be shown upon choosing appropriate $F_1^0$, $F_1^1$ and $\lambda \in(0,1)$.
\end{rem}

\begin{proof}[{\textbf{Proof of Proposition~\ref{cor:negative result}:}}]
By Proposition~\ref{prop:CxLS}, it suffices to show that $\CoVaR_{\a|\b}$, $\CoES_{\a|\b}$ and $\MES_{\b}$ violate the CxLS property on any class $\F\subseteq \F^0(\R^2)$ containing all bivariate normal distributions along with their mixtures.

To task this, consider $X, Y\sim \mathcal N(0,1)$ which are jointly Gaussian and have a correlation of $\rho\neq0$.
Let $(X^0, Y^0) = (X-1,Y)$ and $(X^1,Y^1)= (X+1,Y)$ with distribution functions $F^0$ and $F^1$. Moreover, let $F^{0.5} = 0.5 F^0 + 0.5 F^1$ and let $(X^{0.5}, Y^{0.5})$ have distribution $F^{0.5}$.
Clearly $\VaR_\b(X^0) = \Phi^{-1}(\b)-1$ and $\VaR_\b(X^1) = \Phi^{-1}(\b)+1$, where $\Phi$ is the distribution function of a standard normal and $\phi$ is its density function. 
For $\VaR_\b(X^{0.5})$, there is generally no closed form solution.
We therefore illustrate the result numerically for $\b = 0.95$. Here, $\VaR_{0.95}(X^{0.5}) \approx 2.28$, while $\Phi^{-1}(0.95)\approx 1.64$.
It holds that 
\begin{align*}
F_{Y^0\mid X^0 \succcurlyeq \VaR_\b(X^0)}(y)
&= F_{Y^1\mid X^1 \succcurlyeq \VaR_\b(X^1)}(y)
=\p(Y\le y\mid X>\Phi^{-1}(\b)\big),\\
F_{Y^{0.5}\mid X^{0.5} \succcurlyeq \VaR_\b(X^{0.5})}(y)
&= 0.5\p\big\{Y\le y\mid X-1>\VaR_\b(X^{0.5})\big\} \\
& \hspace{0.5cm} + 0.5\p\big\{Y\le y\mid X+1>\VaR_\b(X^{0.5})\big\}.
\end{align*}
We exploit the fact that $\E[Y\,|\, X > z] = \rho \phi(z)/[1-\Phi(z)]$. Therefore,
$\MES_\b(Y^0|X^0) = \MES_\b(Y^1|X^1) = \E[Y\,|\, X > \Phi^{-1}(\b)] = \rho \phi\big(\Phi^{-1}(\b)\big)/[1-\b]$. 
On the other hand,
\[
\MES_\b(Y^{0.5}|X^{0.5}) = 0.5 \rho \left[  \frac{\phi\big(\VaR_\b(X^{0.5}) +1\big)}{1-\Phi\big(\VaR_\b(X^{0.5}) +1\big)} + \frac{\phi \big(\VaR_\b(X^{0.5}) -1\big)}{1-\Phi\big(\VaR_\b(X^{0.5}) -1\big)} \right].
\]
These values do not coincide for $\b\in(0,1)$ and $\rho\neq 0$. E.g., $\MES_{0.95}(Y^0|X^0) = 2.06\rho$ and $\MES_{0.95}(Y^{0.5}|X^{0.5}) = 2.65\rho$.

For $\CoVaR_{\a|\b}$, it is sufficient to check that $\VaR_\a(F_{Y^{0.5}\mid X^{0.5} \succcurlyeq \VaR_\b(X^{0.5})})$ does not coincide with $\VaR_\a(F_{Y^0\mid X^0 \succcurlyeq \VaR_\b(X^0)})$.
This can only be solved numerically and we illustrate again the case for $\a=\b=0.95$, setting $\rho=0.8$, emphasising that other choices for $\a,\b\in(0,1)$ and $\rho\neq0$ also hold similarly.
We get $\VaR_{0.95}(F_{Y^0\mid X^0 \succcurlyeq \VaR_{0.95}(X^0)}) \approx 2.77$ and 
$\VaR_{0.95}(F_{Y^{0.5}\mid X^{0.5} \succcurlyeq \VaR_{0.95}(X^{0.5})}) \approx 3.20$.

The results for $\CoES_{\a|\b}$ follow from the fact that $\ES_\a$ itself violates the CxLS-property. We mention that $ (\CoVaR_{\a|\b}, \CoES_{\a|\b} )$ also violates the CxLS-property, which can be shown similarly as for $\CoVaR_{\a|\b}$.
\end{proof}

Next, we prove Theorem~\ref{thm:DM}. The technical assumptions of Theorem~\ref{thm:DM} are only required to ensure a multivariate central limit theorem for the score differences $\{\vd_t\}$
\begin{equation}\label{eq:mvclt}
	\mOmega_n^{-1/2}\sqrt{n}\overline{\vd}_n\overset{d}{\longrightarrow}N(\vzero,\mI_{2\times2}),\qquad\text{as }n\to\infty,
\end{equation}
and consistent long-run variance estimation
\begin{equation}\label{eq:clrve}
	\|\widehat{\mOmega}_n-\mOmega_n\|\overset{\p}{\longrightarrow}\vzero,\qquad\text{as }n\to\infty.
\end{equation}
Thus, Assumption~\ref{ass:DM} may be replaced by any other conditions ensuring \eqref{eq:mvclt} and \eqref{eq:clrve}. Some of the most general conditions can be found in \citet[Section~24.4]{Dav94} for \eqref{eq:mvclt} and \citet[Theorem~2.1]{DD00} for \eqref{eq:clrve}. While our conditions are slightly less general, they are much simpler.

\begin{assumption}\label{ass:DM}
\label{ass:B}
\begin{enumerate}

	\item[B1:] There is some $\Delta_{\vd}<\infty$ such that $\E|\vd_t^\prime\vd_t|^{r}\leq\Delta_{\vd}$ for all $t\ge1$, where $r>2$.
	
	\item[B2:] $\{\vd_t\}$ is $\alpha$-mixing of size $-2r/(r-2)$ or $\phi$-mixing of size $-r/(r-1)$.
	
		\item[B3:] The sequence of integers $m_n$ satisfies $m_n\rightarrow\infty$ and $m_n=o(n^{1/4})$, as $n\to\infty$.
		
	\item[B4:] There is some $\Delta_{w}<\infty$ such that $|w_{n,h}|\leq\Delta_{w}$ for all $n\in\mathbb{N}$ and $h\in\{1,\ldots,m_n\}$, and $w_{n,h}\rightarrow1$, as $n\to\infty$, for all $h=1,\ldots,m_n$.
\end{enumerate}
\end{assumption}

\begin{proof}[{\textbf{Proof of Theorem~\ref{thm:DM}:}}]
The proof is similar to that of Theorem~4 in \citet{GW06}, so we only sketch it here. A Cram\'{e}r--Wold device and Theorem~5.20 in \citet{Whi01} ensure that 
$\mOmega_n^{-1/2}\sqrt{n}\overline{\vd}_n\overset{d}{\longrightarrow}N(\vzero,\mI_{2\times2})$, as $n\to\infty$,
under B1 and B2. The fact that $\|\widehat{\mOmega}_n-\mOmega_n\|\overset{\p}{\longrightarrow}\vzero$, as $n\to\infty$, follows from Theorem~6.20 in \citet{Whi01} under B1--B4.
\end{proof}

\begin{figure}
	\centering
		\includegraphics[width=\textwidth]{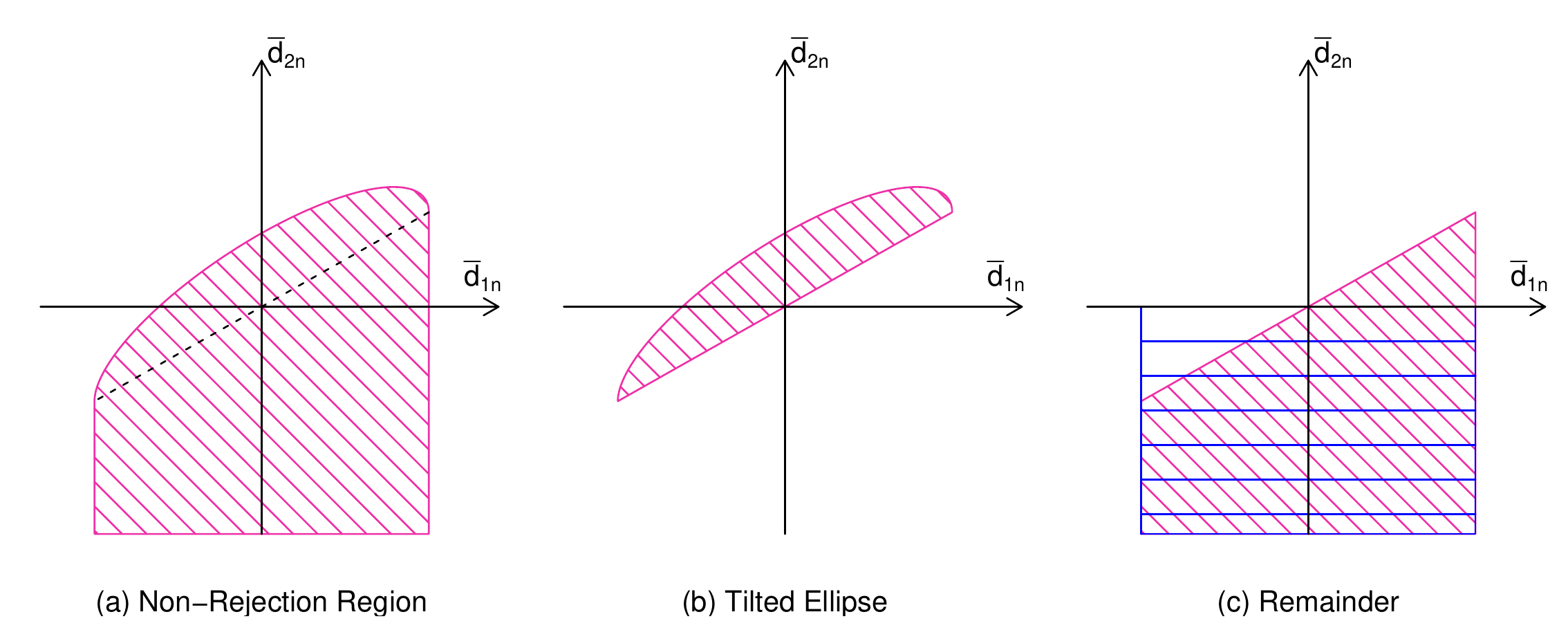}
	\caption{Decomposition of Non-Rejection Region.}
	\label{fig:Figure5}
\end{figure}

\begin{proof}[{\textbf{Proof of Proposition~\ref{prop:OS}:}}]
By construction of the non-rejection region (displayed in Figure~\ref{fig:Figure5}~(a)), we have
\begin{align}
&\sup_{c\leq0}\lim_{n\to\infty}\p\Big\{\mathcal{T}_n^{\OS}> \chi_{2,1-\widetilde{\nu}}^2\ \Big\vert\ H_0^{(c)}\ \text{holds}\Big\} 
=\lim_{n\to\infty}\p\Big\{\mathcal{T}_n^{\OS}> \chi_{2,1-\widetilde{\nu}}^2\ \Big\vert\ H_0^{(0)}\ \text{holds}\Big\}\notag\\
	&=\lim_{n\to\infty}\p\Big\{\mathcal{T}_n^{\OS}> \chi_{2,1-\widetilde{\nu}}^2\ \Big\vert\ \E[\overline{d}_{1n}]=0\ \text{and}\ \E[\overline{d}_{2n}]=0 \Big\}
	= 1-\p\Big\{\mathcal{T}^{\OS} \le \chi_{2,1-\widetilde{\nu}}^2 \Big\},
	\label{eq:as prob bd}
\end{align}
where, for $\mOmega=( (\sigma_{11}, \sigma_{12})^\prime, (\sigma_{12}, \sigma_{22})^\prime)$,
\begin{equation*}
	\mathcal{T}^{\OS}=\Big(Z_1, \max\big\{Z_2, ({\sigma}_{12}/{\sigma}_{11})Z_1\big\}\Big){\mOmega}^{-1}
	\begin{pmatrix}
	Z_1\\ 
	\max\big\{Z_2, ({\sigma}_{12}/{\sigma}_{11})Z_1\big\}
	\end{pmatrix},
\end{equation*}
with $(Z_1,Z_2)'\sim N(\vzero,\mOmega)$, which is the distributional limit of $\sqrt{n}(\overline{d}_{1n}, \overline{d}_{2n})^\prime$.
By construction, $\p\{\mathcal{T}^{\OS} \le \chi_{2,1-\widetilde{\nu}}^2 \}$ corresponds to the probability that $(Z_1,Z_2)'$
falls into the non-rejection region sketched in Figure~\ref{fig:Figure5}~(a). 
This region admits a disjoint decomposition into the upper half of the tilted ellipse (displayed in panel~(b)), corresponding to $Z_2 > (\sigma_{12}/\sigma_{11})Z_1$ and the remainder (displayed in panel~(c)), corresponding to $Z_2 \le (\sigma_{12}/\sigma_{11})Z_1$. 
Since the full ellipse has probability content $1-\widetilde{\nu}$, we get 
\[
	\p\Big\{(Z_1,Z_2)^\prime\ \text{falls into upper half of tilted ellipse}\Big\}=\frac{1}{2}(1-\widetilde{\nu}).
\]
By symmetry of the normal distribution, the probability content of the remainder equals that of the area shaded in blue in panel~(c). Thus,
\begin{align*}
	\p\Big\{(Z_1,Z_2)^\prime\ \text{falls into remainder}\Big\}&=\p\Big\{|Z_1|\leq e_{1,\max},\ Z_2\leq0\Big\}
	=\frac{1}{2}\p\Big\{|Z_1|\leq e_{1,\max}\Big\},
\end{align*}
where $e_{1,\max}$ is the maximal extension of the tilted ellipse in the horizontal direction. For the ellipse
$\{\vz\in\R^2\colon \vz'{\mOmega}^{-1}\vz \le \chi_{2,1-\widetilde{\nu}}^2\}$,
this is $e_{1,\max}=\sqrt{{\sigma}_{11} \chi_{2,1-\widetilde{\nu}}^2}$. Thus,
\begin{align*}
	\p\Big\{(Z_1,Z_2)^\prime\ \text{falls into remainder}\Big\} 
	&= \frac{1}{2}\p\Big\{{\sigma}_{11}^{-1/2}|Z_1|\leq \sqrt{\chi_{2,1-\widetilde{\nu}}^2}\Big\}
	= \frac{1}{2}F_{\chi_1^{2}}(\chi_{2,1-\widetilde{\nu}}^2).
\end{align*}
We end up with
\(
\p\big\{\mathcal{T}^{\OS} \le \chi_{2,1-\widetilde{\nu}}^2 \big\}
= \frac{1}{2}\big(1-\widetilde{\nu} + F_{\chi_1^{2}}(\chi_{2,1-\widetilde{\nu}}^2)\big).
\)
Combining this with \eqref{eq:as prob bd}, the conclusion follows.
\end{proof}

\begin{proof}[{\textbf{Proof of Proposition~\ref{prop:CoVaR}:}}]
The proof exploits Osband's Principle in \citet[Theorem 3.2]{FZ16a}, which originates from the seminal work of \cite{Osb85}. Suppose there is an $\F$-consistent scoring function $S\colon \A\times \R^2\to\R$, whose expectation $\bar{S}:\A\to\R$ is twice partially differentiable with continuous second-order derivatives for all $F\in\F$. Using Assumptions \ref{ass:G} and conditions \eqref{ass:15} and \eqref{ass:16} of Proposition~\ref{prop:CoVaR}, \citet[Theorem 3.2]{FZ16a} shows that there is a  differentiable matrix-valued function $\vh\colon\interior(\A)\to\R^{2\times2}$ with components $h_{ij}$, $i,j=1,2$, such that for all $(v,c)\in\interior(\A)$ and for all $F_{X,Y}\in\F$
\be{eq:Osband}
\nabla \bar S\big((v,c),F_{X,Y}\big) = \vh(v,c) \bar \mV\big((v,c),F_{X,Y}\big).
\ee
The symmetry of the Hessian $\nabla^2\bar S\big((v,c),F_{X,Y}\big)$ \cite[Corollary 3.3]{FZ16a} then yields that 
\be{eq:symmetry}
\partial_c\partial_v \bar S\big((v,c),F_{X,Y}\big) = \partial_v\partial_c \bar S\big((v,c),F_{X,Y}\big)
\ee 
for all $(v,c)\in\interior(\A)$ and for all $F_{X,Y}\in\F$. 
Denote by $\F^0(\R)^{\otimes 2}$ the family of independent bivariate distributions. Then, straightforward calculations yield that for $F_{X,Y}\in\F\cap \F^0(\R)^{\otimes 2}$, it holds that
\begin{align*}
\bar V_1\big((v,c),F_{X,Y}\big) &= F_X(v) - \b,\\
\bar V_2\big((v,c),F_{X,Y}\big) &= (F_{Y}(c) - \a)(1-F_X(v)),\\
\partial_c \bar V_1\big((v,c),F_{X,Y}\big) &= 0, \\
\partial_v \bar V_1\big((v,c),F_{X,Y}\big) &= f_X(v),\\
\partial_c \bar V_2\big((v,c),F_{X,Y}\big) &= f_{Y}(c)(1-F_X(v)),\\
\partial_v \bar V_2\big((v,c),F_{X,Y}\big) &= -f_X(v)(F_{Y}(c) - \a).
\end{align*}

Let $(v^*,c^*)= \mT(F_{X,Y})$.
Using \eqref{eq:Osband}, the chain rule, and the fact that the expected identification function vanishes at the true report, we get
\begin{equation*}
 h_{12}(v^*,c^*)\partial_c\bar V_2\big((v^*,c^*),F_{X,Y}\big)
 =\partial_c\partial_v \bar S\big((v^*,c^*),F_{X,Y}\big).
\end{equation*}
Similarly,
\begin{equation*}
\partial_v\partial_c \bar S\big((v^*,c^*),F_{X,Y}\big)
= h_{21}(v^*,c^*) \partial_v\bar V_1\big((v^*,c^*),F_{X,Y}\big)
+ h_{22}(v^*,c^*) \partial_v\bar V_2\big((v^*,c^*),F_{X,Y}\big).
\end{equation*}
For $F_{X,Y}\in\F\cap \F^0(\R)^{\otimes 2}$, \eqref{eq:symmetry} yields that 
\be{eq:symmetry2}
h_{12}(v^*,c^*) f_Y(c^*) ( 1-\b) = h_{21}(v^*,c^*)f_X(v^*).
\ee
Exploiting part \eqref{ass:17} of Proposition \ref{prop:CoVaR} we get that $h_{21}(v^*,c^*) =0$. Since $f_{Y}(c^*)>0$, also $h_{12}(v^*,c^*) =0$. Due to the surjectivity of $\mT$ onto $\A$, we obtain that for all $(v,c)\in\interior(\A)$,
\(
h_{12}(v,c)= h_{21}(v,c)=0.
\)
Now we exploit \eqref{eq:symmetry} for a general point $(v,c)\in\interior(\A)$, such that $F_X(v) \neq \b$ and $F_Y(c) \neq \a$. For $F_{X,Y}\in\F\cap \F^0(\R)^{\otimes 2}$ we obtain
\(
\partial_ch_{11}(v,c) (F_X(v) - \b) 
= (F_{Y}(c)- \a) \big[\partial_v h_{22}(v,c) (1-F_X(v)) - h_{22}(v,c) f_X(v)\big].
\)
Again invoking condition \eqref{ass:17}, we get that $h_{22}(v,c)=0$. Since we can perform the argument for any $(v,c)\in\interior(\A)$, making use of the surjectivity, it holds that 
\(
h_{22}\equiv \partial_v h_{22} \equiv0.
\)
Therefore, also $\partial_c h_{11}\equiv0$. In conclusion, the matrix function $\vh$ 
is a diagonal matrix and the only non vanishing entry is $h_{11}(v,c) = g(v)$
for some function $g\colon\R\to\R$. Hence, using \eqref{eq:Osband} we obtain that 
\(
\partial_c \bar S\big((v,c),F_{X,Y}\big) =0
\)
for all $(v,c)\in\interior(\A)$. Since $\interior(\A)$ is open, non-empty and simply connected, the expected score is constant in $c$. This shows that $S$ cannot be strictly $\F$-consistent.
\end{proof}

\begin{proof}[{\textbf{Proof of Proposition~\ref{prop:CoES}:}}]
The proof basically follows along the lines of the proof of Proposition \ref{prop:CoVaR}.
So we get the counterpart of \eqref{eq:Osband}
\be{eq:Osband3}
\nabla \bar S\big((v,c,e),F_{X,Y}\big) = \vh(v,c,e) \bar \mV\big((v,c,e),F_{X,Y}\big),
\ee
where $\vh(v,c,e) \in\R^{3\times 3}$.
For $F_{X,Y}\in\F\cap \F^0(\R)^{\otimes 2}$ we get that
\begin{align*}
\bar V_1\big((v,c,e),F_{X,Y}\big) & = F_X(v) - \b,\\
\bar V_2\big((v,c,e),F_{X,Y}\big) & = (1-F_X(v))(F_{Y}(c) - \a),\\
\bar V_3\big((v,c,e),F_{X,Y}\big) & = (1-F_X(v))\Big[e- \frac{1}{1-\a}\Big(\int_c^\infty y f_Y(y)\diff y + c(F_Y(c)-\a)\Big)\Big],
\end{align*}
and
\begin{align*}
\partial_v \bar V_1\big((v,c,e),F_{X,Y}\big) &= f_X(v),\\
\partial_c \bar V_1\big((v,c,e),F_{X,Y}\big) &= \partial_e \bar V_1\big((v,c,e),F_{X,Y}\big) =0,\\
\partial_v\bar V_2\big((v,c,e),F_{X,Y}\big)  &= -f_X(v)(F_Y(c)-\a),\\
\partial_c\bar V_2\big((v,c,e),F_{X,Y}\big)  &= f_Y(c)(1-F_X(v)),\\
\partial_e\bar V_2\big((v,c,e),F_{X,Y}\big)  &= 0,\\
\partial_v\bar V_3\big((v,c,e),F_{X,Y}\big) &= f_X(v)\Big[e- \frac{1}{1-\a}\Big(\int_c^\infty y f_Y(y)\diff y + c(F_Y(c)-\a)\Big)\Big],\\
\partial_c\bar V_3\big((v,c,e),F_{X,Y}\big) &= -(1-F_X(v))(F_Y(c)-\a)/(1-\a),\\
\partial_e\bar V_3\big((v,c,e),F_{X,Y}\big) &= (1-F_X(v)).
\end{align*}
The symmetry of the Hessian implies that for all $(v,c,\mu)\in\interior(\A)$ and for all $F_{X,Y}\in\F\cap \F^0(\R)^{\otimes 2}$,
\begin{align}\label{eq:sym1}
\partial_c\partial_v \bar S((v,c,e), F_{X,Y}) = \partial_v\partial_c \bar S((v,c,e), F_{X,Y}), \\ 
\label{eq:sym2}
\partial_e\partial_v \bar S((v,c,e), F_{X,Y}) = \partial_v\partial_e \bar S((v,c,e), F_{X,Y}), \\ 
\label{eq:sym3}
\partial_e\partial_c \bar S((v,c,e), F_{X,Y}) = \partial_c\partial_e \bar S((v,c,e), F_{X,Y}).
\end{align}
Using \eqref{eq:Osband3} to evaluate \eqref{eq:sym1} at $(v^*,c^*,e^*)= \mT(F_{X,Y})\in \interior(\A)$ implies the identity
\begin{equation*}
h_{12}(v^*,c^*,e^*)f_Y(c^*)(1-\b) = h_{21}(v^*,c^*,e^*) f_X(v^*).
\end{equation*}
Using condition \eqref{ass:37} of Proposition \ref{prop:CoES} together with the positivity of $f_Y$ and with the surjectivity of $\mT$ implies that 
\(
h_{12}\equiv h_{21} \equiv 0.
\)
Using \eqref{eq:Osband3} to evaluate \eqref{eq:sym2} at $(v^*,c^*,e^*)= \mT(F_{X,Y})\in \interior(\A)$ implies that
\(
h_{13}(v^*,c^*,e^*)(1-\b) = h_{31}(v^*,c^*,e^*)f_X(v^*).
\) 
Using the same arguments, we obtain that 
\(
h_{13}\equiv h_{31} \equiv 0.
\)
Finally, evaluating \eqref{eq:sym3} with \eqref{eq:Osband3} at 
$(v^*,c^*,e^*)= \mT(F_{X,Y})\in \interior(\A)$ yields that 
\(
h_{23}(v^*,c^*,e^*)(1-\b)= h_{32}(v^*,c^*,e^*)f_Y(c^*)(1-\b). 
\)
Hence, invoking the surjectivity of $\mT$ and condition \eqref{ass:37}, we get that 
\(
h_{23}\equiv h_{32} \equiv 0.
\)
We already know that $\vh$ is a diagonal matrix. \eqref{eq:sym1} therefore simplifies to
\(
\partial_c h_{11}(v,c,e)(F_X(v) - \b) = (F_Y(c) - \a)\big[\partial_v h_{22}(v,c,e)(1-F_X(v)) - h_{22}(v,c,e)f_X(v)\big].
\)
Evaluating this for points such that $F_X(v)\neq \beta$ and $F_Y(c)\neq \a$, we can again exploit condition \eqref{ass:37} of Proposition \ref{prop:CoES} and the surjectivity of $\mT$ to obtain
\(
h_{22} \equiv  \partial_v h_{22} \equiv \partial_c h_{11}\equiv 0.
\)
Now, we can evaluate \eqref{eq:sym3} for some point $(v,c,e)\in \interior(\A)$ such that $F_Y(c) \neq\a$ and $e\neq \frac{1}{1-\a}\left(\int_c^\infty y f_Y(y)\diff y + c(F_Y(c)-\a)\right)$ to obtain
\(
0 = \big[e- \frac{1}{1-\a}\big(\int_c^\infty y f_Y(y)\diff y + c(F_Y(c)-\a)\big)\big]\times
\Big[\big(1-F_X(v)\big)\partial_v h_{33}(v,c,e) +f_X(v) h_{33}(v,c,e)\Big].
\)

Again invoking condition \eqref{ass:37} and the surjectivity, we get that 
\(
h_{33}\equiv0.
\)
Finally, \eqref{eq:sym2} simplifies to
\(
\partial_e h_{11}(v,c,e)(F_X(v) - \b) = 0,
\)
which implies that
\(
\partial_e h_{11} \equiv0.
\)
In conclusion, the only non-vanishing component of $\vh$ is $h_{11}(v,c,e) = g(v)$ for some function $g\colon\R\to\R$. 
With the same arguments as in the proof of Proposition \ref{prop:CoVaR} we deduce that $(v,c,e)\mapsto \bar S\big((v,c,e), F_{X,Y}\big)$ is constant in $c$ and $e$. Therefore, $S$ cannot be strictly $\F$-consistent.
\end{proof}

\begin{proof}[{\textbf{Proof of Proposition~\ref{prop:MES}:}}]
The proof first follows along the lines of the proof of Proposition~\ref{prop:CoVaR}
 up to equation \eqref{eq:symmetry}, \textit{mutatis mutandis}.
For $F_{X,Y}\in \F\cap \F^0(\R)^{\otimes 2}$, we get
\begin{align*}
\bar V_1\big((v,\mu),F_{X,Y}\big) &= F_X(v) - \b,\\
\bar V_2\big((v,\mu),F_{X,Y}\big) &= (1-F_X(v))\big(\mu - \E(F_{Y})\big),\\
\partial_v \bar V_1\big((v,\mu),F_{X,Y}\big) &= f_X(v),\\
\partial_\mu \bar V_1\big((v,\mu),F_{X,Y}\big) &= 0, \\
\partial_v \bar V_2\big((v,\mu),F_{X,Y}\big) &= -f_X(v)\big(\mu -\E(F_{Y})\big),\\
\partial_\mu \bar V_2\big((v,\mu),F_{X,Y}\big) &= (1-F_X(v)).
\end{align*}

Let $(v^*,\mu^*) = \mT(F_{X,Y})$.
Since $\partial_\mu\partial_v \bar S((v^*, \mu^*),F_{X,Y}) = \partial_v\partial_\mu \bar S\big((v^*,\mu^*),F_{X,Y}\big)$ by the symmetry of the Hessian, we obtain
\(
h_{12}(v^*,\mu^*) (1-\b) = h_{21}(v^*,\mu^*)f_X(v^*).
\)
Exploiting condition \eqref{ass:27} of Proposition \ref{prop:MES} and the surjectivity of $\mT$ we get that 
\(
h_{12}\equiv h_{21} \equiv0.
\)
Now, we exploit the symmetry of the Hessian at a general point $(v,\mu)\in\interior(\A)$ such that $F_X(v)\neq \b$ and $\E(F_Y)\neq \mu$. We obtain
\begin{align*}
\partial_\mu\partial_v \bar S((v, \mu),F_{X,Y}) &= \partial_\mu h_{11}(v,\mu) \big(F_X(v)-\b\big),\\
\partial_v\partial_\mu \bar S((v, \mu),F_{X,Y}) &=\big(\mu - \E(F_{Y})\big)
\big[\partial_vh_{22}(v,\mu)(1-F_X(v)) - h_{22}(v,\mu)f_X(v)\big].
\end{align*}
Again invoking \eqref{ass:27} of Proposition \ref{prop:MES} and the surjectivity of $\mT$ we get that  
\(
h_{22}\equiv \partial_v h_{22} \equiv \partial_c h_{11}\equiv0.
\)
Therefore, the only non-vanishing component of $\vh$ is $h_{11}(v,\mu) = g(v)$ for some function $g\colon\R\to\R$. 
Concluding as in the proof of Proposition \ref{prop:CoVaR}, $S$ fails to be strictly $\F$-consistent for $\mT$.
\end{proof}

\section{Monte Carlo simulations}\label{app:Monte Carlo Simulations}

Here, we investigate the finite-sample performance of our tests of $H_0^{=}$ and $H_0^{\lex}$. For distinct VaR forecasts, we use $\mathcal{T}_n$ (defined in \eqref{eq:T_n}) to test $H_0^{=}$, and $\mathcal{T}_n^{\OS}$ (defined in \eqref{eq:equiv NR}) to test $H_0^{\lex}$. For identical VaR forecasts, we use $\mathcal{T}_{2n}=\sqrt{n}\overline{d}_{2n}/\widehat{\sigma}_{22,n}^{1/2}$ to test both hypotheses.
We do so for $(\VaR,\CoVaR)$ and $(\VaR,\CoVaR,\CoES)$ forecasts generated from a model of asset returns that may reasonably approximate actual stock market dynamics. We leave out the pair $(\VaR, \MES)$ for brevity, because the results are qualitatively similar. Throughout, we consider one-step-ahead forecasts. For one-step-ahead forecasts, \citet{DM95}, \citet{GW06} and others recommend to use $m_n=0$ for $\widehat{\mOmega}_n$ in \eqref{eq:hat Omega_n} (where an empty sum is defined to be zero), implicitly assuming that $\{\vd_t\}$ is uncorrelated. 
We follow their suggestion here.

\subsection{Data generating process}\label{DGP}

We simulate $\{(X_t,Y_t)\}_{t=1,\ldots,n}$ from a bivariate GARCH model with GAS-driven $t$-copula. 
In particular,
we use GARCH(1,1) marginals
\begin{equation}\label{eq:GARCH}
\begin{split}
X_t &= \sigma_{x,t}\varepsilon_{x,t},\qquad \sigma_{x,t}^2 = \omega_x+\alpha_xX_{t-1}^2+\beta_{x}\sigma_{x,t-1}^2,\\
Y_t &= \sigma_{y,t}\varepsilon_{y,t},\qquad \sigma_{y,t}^2 = \omega_y+\alpha_yY_{t-1}^2+\beta_{y}\sigma_{y,t-1}^2,
\end{split}
\end{equation}
where $\omega_z>0$, $\alpha_z\geq0$, $\beta_z\geq0$ ($z\in\{x, y\}$). The innovations $(\varepsilon_{x,t}, \varepsilon_{y,t})$ are identically distributed with $\varepsilon_{x,t}\sim N(0,1)$ and $\varepsilon_{y,t}$ having a standardised Student's $t_{5}$-distribution.

For the dependence structure, consider the probability integral transforms (PITs) $\mU_t=\big(F_x(\varepsilon_{x,t}), F_y(\varepsilon_{y,t})\big)$, where $F_x$ and $F_y$ denote the cdfs of $\varepsilon_{x,t}$ and $\varepsilon_{y,t}$, respectively. Define $\mathfrak{F}_t=\sigma\big((X_t,Y_t), (X_{t-1},Y_{t-1}), \ldots\big)$. Then, we model the conditional distribution of $\mU_t\mid \mathfrak{F}_{t-1}$ with a $t$-copula with density $c(\,\cdot\,; \vartheta, \rho_t)\colon [0,1]^2\to[0,\infty)$, where $\vartheta>0$ is the (constant) degrees-of-freedom parameter, and $\rho_t\in(-1,1)$ is the time-varying correlation parameter. Letting the correlation vary over time while keeping the degrees of freedom constant is standard in the literature \citep{DP15,BC19}. To restrict $\rho_t$ to $(-1,1)$, we model the real-valued parameter $f_t$ with a GAS dynamic
and then set $\rho_t := \Delta(f_t)$ where $\Delta(x)=[1-\exp(-x)]/[1+\exp(x)]\in(-1,1)$. 
In particular, following \citet{CKL13}, we set
\begin{equation}\label{eq:GAS}
	f_{t}=\omega^{\dagger}+\alpha^{\dagger}s_{t-1}+\beta^{\dagger}f_{t-1}, \qquad 
	s_{t-1}:=\frac{\diff}{\diff f_{t-1}}\log c\big(\mU_{t-1}; \vartheta,\Delta(f_{t-1})\big).
\end{equation}

We choose the empirically plausible values $(\omega^{\dagger},\ \alpha^{\dagger},\ \beta^{\dagger},\ \vartheta)=(0.001,\ 0.1,\ 0.99,\ 5)$ for the dependence parameters, and $(\omega_x,\ \alpha_x,\ \beta_{x})=(\omega_y,\ \alpha_y,\ \beta_{y})=(0.001,\ 0.2,\ 0.79)$ for the marginal parameters \citep[see, e.g.,][]{DP15,Hog20a+}.

\subsection{Risk forecasts}\label{Risk Forecasts}

In practice, the most popular approach to compute (risk) forecasts is a moving-window approach. In this case, a window that is rolled through the sample is used as the basis for frequent (often daily) re-estimation of the model producing the risk forecasts (e.g., a bivariate GARCH with GAS-driven $t$-copula, or a simple white noise model). We adopt such a rolling-window scheme in the empirical application in Section~\ref{Empirical Application}. However, since here we perform 10\,000 replications for the simulation results to have low standard errors, such an approach would be computationally infeasible. Instead, we use a fixed-window approach here. That is, we obtain model parameter estimates from some `in-sample' period $\{(X_t,Y_t)\}_{t=-r+1,\ldots,0}$, and use these estimates to produce `out-of-sample' risk forecasts for $\{(X_t,Y_t)\}_{t=1,\ldots,n}$.

We describe how we do this next. First, we generate a trajectory $\{(X_t,Y_t)\}_{t=-r+1,\ldots,n}$ from the bivariate GARCH in \eqref{eq:GARCH} with GAS-driven $t$-copula. Then, based on the `in-sample' data $\{(X_t,Y_t)\}_{t=-r+1,\ldots,0}$, we estimate the marginal parameters $(\omega_x,\alpha_x,\beta_x)$ and $(\omega_y,\alpha_y,\beta_y)$ via standard Gaussian quasi-maximum likelihood estimation \citep{FZ04}. In particular, we obtain the estimated conditional variances $\widehat{\sigma}_{x,t}^2=\widehat{\omega}_x+\widehat{\alpha}_x X_{t-1}^2+\widehat{\beta}_{x}\widehat{\sigma}_{x,t-1}^2$ (and similarly for $\widehat{\sigma}_{y,t}^2$) and the standardized residuals $\widehat{\varepsilon}_{x,t}=X_t/\widehat{\sigma}_{x,t}$ and $\widehat{\varepsilon}_{y,t}=Y_t/\widehat{\sigma}_{y,t}$. 
By doing so, we also compute the estimated PITs $\widehat{\mU}_t=\big(\widehat{F}_x(\widehat{\varepsilon}_{x,t}), \widehat{F}_y(\widehat{\varepsilon}_{y,t})\big)$, where $\widehat{F}_x$ and $\widehat{F}_y$ are the empirical cdfs of the $\widehat{\varepsilon}_{x,t}$ and $\widehat{\varepsilon}_{y,t}$, respectively. Based on the $\widehat{\mU}_t$, we estimate $(\omega^{\dagger},\alpha^{\dagger},\beta^{\dagger},\vartheta)^\prime$ via maximum likelihood, as proposed by \citet{CKL13}. 

In line with the empirical application in the main paper, we forecast the conditional risk measures $\VaR_{t}(X_t) = \VaR_{\b}(F_{X_t\mid\mathfrak F_{t-1}})$, 
$\CoVaR_{t}(Y_t\vert X_t) = \CoVaR_{\a|\b}(F_{(X_t,Y_t) \mid \mathfrak F_{t-1}})$, and
$\CoES_{t}(Y_t\vert X_t) = \CoES_{\a|\b}(F_{(X_t,Y_t) \mid \mathfrak F_{t-1}})$. Here, $F_{(X_t,Y_t) \mid \mathfrak F_{t-1}}(x,y) = \p\{X_t\le x, Y_t\le y \mid \mathfrak F_{t-1}\} =: \p_{t-1}\{X_t\le x, Y_t\le y\}$ for $x,y\in\R$, and $\mathfrak{F}_{t-1} = \sigma\big((X_{t-1},Y_{t-1}), (X_{t-2}, Y_{t-2}), \ldots\big)$. To ease notation, we suppress the dependence of the risk measures on the risk levels $\alpha$ and $\beta$, which we set to be $\alpha=\beta=0.95$.

Now, we describe how we calculate the `out-of-sample' (conditional) risk forecasts for the observations $\{(X_t, Y_t)\}_{t=1,\ldots,n}$. We use the estimated parameters from the `in-sample' period to compute the (one-step-ahead) Value-at-Risk for the $X_t$ via
\begin{equation*}
	\widehat{\VaR}_t := \widehat{\VaR}_t(X_t) = \widehat{\sigma}_{x,t}\widehat{\VaR}_{x,\varepsilon},\qquad t=1,\ldots,n,
\end{equation*}
where $\widehat{\VaR}_{x,\varepsilon}$ is the empirical $\beta$-quantile of the $\{\widehat{\varepsilon}_{x,t}\}_{t=-r+1,\ldots,0}$. 
In terms of the conditional distribution, CoVaR solves $\p_{t-1}\{Y_{t}>\CoVaR_{t}(Y_t\vert X_t)\mid X_{t}\geq\VaR_{t}(X_t)\}=1-\alpha$. Using this and the fact that $\p_{t-1}\{X_{t}\geq\VaR_{t}(X_t)\}=1-\beta$, elementary calculations yield that $\CoVaR_{t}(Y_t\vert X_t)$ is implicitly defined by the equation
\begin{equation}\label{eq:CoVaR forecast}
	(1-\alpha)(1-\beta) = \int_{F_y(\CoVaR_{t}(Y_t\vert X_t)/\sigma_{y,t})}^{1}\int_{\beta}^{1}c\big((u_1,u_2), \vartheta, \rho_t\big)\D u_1\D u_2.
\end{equation}
Plugging in estimates (e.g., replacing $\vartheta$ with the `in-sample' estimate $\widehat{\vartheta}$, replacing $F_y$ with the `in-sample' empirical cdf $\widehat{F}_y$, etc.) and numerically solving \eqref{eq:CoVaR forecast} for $\CoVaR_{t}(Y_t\vert X_t)$, yields the forecasts $\widehat{\CoVaR}_t$ for $t=1,\ldots,n$. From a discretization of the integral in formula \eqref{eq:CoES}, we also obtain CoES predictions $\widehat{\CoES}_t$ for $t=1,\ldots,n$. We collect these (correctly specified) VaR and systemic risk forecasts in the vector $\widehat{\vr}_{t}=(\widehat{r}_t^{\VaR}, \widehat{\vr}_{t}^{\SR})=\big(\widehat{\VaR}_{t},\widehat{\CoVaR}_{t}, \widehat{\CoES}_{t}\big)$.

Clearly, it is very challenging to come up with two forecasts from distinct models, such that $H_{0}^{=}$ is satisfied for our data-generating process; see also \cite{ZhuTimmermann2020}. Hence, for simplicity, we confound the correctly specified forecasts $\widehat{\vr}_t$ by two multiplicative noises of equal magnitude to simulate under $H_0^{=}$. We use (strictly positive) multiplicative noise instead of additive noise, because this ensures that positive risk forecasts remain positive after confounding them. This is important, since the 0-homogeneous loss functions we use require positive (systemic) risk forecasts; see also Section~\ref{Description of Tests}. The noises may be thought of as uninformative (multiplicative) predictors. Specifically, consider $\widehat{\vr}_{t,(i)}=(\widehat{r}_{t,(i)}^{\VaR}, \widehat{\vr}_{t,(i)}^{\SR})=(\widehat{\VaR}_{t,(i)},\widehat{\CoVaR}_{t,(i)}, \widehat{\CoES}_{t,(i)})$ ($i=1,2$), where $\widehat{\VaR}_{t,(i)}=\widehat{\VaR}_{t}\cdot\epsilon_{t,(i)}^{\VaR}$, $\widehat{\CoVaR}_{t,(i)}=\widehat{\CoVaR}_{t}\cdot\epsilon_{t,(i)}^{\CoVaR}$ and $\widehat{\CoES}_{t,(i)}=\widehat{\CoES}_{t}\cdot\epsilon_{t,(i)}^{\CoES}$. Thus, the forecasts from the correctly specified model are contaminated by some mutually independent, positive multiplicative noises $\{\epsilon_{t,(i)}^{\VaR}\}_{t=1,\ldots,n}$, $\{\epsilon_{t,(i)}^{\CoVaR}\}_{t=1,\ldots,n}$ and $\{\epsilon_{t,(i)}^{\CoES}\}_{t=1,\ldots,n}$, each assumed to be serially independent and Weibull-distributed. Hence, they have common density $f(x)=(k/\lambda)(x/\lambda)^{k-1}e^{-(x/\lambda)^{k}}$ for $x>0$, where we choose $k=10$ for the shape parameter and $\lambda=0.3$ for the scale parameter. This results in a mean of $\lambda\Gamma(1+1/k)\approx 0.285$, so that the $\widehat{\vr}_{t,(i)}$ are seriously misspecified on average. Clearly, $H_0^{=}$ (and also $H_0^{\lex}$) holds for $\widehat{\vr}_{t,(1)}$ and $\widehat{\vr}_{t,(2)}$.

\subsection{Description of tests}\label{Description of Tests}

In evaluating the risk forecasts, we choose multi-objective scores with components that render the score differences 0-homogeneous. That is, the score differences remain unchanged when scaling all forecasts and observations with a positive constant. 
This leads to `unit-consistent' comparisons and, quite often, higher power of DM tests \citep{NZ17,Tay17,PZC19}. We also found in unreported simulations that 0-homogeneous score differences lead to DM tests with higher power than other degrees of homogeneity, e.g., 1-homogeneous scores. This may be explained as follows. Our data and forecasts are conditionally heteroscedastic. Invoking the 0-homogeneity, the loss differences $\vd_t$ then are homoscedastic. Thus, we obtain smaller `standard errors' in $\widehat{\mOmega}_n$ by using 0-homogenous instead of, say, 1-homogeneous score differences. This effect is likely not outweighed by the increased average score differences of 1-homogenous vis-\`{a}-vis 0-homogenous scores. (Recall that 1-homogeneous scores are more sensitive to outliers.) If this is indeed the case, then 0-homogenous scores give tests higher power. When using 0-homogeneous score differences, one has to assume additionally that risk forecasts are positive. However, this assumption is typically innocuous.

To achieve 0-homogeneity in the first component of our multi-objective score, we set $h(z)=\log (z)$ and $a^{\VaR}(x,y)=\log(x)$ in the $S^{\VaR}$-function from Theorem~\ref{thm:mo el} \citep[Example~4]{NZ17}. When comparing $(\VaR, \CoVaR)$ forecasts, we choose $g(z)=\log (z)$, $a(y)=\log(y)$ and $a^{\CoVaR}(x,y)=0$ in $S_v^{\CoVaR}$ from \eqref{eq:S_CoVaR}, leading to
\begin{equation}\label{eq:S CoVaR appl}
	S_v^{\CoVaR}\big(c,(x,y)\big) 
			= \one\{x>v\} \Big[\big(\one\{y\le c\} - \a\big) \log(c) +\one\{y\ge c\}\log(y)\Big].
\end{equation} 
When also $\CoES$ forecasts are of interest, we use $S_v^{(\CoVaR, \CoES)}$ from \eqref{eq:S_CoES} with $g(z)=0$, $a(y)=a^{\CoES}(x,y)=0$,
and $\phi(z)=-\log (z)$ ($z>0$). 
This choice leads to
\begin{equation}\label{eq:S CoES appl}
	S_v^{(\CoVaR, \CoES)}\big((c,e),(x,y)\big)=\frac{\one\{x>v\}}{1-\alpha}\Bigg[\one\{y>c\}\frac{y-c}{e}+(1-\alpha)\Big(\frac{c}{e}-1+\log (e)\Big)\Bigg].
\end{equation} 

To test $H_0^{=}$ and $H_0^{\lex}$, we simulate $R=10\,000$ trajectories $\{(X_t,Y_t)\}_{t=-r+1,\ldots,n}$ from the bivariate GARCH model \eqref{eq:GARCH} with GAS-driven $t$-copula. 
We fix $r=1000$ and let $n\in\{500, 1000\}$. For each trajectory, we fit the model, forecast VaR and systemic risk, and compute the score differences. Then, depending on the forecasts, we carry out the formal tests based on $\mathcal{T}_n$ (for a test of $H_0^{=}$) and $\mathcal{T}_n^{\OS}$ (for a test of $H_0^{\lex}$), or---when VaR forecasts are identical---based on $\mathcal{T}_{2n}$. In the latter case, $H_0^{=}$ ($H_0^{\lex}$) boils down to $\E[\overline{d}_{2n}]=0$ ($\E[\overline{d}_{2n}]\leq0$) and is rejected if $|\mathcal{T}_{2n}|>\Phi^{-1}(1-\nu/2)$ ($\mathcal{T}_{2n}>\Phi^{-1}(1-\nu)$); see Remark~\ref{rem:degen}.

For distinct VaR forecasts, we consider $\vr_{t,(1)}=\widehat{\vr}_{t,(1)}$ and $\vr_{t,(2)}=\widehat{\vr}_{t,(2)}$ under both $H_0^{=}$ and $H_0^{\lex}$. Since both forecasts are equally misspecified, we have $\E[\overline{\vd}_n]=\vzero$ when we evaluate either $(\VaR, \CoVaR)$ or $(\VaR, \CoVaR, \CoES)$. Under the alternative, we compare the forecasts $\vr_{t,(1)}=\widehat{\vr}_{t,(1)}=(\widehat{r}_{t,(1)}^{\VaR}, \widehat{\vr}_{t,(1)}^{\SR}) $ and $\vr_{t,(2)}=(\widehat{r}_{t,(2)}^{\VaR}, \widehat{\vr}_{t}^{\SR}) =(\widehat{\VaR}_{t,(2)}, \widehat{\CoVaR}_{t}, \widehat{\CoES}_{t}) $. The VaR forecasts are comparable, giving $\E[\overline{d}_{1n}]=0$. Yet, the systemic risk forecasts of $\vr_{t,(1)}$ are inferior to those of $\vr_{t,(2)}$, implying $\E[\overline{d}_{2n}]>0$. In particular, we simulate under the alternative of both $H_0^{=}$ and $H_0^{\lex}$. 

For identical VaR forecasts, we consider $\vr_{t,(1)}=\widehat{\vr}_{t,(1)}=(\widehat{r}_{t,(1)}^{\VaR}, \widehat{\vr}_{t,(1)}^{\SR}) $ and $\vr_{t,(2)}=(\widehat{r}_{t,(1)}^{\VaR}, \widehat{\vr}_{t,(2)}^{\SR}) $ under $H_0^{=}$ and $H_0^{\lex}$. Under the alternative, we compare $\vr_{t,(1)}=(\widehat{r}_{t,(1)}^{\VaR}, \widehat{\vr}_{t,(1)}^{\SR}) $ and $\vr_{t,(2)}=(\widehat{r}_{t,(1)}^{\VaR}, \widehat{\vr}_t^{\SR}) $, such that again $\E[\overline{d}_{2n}]>0$.

\subsection{Simulation results}\label{Simulation Results}

\begin{table}
	\caption{\label{tab:freq} Rejection frequencies (in \%) of $H_0^{=}$ and $H_0^{\lex}$ using Remark~\ref{rem:degen} (column `$\widehat{\VaR}_{t,(1)}$') and Theorem~\ref{thm:DM} (column `$\widehat{\VaR}_{t,(i)}$').}
	
	\centering
		\begin{tabular}{lllccccc}
			\toprule
	$n$	 	& Null				& Forecasts	& \multicolumn{2}{c}{(VaR, CoVaR)}					&&  \multicolumn{2}{c}{(VaR, CoVaR, CoES)} 	\\
												\cline{4-5} 					\cline{7-8}\\[-2ex]
				&     				&						& \multicolumn{2}{c}{$r_{t,(i)}^{\VaR}$} 				 && \multicolumn{2}{c}{$r_{t,(i)}^{\VaR}$} \\
												\cline{4-5} 					\cline{7-8}\\[-2ex]
				&     				&						&	$\widehat{\VaR}_{t,(1)}$ & $\widehat{\VaR}_{t,(i)}$ && $\widehat{\VaR}_{t,(1)}$ & $\widehat{\VaR}_{t,(i)}$  \\
			\midrule
	500 	&$H_0^{=}$		&	$\vr_{t,(1)}^{\SR}=\widehat{\vr}_{t,(1)}^{\SR}$	& \multirow{2}{*}{5.1} & \multirow{2}{*}{4.7} 	&& \multirow{2}{*}{4.0} & \multirow{2}{*}{4.5} 	\\
				&     				& $\vr_{t,(2)}^{\SR}=\widehat{\vr}_{t,(2)}^{\SR}$  \\
				\cline{3-8}\\[-2ex]
				&     		 		& $\vr_{t,(1)}^{\SR}=\widehat{\vr}_{t,(1)}^{\SR}$	& \multirow{2}{*}{83.6} & \multirow{2}{*}{76.0} 	&& \multirow{2}{*}{85.0} & \multirow{2}{*}{76.5} 	\\
				&     		 		& $\vr_{t,(2)}^{\SR}=\widehat{\vr}_{t}^{\SR}$	\\
				\cline{2-8}\\[-2ex]
				&$H_0^{\lex}$	&	$\vr_{t,(1)}^{\SR}=\widehat{\vr}_{t,(1)}^{\SR}$	& \multirow{2}{*}{5.3} & \multirow{2}{*}{4.7} 	&& \multirow{2}{*}{5.1} & \multirow{2}{*}{4.9} 	\\
				&     				& $\vr_{t,(2)}^{\SR}=\widehat{\vr}_{t,(2)}^{\SR}$  \\
				\cline{3-8}\\[-2ex]
				&     				& $\vr_{t,(1)}^{\SR}=\widehat{\vr}_{t,(1)}^{\SR}$	& \multirow{2}{*}{89.2} & \multirow{2}{*}{81.2} 	&& \multirow{2}{*}{90.7} & \multirow{2}{*}{81.8} 	\\
				&     				& $\vr_{t,(2)}^{\SR}=\widehat{\vr}_{t}^{\SR}$	\\
				\midrule

	1000 	&$H_0^{=}$  	&	$\vr_{t,(1)}^{\SR}=\widehat{\vr}_{t,(1)}^{\SR}$	& \multirow{2}{*}{5.2} & \multirow{2}{*}{4.4} 	&& \multirow{2}{*}{4.3} & \multirow{2}{*}{4.6} 	\\
				&     		  	& $\vr_{t,(2)}^{\SR}=\widehat{\vr}_{t,(2)}^{\SR}$  \\
				\cline{3-8}\\[-2ex]
				&     			  & $\vr_{t,(1)}^{\SR}=\widehat{\vr}_{t,(1)}^{\SR}$	& \multirow{2}{*}{95.3} & \multirow{2}{*}{92.8} 	&& \multirow{2}{*}{96.2} & \multirow{2}{*}{94.2} 	\\
				&     			  & $\vr_{t,(2)}^{\SR}=\widehat{\vr}_{t}^{\SR}$	\\

					\cline{2-8}\\[-2ex]
				&$H_0^{\lex}$	&	$\vr_{t,(1)}^{\SR}=\widehat{\vr}_{t,(1)}^{\SR}$	& \multirow{2}{*}{5.0} & \multirow{2}{*}{4.5} 	&& \multirow{2}{*}{4.8} & \multirow{2}{*}{4.5} 	\\
				&     	  		& $\vr_{t,(2)}^{\SR}=\widehat{\vr}_{t,(2)}^{\SR}$  \\
				\cline{3-8}\\[-2ex]
				&     			  & $\vr_{t,(1)}^{\SR}=\widehat{\vr}_{t,(1)}^{\SR}$	& \multirow{2}{*}{96.5} & \multirow{2}{*}{93.8} 	&& \multirow{2}{*}{96.6} & \multirow{2}{*}{94.9} 	\\
				&     			  & $\vr_{t,(2)}^{\SR}=\widehat{\vr}_{t}^{\SR}$	\\	
			\bottomrule
		\end{tabular}
\end{table}

We test $H_0^{=}$ and $H_0^{\lex}$ at significance level $\nu=5\%$.
The columns `$\widehat{\VaR}_{t,(i)}$' (`$\widehat{\VaR}_{t,(1)}$') in Table~\ref{tab:freq} present the results for distinct (identical) VaR forecasts. We draw the following conclusions:

\begin{enumerate}
	\item Size is adequate in all cases, even for $n=500$. This is encouraging, since effective sample sizes in risk forecast comparisons are small, and are smaller still when systemic risk is concerned.
	
	\item Comparing the results for $n=500$ and $n=1000$, we see the expected increase in power. 
	
	\item When comparing systemic risk forecasts, it is slightly easier to distinguish between two different (CoVaR, CoES) forecasts than if CoVaR alone is evaluated. For instance, under $H_0^{=}$ and for $n=500$ there is a statically significant difference between $\vr_{t,(1)}=(\widehat{\VaR}_{t,(1)}, \widehat{\CoVaR}_{t,(1)})$ and $\vr_{t,(2)}=(\widehat{\VaR}_{t,(2)}, \widehat{\CoVaR}_{t}) $ in 76.0\% of all cases, and when CoES is added to the evaluation (so that $\vr_{t,(1)}=(\widehat{\VaR}_{t,(1)}, \widehat{\CoVaR}_{t,(1)}, \widehat{\CoES}_{t,(1)}) $ and $\vr_{t,(2)}=(\widehat{\VaR}_{t,(2)}, \widehat{\CoVaR}_{t}, \widehat{\CoES}_{t}))$ this number is 76.5\%. One reason for this may be the increase in informational content of the forecasts, such that differences are easier to identify.
	
	\item Under both distinct and identical VaR forecasts, departures from $H_0^{\lex}$ are detected more often than from $H_0^{=}$. E.g., for $\vr_{t,(1)}=(\widehat{\VaR}_{t,(1)}, \widehat{\CoVaR}_{t,(1)}) $ and $\vr_{t,(2)}=(\widehat{\VaR}_{t,(2)}, \widehat{\CoVaR}_{t}) $, we reject $H_0^{=}$ with a percentage of 76.0\% for $n=500$. However, $H_0^{\lex}$ is rejected in 81.2\% of all cases for the same forecasts. This is as expected for one-sided tests. 
	
	\item Generally, it is easier to detect differences in the systemic risk forecasts $\vr_{t,(1)}^{\SR}$ and $\vr_{t,(2)}^{\SR}$, when the same VaR forecasts are used compared to when VaR forecasts are merely comparable. E.g., for $\vr_{t,(1)}=(\widehat{\VaR}_{t,(1)}, \widehat{\CoVaR}_{t,(1)})$ and $\vr_{t,(2)}=(\widehat{\VaR}_{t,(1)}, \widehat{\CoVaR}_{t})$, we reject $H_0^{=}$ with probability 83.6\% for $n=500$. 
	Yet, when comparing $\vr_{t,(1)}=(\widehat{\VaR}_{t,(1)}, \widehat{\CoVaR}_{t,(1)}) $ and $\vr_{t,(2)}=(\widehat{\VaR}_{t,(2)}, \widehat{\CoVaR}_{t})$, $H_0^{=}$ is rejected in only 76.0\% of all replications.	
	Even when taking into account the slightly lower size in the latter case, there appears to be a power difference. Intuitively, when VaR forecasts are identical, instead of merely comparable, the test can exclusively focus on differences in the systemic risk component, thus increasing power.
	
\end{enumerate}

Overall, the simulations show that our tests work quite well in that they keep size and have good power in detecting significant differences in forecast ability.

\section{The perils of using non-strict identification functions for backtesting}
\label{app:Remark 4.5}

For $t=1,\ldots,n$, consider $(X_t, Y_t)\overset{\text{i.i.d.}}{\sim}N(\vzero, \mOmega)$ with 
$\mOmega=\left(\begin{smallmatrix}
1 & 0.5 \\ 0.5 & 2
\end{smallmatrix}\right)$. 
Suppose the goal is to forecast $\VaR_{\b}(X_t)$ and $\CoVaR_{\a|\b}(Y_t|X_t)$ for $\a=\b=0.95$. Then, the correctly specified forecasts (conditional or unconditional---due to the i.i.d.\ nature this is the same) are $\widehat{\VaR}_{\b,t}=\VaR_{\b}(X_t)\approx1.64$ and $\widehat{\CoVaR}_{\a|\b,t}=\CoVaR_{\a|\b}(Y_t|X_t)\approx3.23$. We also consider the misspecified forecasts $\widehat{\VaR}_{\b^\prime,t}=\VaR_{\b^\prime}(X_t)\approx2.33$ and $\widehat{\CoVaR}_{\a^\prime|\b^\prime,t}=\CoVaR_{\a^\prime|\b^\prime}(Y_t|X_t)\approx2.23$, where $\a^\prime=0.75$ and $\b^\prime=0.99$. Note that in this case $(1-\a')(1-\b') = (1-\a)(1-\b)$, such that the (non-strict) identification function of \citet{Banulescu-RaduETAL2019} in \eqref{eq:Banulescu-Radu CoVaR id} vanishes in expectation, yet our strict one $\mV^{(\VaR, \CoVaR)}$ from \eqref{eq:V_CoVaR} does not.
We test the null hypothesis of correct unconditional calibration, which---in terms of our strict identification function---reads as
\[
	H_0 \colon \E\Big[\mV^{(\VaR, \CoVaR)}\big((\widehat{\VaR}_{t}, \widehat{\CoVaR}_{t}),(X_t,Y_t)\big)\Big] =\vzero \qquad\text{for all }t=1,2,\ldots 
\]
for $(\widehat{\VaR}_{t}, \widehat{\CoVaR}_{t})\in\big\{(\widehat{\VaR}_{\b,t}, \widehat{\CoVaR}_{\a|\b,t}), (\widehat{\VaR}_{\b^\prime,t}, \widehat{\CoVaR}_{\a^\prime|\b^\prime,t}) \big\}$. 
Following \citet{NZ17}, we test $H_0$ using a standard Wald-test based on 
\begin{align*}
	\overline{\mV}^{(\VaR, \CoVaR)} &= \frac{1}{n}\sum_{t=1}^{n}\mV^{(\VaR, \CoVaR)}\big((\widehat{\VaR}_{t}, \widehat{\CoVaR}_{t}),(X_t,Y_t)\big)\qquad \text{and}\\
	\overline{V}&=\quad\frac{1}{n}\sum_{t=1}^{n}V\big((\widehat{\VaR}_{t}, \widehat{\CoVaR}_{t}),(X_t,Y_t)\big),
\end{align*}
respectively.

Table~\ref{tab:Rem45} displays the rejection frequencies---calculated from 10\,000 replications---for the tests (at a 5\%-level) for $n\in\{500, 1000\}$. 
For our test statistic $\overline{\mV}^{(\VaR, \CoVaR)}$ size is close to the nominal level, and the misspecified forecasts are identified as such almost with certainty. However, with \citeauthor{Banulescu-RaduETAL2019}'s \citeyearpar{Banulescu-RaduETAL2019} test statistic $\overline{V}$, the null is rejected with about the same frequency for the correctly specified \textit{and} the misspecified forecasts. This is because the null that is actually tested with $\overline{V}$ is
\[
	H_0^{*} \colon \E\Big[V\big((\widehat{\VaR}_{t}, \widehat{\CoVaR}_{t}),(X_t,Y_t)\big)\Big] =0 \qquad\text{for all }t=1,2,\ldots 
\]
where $V$ is given in \eqref{eq:Banulescu-Radu CoVaR id}. Obviously, $H_0^{*}$ does not amount to a null of correct calibration, because it is also satisfied for specific misspecified forecasts, such as $(\widehat{\VaR}_{\b^\prime,t}, \widehat{\CoVaR}_{\a^\prime|\b^\prime,t})$. This highlights that traditional backtests should not be carried out using non-strict identification functions, because $H_0^{*}$ is too broad in that it accommodates possibly misspecified forecasts.

\begin{table}[t!]
	\caption{\label{tab:Rem45}Rejection frequencies (in \%) of $H_0$ based on $\overline{\mV}^{(\VaR, \CoVaR)}$ and $\overline{V}$ for $n\in\{500, 1000\}$. Results are displayed for correctly specified forecasts $(\widehat{\VaR}_{\b,t}, \widehat{\CoVaR}_{\a|\b,t})$ and misspecified forecasts $(\widehat{\VaR}_{\b^\prime,t}, \widehat{\CoVaR}_{\a^\prime|\b^\prime,t})$.}
	\centering
		\begin{tabular}{lccccc}
			\toprule
	$n$	 	&  \multicolumn{2}{c}{$\overline{\mV}^{(\VaR, \CoVaR)}$}					&&  \multicolumn{2}{c}{$\overline{V}$} 	\\
												\cline{2-3} 					\cline{5-6}\\[-2ex]
				&	 $\begin{pmatrix}\widehat{\VaR}_{\b,t}\\ \widehat{\CoVaR}_{\a|\b,t}\end{pmatrix}$ & $\begin{pmatrix}\widehat{\VaR}_{\b^\prime,t}\\ \widehat{\CoVaR}_{\a^\prime|\b^\prime,t}\end{pmatrix}$				 &&  $\begin{pmatrix}\widehat{\VaR}_{\b,t}\\ \widehat{\CoVaR}_{\a|\b,t}\end{pmatrix}$ & $\begin{pmatrix}\widehat{\VaR}_{\b^\prime,t}\\ \widehat{\CoVaR}_{\a^\prime|\b^\prime,t}\end{pmatrix}$	 \\
			\midrule
	500 	& 6.8 &   99.9 & &    28.9  &  28.1 \\
				\midrule          
	1000 	& 6.4 &   100  & &    8.1   &  8.3 \\
			\bottomrule
		\end{tabular}
\end{table}


Observe that particularly for $n=500$, size for the test based on $\overline{V}$ is very far from the nominal level of $5\%$. This may be explained as follows: From \eqref{eq:Banulescu-Radu CoVaR id} we see that the identification function is constant except when the $X$ component exceeds the VaR forecast \textit{and} the $Y$ component exceeds the CoVaR forecast. Under the null this only occurs with probability $(1-\alpha)(1-\beta)=(1-0.95)(1-0.95)=0.25\%$, such that the effective sample size is reduced from $n=500$ to $n\cdot0.25\%=1.25$. Thus, size distortions can be expected. Note that this issue is somewhat alleviated for our two-dimensional score $\mV^{(\VaR,\CoVaR)}$, where the effective sample size---at least in the first component of the score---is $n(1-\beta)=25$, leading to much better size for $n=500$.

\begin{figure}[b!]
	\centering
		\includegraphics[width=\textwidth]{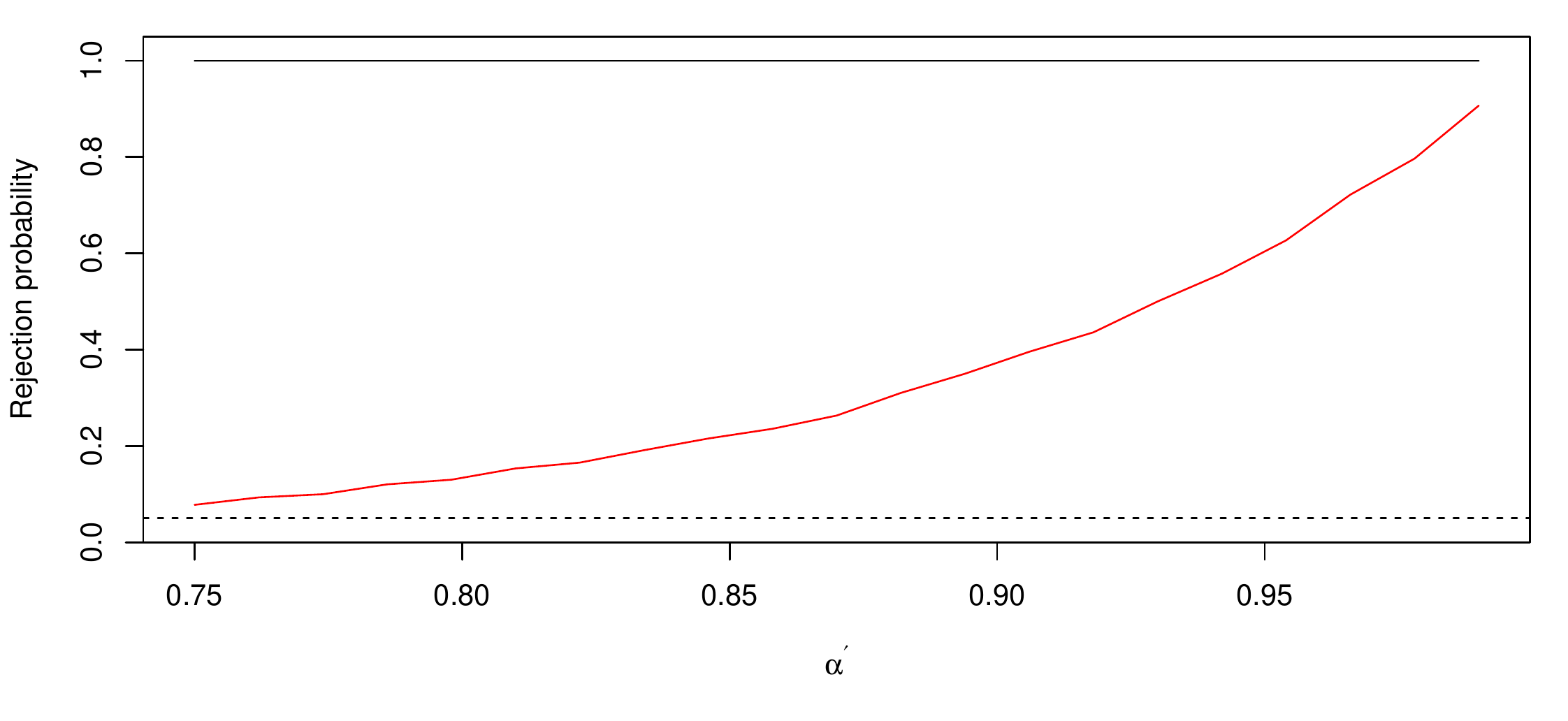}
	\caption{Rejection probabilities of $H_0$ as a function of $\alpha^\prime\in(0.75,0.99)$ for our $\overline{\mV}^{(\VaR, \CoVaR)}$ (black) and \citeauthor{Banulescu-RaduETAL2019}'s \citeyearpar{Banulescu-RaduETAL2019} $\overline{V}$ (red) based on $n=1000$. Nominal level of $5\%$ indicated by the dashed horizontal line.}
	\label{fig:Figure_comp}
\end{figure}

We emphasize that the non-strictness of \citeauthor{Banulescu-RaduETAL2019}'s \citeyearpar{Banulescu-RaduETAL2019} $V$ at \eqref{eq:Banulescu-Radu CoVaR id} is not only problematic when $(1-\alpha)(1-\beta)=(1-\alpha^\prime)(1-\beta^\prime)$ (which is already the case for uncountably many configurations), but in all cases where VaR is overpredicted (underpredicted) and CoVaR is underpredicted (overpredicted). This is because the different biases cancel each other out such that $V$ remains close to zero on average, leading to a loss of power. To illustrate this, we fix the probability levels of the misspecified forecasts at $\beta^\prime=0.99$ (as before) and let $\alpha^\prime$ vary between $(0.75, 0.95)$. Then, $(1-\alpha)(1-\beta)\neq(1-\alpha^\prime)(1-\beta^\prime)$, but VaR is still overpredicted and CoVaR is underpredicted. We again test $H_0$ using $\overline{\mV}^{(\VaR, \CoVaR)}$ and $\overline{V}$ for $(\widehat{\VaR}_{t}, \widehat{\CoVaR}_{t})\in\big\{(\widehat{\VaR}_{\b,t}, \widehat{\CoVaR}_{\a|\b,t}), (\widehat{\VaR}_{\b^\prime,t}, \widehat{\CoVaR}_{\a^\prime|\b^\prime,t}) \big\}$. These forecasts are misspecified for all $\alpha^\prime\in(0.75, 0.95)$, such that we expect high rejection frequencies.

Figure~\ref{fig:Figure_comp} plots the rejection probabilities as a function of $\alpha^\prime$ for $n=1000$. We see that $\overline{\mV}^{(\VaR, \CoVaR)}$ always identifies the forecasts as misspecified. As could be expected from Table~\ref{tab:Rem45}, for $\alpha^\prime=0.75$ the test based on $\overline{V}$ has only close to trivial power of $8.3\%$. However, the power increases only slowly for increasing $\alpha^\prime$. Even for $\alpha^\prime\in(0.95, 0.99)$, where both VaR and CoVaR are overpredicted, our proposal based on the strict identification function has markedly higher power than the test of \citet{Banulescu-RaduETAL2019}.

\end{appendix}


\end{bibunit}


\begin{thebibliography}{39}
\providecommand{\natexlab}[1]{#1}

\bibitem[{Acharya \emph{et~al.}(2017)Acharya, Pedersen, Philippon and
  Richardson}]{Aea17}
Acharya VV, Pedersen LH, Philippon T, Richardson M. 2017. Measuring systemic
  risk. \emph{The Review of Financial Studies} \textbf{30}: 2--47.

\bibitem[{Adrian \emph{et~al.}(2019)Adrian, Boyarchenko and Giannone}]{ABG19}
Adrian T, Boyarchenko N, Giannone D. 2019. Vulnerable growth. \emph{The
  American Economic Review} \textbf{109}: 1263--1289.

\bibitem[{Adrian and Brunnermeier(2016)}]{AB16}
Adrian T, Brunnermeier MK. 2016. {CoVaR}. \emph{The American Economic Review}
  \textbf{106}: 1705--1741.

\bibitem[{Adrian \emph{et~al.}(2021)Adrian, Grinberg, Liang, Malik and
  Yu}]{Aea21}
Adrian T, Grinberg F, Liang N, Malik S, Yu J. 2021. The term structure of
  {G}rowth-at-{R}isk. \emph{\textup{Forthcoming in} American Economic Journal:
  Macroeconomics}
  \url{https://www.aeaweb.org/articles/pdf/doi/10.1257/mac.20180428}.

\bibitem[{Artzner \emph{et~al.}(1999)Artzner, Delbaen, Eber and Heath}]{Aea99}
Artzner P, Delbaen F, Eber JM, Heath D. 1999. Coherent measures of risk.
  \emph{Mathematical Finance} \textbf{9}: 203--228.

\bibitem[{{Bank for International Settlements}(2019)}]{BCBSBF19}
{Bank for International Settlements}. 2019. \emph{Basel Framework}. Basel,
  \url{http://www.bis.org/basel\_framework/index.htm?export=pdf}.

\bibitem[{Banulescu-Radu \emph{et~al.}(2021)Banulescu-Radu, Hurlin, Leymarie
  and Scaillet}]{Banulescu-RaduETAL2019}
Banulescu-Radu D, Hurlin C, Leymarie J, Scaillet O. 2021. Backtesting marginal
  expected shortfall and related systemic risk measures. \emph{Management
  Science} \textbf{67}: 5730--5754.

\bibitem[{Benoit \emph{et~al.}(2017)Benoit, Colliard, Hurlin and
  P{\'e}rignon}]{Bea17}
Benoit S, Colliard JE, Hurlin C, P{\'e}rignon C. 2017. Where the risks lie: A
  survey on systemic risk. \emph{Review of Finance} \textbf{21}: 109--152.

\bibitem[{Bernardi and Catania(2019)}]{BC19}
Bernardi M, Catania L. 2019. Switching generalized autoregressive score copula
  models with application to systemic risk. \emph{Journal of Applied
  Econometrics} \textbf{34}: 43--65.

\bibitem[{Brownlees \emph{et~al.}(2011)Brownlees, Engle and Kelly}]{BEK11}
Brownlees C, Engle R, Kelly B. 2011. A practical guide to volatility
  forecasting through calm and storm. \emph{Journal of Risk} \textbf{14}:
  3--22.

\bibitem[{Brownlees and Engle(2017)}]{BE17}
Brownlees C, Engle RF. 2017. {SRISK}: A conditional capital shortfall measure
  of systemic risk. \emph{The Review of Financial Studies} \textbf{30}: 48--79.

\bibitem[{Brownlees and Souza(2021)}]{BS21}
Brownlees C, Souza ABM. 2021. Backtesting global {G}rowth-at-{R}isk.
  \emph{Journal of Monetary Economics} \textbf{118}: 312--330.

\bibitem[{Brunnermeier \emph{et~al.}(2020)Brunnermeier, Rother and
  Schnabel}]{BRS20}
Brunnermeier M, Rother S, Schnabel I. 2020. Asset price bubbles and systemic
  risk. \emph{The Review of Financial Studies} \textbf{33}: 4272--4317.

\bibitem[{Chen \emph{et~al.}(2013)Chen, Iyengar and
  Moallemi}]{ChenIyengarMoallemi2013}
Chen C, Iyengar G, Moallemi CC. 2013. {An Axiomatic Approach to Systemic Risk}.
  \emph{Management Science} \textbf{59}: 1373--1388.

\bibitem[{Creal \emph{et~al.}(2013)Creal, Koopman and Lucas}]{CKL13}
Creal D, Koopman SJ, Lucas A. 2013. Generalized autoregressive score models
  with applications. \emph{Journal of Applied Econometrics} \textbf{28}:
  777--795.

\bibitem[{Diebold and Mariano(1995)}]{DM95}
Diebold FX, Mariano RS. 1995. Comparing predictive accuracy. \emph{Journal of
  Business \& Economic Statistics} \textbf{13}: 253--263.

\bibitem[{Dimitriadis \emph{et~al.}(2020{\natexlab{a}})Dimitriadis, Fissler and
  Ziegel}]{DFZ2020}
Dimitriadis T, Fissler T, Ziegel JF. 2020{\natexlab{a}}. {The Efficiency Gap}.
  \emph{Preprint.} \url{https://arxiv.org/abs/2010.14146}.

\bibitem[{Dimitriadis \emph{et~al.}(2020{\natexlab{b}})Dimitriadis, Liu and
  Schnaitmann}]{DLS20}
Dimitriadis T, Liu X, Schnaitmann J. 2020{\natexlab{b}}. Encompassing tests for
  value at risk and expected shortfall multi-step forecasts based on inference
  on the boundary. \emph{Preprint.} \url{https://arxiv.org/abs/2009.07341}.

\bibitem[{Eckernkemper(2018)}]{Eck18}
Eckernkemper T. 2018. Modeling systemic risk: Time-varying tail dependence when
  forecasting marginal expected shortfall. \emph{Journal of Financial
  Econometrics} \textbf{16}: 63--117.

\bibitem[{Ehm \emph{et~al.}(2016)Ehm, Gneiting, Jordan and
  Kr{\"u}ger}]{EhmETAL2016}
Ehm W, Gneiting T, Jordan A, Kr{\"u}ger F. 2016. {Of quantiles and expectiles:
  consistent scoring functions, Choquet representations and forecast rankings}.
  \emph{Journal of the Royal Statistical Society: Series B (Statistical
  Methodology)} \textbf{78}: 505--562.

\bibitem[{Ehrgott(2005)}]{Ehrgott2005}
Ehrgott M. 2005. \emph{Multicriteria Optimization}. Berlin, Heidelberg:
  Springer.

\bibitem[{Emmer \emph{et~al.}(2015)Emmer, Kratz and Tasche}]{EKT15}
Emmer S, Kratz M, Tasche D. 2015. What is the best risk measure in practice? a
  comparison of standard measures. \emph{Journal of Risk} \textbf{18}: 31--60.

\bibitem[{Feinstein \emph{et~al.}(2017)Feinstein, Rudloff and
  Weber}]{FeinsteinRudloffWeber2017}
Feinstein Z, Rudloff B, Weber S. 2017. Measures of systemic risk. \emph{SIAM
  Journal on Financial Mathematics} \textbf{8}: 672--708.

\bibitem[{Fissler \emph{et~al.}(2021)Fissler, Frongillo, Hlavinov\'a and
  Rudloff}]{FFHR2021}
Fissler T, Frongillo R, Hlavinov\'a J, Rudloff B. 2021. Forecast evaluation of
  quantiles, prediction intervals, and other set-valued functionals.
  \emph{Electronic Journal of Statistics} \textbf{15}: 1034--1084.

\bibitem[{Fissler and Ziegel(2016)}]{FZ16a}
Fissler T, Ziegel JF. 2016. Higher order elicitability and {O}sband's
  principle. \emph{The Annals of Statistics} \textbf{44}: 1680--1707.

\bibitem[{Fissler \emph{et~al.}(2016)Fissler, Ziegel and Gneiting}]{FZG16}
Fissler T, Ziegel JF, Gneiting T. 2016. Expected shortfall is jointly
  elicitable with value-at-risk: Implications for backtesting. \emph{Risk
  Magazine} : 58--61.

\bibitem[{Giacomini and White(2006)}]{GW06}
Giacomini R, White H. 2006. Tests of conditional predictive ability.
  \emph{Econometrica} \textbf{74}: 1545--1578.

\bibitem[{Giesecke and Kim(2011)}]{GK11}
Giesecke K, Kim B. 2011. Systemic risk: What defaults are telling us.
  \emph{Management Science} \textbf{57}: 1387--1405.

\bibitem[{Giglio \emph{et~al.}(2016)Giglio, Kelly and Pruitt}]{GKP16}
Giglio S, Kelly B, Pruitt S. 2016. Systemic risk and the macroeconomy: An
  empirical evaluation. \emph{Journal of Financial Economics} \textbf{119}:
  457--471.

\bibitem[{Girardi and Tolga~Erg{\"u}n(2013)}]{GT13}
Girardi G, Tolga~Erg{\"u}n A. 2013. Systemic risk measurement: Multivariate
  {GARCH} estimation of {CoVaR}. \emph{Journal of Banking \& Finance}
  \textbf{37}: 3169--3180.

\bibitem[{Glosten \emph{et~al.}(1993)Glosten, Jagannathan and Runkle}]{GJR93}
Glosten LR, Jagannathan R, Runkle DE. 1993. On the relation between the
  expected value and the volatility of the nominal excess return on stocks.
  \emph{The Journal of Finance} \textbf{48}: 1779--1801.

\bibitem[{Gneiting(2011)}]{Gne11}
Gneiting T. 2011. Making and evaluating point forecasts. \emph{Journal of the
  American Statistical Association} \textbf{106}: 746--762.

\bibitem[{Hoga(2021)}]{Hog20a+}
Hoga Y. 2021. Modeling time-varying tail dependence, with application to
  systemic risk forecasting. \emph{\textup{Forthcoming in} Journal of Financial
  Econometrics} : 1--31.

\bibitem[{Mainik and Schaanning(2014)}]{MS14}
Mainik G, Schaanning E. 2014. On dependence consistency of {CoVaR} and some
  other systemic risk measures. \emph{Statistics \& Risk Modeling} \textbf{31}:
  49--77.

\bibitem[{Mas-Colell \emph{et~al.}(1995)Mas-Colell, Whinston and
  Green}]{MWG1995}
Mas-Colell A, Whinston MD, Green JR. 1995. \emph{Microeconomic Theory}. Oxford
  University Press.

\bibitem[{Nolde and Zhang(2020)}]{NZ19+}
Nolde N, Zhang J. 2020. Conditional extremes in asymmetric financial markets.
  \emph{Journal of Business \& Economic Statistics} \textbf{38}: 201--213.

\bibitem[{Nolde and Ziegel(2017)}]{NZ17}
Nolde N, Ziegel JF. 2017. Elicitability and backtesting: Perspectives for
  banking regulation. \emph{The Annals of Applied Statistics} \textbf{11}:
  1833--1874.

\bibitem[{Steinwart \emph{et~al.}(2014)Steinwart, Pasin, Williamson and
  Zhang}]{SteinwartPasinETAL2014}
Steinwart I, Pasin C, Williamson R, Zhang S. 2014. {Elicitation and
  Identification of Properties}. \emph{JMLR Workshop Conf. Proc.} \textbf{35}:
  1--45.

\bibitem[{White(2001)}]{Whi01}
White H. 2001. \emph{Asymptotic Theory for Econometricians}. San Diego:
  Academic Press, {F}irst edn.

\end{thebibliography}

\begin{thebibliography}{33}
\providecommand{\natexlab}[1]{#1}

\bibitem[{Banulescu-Radu \emph{et~al.}(2021)Banulescu-Radu, Hurlin, Leymarie
  and Scaillet}]{Banulescu-RaduETAL2019}
Banulescu-Radu D, Hurlin C, Leymarie J, Scaillet O. 2021. Backtesting marginal
  expected shortfall and related systemic risk measures. \emph{Management
  Science} \textbf{67}: 5730--5754.

\bibitem[{Bellini and Bignozzi(2015)}]{BelliniBignozzi2015}
Bellini F, Bignozzi V. 2015. On elicitable risk measures. \emph{Quantitative
  Finance} \textbf{15}: 725--733.

\bibitem[{Bernardi and Catania(2019)}]{BC19}
Bernardi M, Catania L. 2019. Switching generalized autoregressive score copula
  models with application to systemic risk. \emph{Journal of Applied
  Econometrics} \textbf{34}: 43--65.

\bibitem[{Creal \emph{et~al.}(2013)Creal, Koopman and Lucas}]{CKL13}
Creal D, Koopman SJ, Lucas A. 2013. Generalized autoregressive score models
  with applications. \emph{Journal of Applied Econometrics} \textbf{28}:
  777--795.

\bibitem[{Davidson(1994)}]{Dav94}
Davidson J. 1994. \emph{Stochastic Limit Theory}. Oxford: Oxford University
  Press.

\bibitem[{De~Jong and Davidson(2000)}]{DD00}
De~Jong RM, Davidson J. 2000. Consistency of kernel estimators of
  heteroscedastic and autocorrelated covariance matrices. \emph{Econometrica}
  \textbf{68}: 407--424.

\bibitem[{De~Lira~Salvatierra and Patton(2015)}]{DP15}
De~Lira~Salvatierra I, Patton AJ. 2015. Dynamic copula models and high
  frequency data. \emph{Journal of Empirical Finance} \textbf{30}: 120--135.

\bibitem[{Diebold and Mariano(1995)}]{DM95}
Diebold FX, Mariano RS. 1995. Comparing predictive accuracy. \emph{Journal of
  Business \& Economic Statistics} \textbf{13}: 253--263.

\bibitem[{Dimitriadis \emph{et~al.}(2020)Dimitriadis, Fissler and
  Ziegel}]{DFZ2020}
Dimitriadis T, Fissler T, Ziegel JF. 2020. {The Efficiency Gap}.
  \emph{Preprint.} \url{https://arxiv.org/abs/2010.14146}.

\bibitem[{Fissler(2017)}]{Fissler2017}
Fissler T. 2017. \emph{{On Higher Order Elicitability and Some Limit Theorems
  on the Poisson and Wiener Space}}. Ph.D. thesis, University of Bern.

\bibitem[{Fissler \emph{et~al.}(2021)Fissler, Frongillo, Hlavinov\'a and
  Rudloff}]{FFHR2021}
Fissler T, Frongillo R, Hlavinov\'a J, Rudloff B. 2021. Forecast evaluation of
  quantiles, prediction intervals, and other set-valued functionals.
  \emph{Electronic Journal of Statistics} \textbf{15}: 1034--1084.

\bibitem[{Fissler and Ziegel(2016)}]{FZ16a}
Fissler T, Ziegel JF. 2016. Higher order elicitability and {O}sband's
  principle. \emph{The Annals of Statistics} \textbf{44}: 1680--1707.

\bibitem[{Fissler and Ziegel(2019)}]{FisslerZiegel2019}
Fissler T, Ziegel JF. 2019. Order-sensitivity and equivariance of scoring
  functions. \emph{Electronic Journal of Statistics} \textbf{13}: 1166--1211.

\bibitem[{Francq and Zako\"{i}an(2004)}]{FZ04}
Francq C, Zako\"{i}an JM. 2004. Maximum likelihood estimation of pure {GARCH}
  and {ARMA--GARCH} processes. \emph{Bernoulli} \textbf{10}: 605--637.

\bibitem[{Giacomini and White(2006)}]{GW06}
Giacomini R, White H. 2006. Tests of conditional predictive ability.
  \emph{Econometrica} \textbf{74}: 1545--1578.

\bibitem[{Gneiting(2011{\natexlab{a}})}]{Gne11}
Gneiting T. 2011{\natexlab{a}}. Making and evaluating point forecasts.
  \emph{Journal of the American Statistical Association} \textbf{106}:
  746--762.

\bibitem[{Gneiting(2011{\natexlab{b}})}]{Gne11a}
Gneiting T. 2011{\natexlab{b}}. Quantiles as optimal point forecasts.
  \emph{International Journal of Forecasting} \textbf{27}: 197--207.

\bibitem[{Gneiting and Ranjan(2013)}]{GneitingRanjan2013}
Gneiting T, Ranjan R. 2013. Combining predictive distributions.
  \emph{Electronic Journal of Statistics} \textbf{7}: 1747--1782.

\bibitem[{Hoga(2021)}]{Hog20a+}
Hoga Y. 2021. Modeling time-varying tail dependence, with application to
  systemic risk forecasting. \emph{\textup{Forthcoming in} Journal of Financial
  Econometrics} : 1--31.

\bibitem[{Holzmann and Eulert(2014)}]{HolzmannEulert2014}
Holzmann H, Eulert M. 2014. The role of the information set for forecasting --
  with applications to risk management. \emph{The Annals of Applied Statistics}
  \textbf{8}: 79--83.

\bibitem[{Holzmann and Klar(2017)}]{HolzmannKlar2017}
Holzmann H, Klar B. 2017. Focusing on regions of interest in forecast
  evaluation. \emph{The Annals of Applied Statistics} \textbf{11}: 2404--2431.

\bibitem[{Lambert(2019)}]{Lambert2019}
Lambert N. 2019. {Elicitation and Evaluation of Statistical Forecasts}.
  \emph{Preprint.} \url{http://web.stanford.edu/~nlambert/papers/elicitation.pdf}.

\bibitem[{Lambert \emph{et~al.}(2008)Lambert, Pennock and
  Shoham}]{LambertETAL2008}
Lambert N, Pennock DM, Shoham Y. 2008. Eliciting properties of probability
  distributions. In \emph{{Proceedings of the 9th ACM Conference on Electronic
  Commerce}}, Chicago, Il, USA: ACM, pages 129--138.

\bibitem[{McNeil \emph{et~al.}(2015)McNeil, Frey and Embrechts}]{MFE15}
McNeil AJ, Frey R, Embrechts P. 2015. \emph{Quantitative Risk Management:
  Concepts, Techniques and Tools}. Princeton: Princeton University Press,
  revised edn.

\bibitem[{Nolde and Ziegel(2017)}]{NZ17}
Nolde N, Ziegel JF. 2017. Elicitability and backtesting: Perspectives for
  banking regulation. \emph{The Annals of Applied Statistics} \textbf{11}:
  1833--1874.

\bibitem[{Osband(1985)}]{Osb85}
Osband KH. 1985. \emph{Providing Incentives for Better Cost Forecasting}.
  {Ph.D.}~thesis, University of California, Berkeley.

\bibitem[{Patton(2020)}]{Patton2020}
Patton AJ. 2020. Comparing possibly misspecified forecasts. \emph{Journal of
  Business \& Economic Statistics} \textbf{38}: 796--809.

\bibitem[{Patton \emph{et~al.}(2019)Patton, Ziegel and Chen}]{PZC19}
Patton AJ, Ziegel JF, Chen R. 2019. Dynamic semiparametric models for expected
  shortfall (and value-at-risk). \emph{Journal of Econometrics} \textbf{211}:
  388--413.

\bibitem[{Pohle(2020)}]{Pohle2020}
Pohle MO. 2020. The {M}urphy decomposition and the calibration-resolution
  principle: A new perspective on forecast evaluation. \emph{Preprint.} \url{https://arxiv.org/abs/2005.01835}.

\bibitem[{Steinwart \emph{et~al.}(2014)Steinwart, Pasin, Williamson and
  Zhang}]{SteinwartPasinETAL2014}
Steinwart I, Pasin C, Williamson R, Zhang S. 2014. {Elicitation and
  Identification of Properties}. \emph{JMLR Workshop Conf. Proc.} \textbf{35}:
  1--45.

\bibitem[{Taylor(2019)}]{Tay17}
Taylor JW. 2019. Forecasting value at risk and expected shortfall using a
  semiparametric approach based on the asymmetric {L}aplace distribution.
  \emph{Journal of Business \& Economic Statistics} \textbf{37}: 121--133.

\bibitem[{White(2001)}]{Whi01}
White H. 2001. \emph{Asymptotic Theory for Econometricians}. San Diego:
  Academic Press, {F}irst edn.

\bibitem[{Zhu and Timmermann(2020)}]{ZhuTimmermann2020}
Zhu Y, Timmermann A. 2020. Can two forecasts have the same conditional expected
  accuracy? \emph{Preprint.} \url{https://arxiv.org/abs/2006.03238}.

\end{thebibliography}
\end{document}